\documentclass[10pt,journal]{IEEEtran}
\usepackage{cite}
\usepackage{amsmath,amssymb,amsfonts}
\usepackage{amsthm}
\usepackage{graphicx}
\usepackage[justification=centering]{caption}
\usepackage{textcomp}
\usepackage{graphicx}   
\usepackage{mathtools} % Bonus

\newtheorem{lemma}{Lemma}
\usepackage{amsmath}
\usepackage{algpseudocode}
\newtheorem{theorem}{Theorem}
\usepackage{algorithm}  
\usepackage{setspace}
\usepackage{nomencl}
\renewcommand{\figurename}{Fig.}
\usepackage{color}
\usepackage{bm}
\usepackage{enumerate}
\usepackage{titlesec}
\usepackage{algpseudocode}
\makenomenclature
\usepackage{multirow}
\usepackage{subfigure}
\usepackage{bm} 
\usepackage{makecell}

\usepackage[utf8]{inputenc}
\usepackage{multirow}
\usepackage[T1]{fontenc}
\usepackage{stfloats}
\usepackage{booktabs}
\usepackage{mathrsfs}
\usepackage{xcolor}

\usepackage{times}
\usepackage{textcomp}

\usepackage[justification=justified]{caption}
\allowdisplaybreaks[4]

\ifCLASSINFOpdf
\else
\fi
\begin{document}

\title{Enhancing Federated Learning with spectrum allocation optimization and device selection
\thanks{T. Zhang, K. Lam, and J. Zhao are with the Strategic Centre for Research in Privacy-Preserving Technologies and Systems, and the School of Computer Science and Engineering at Nanyang Technological University, Singapore. (Emails: tinghao001@e.ntu.edu.sg; kwokyan.lam@ntu.edu.sg; junzhao@ntu.edu.sg). 

F. Li is with the Strategic Centre for Research in Privacy-Preserving Technologies and Systems at Nanyang Technological University, Singapore. (Email: fengli2002@yeah.net).

Huimei Han is with College of Information Engineering, Zhejiang University of Technology, Hangzhou,
Zhejiang, 310032, P.R. China (Email: hmhan1215@zjut.edu.cn).

Norziana Jamil is with College of Computing \& Informatics, University Tenaga Nasional, Putrajaya Campus,
Selangor, Malaysia (Email: Norziana@uniten.edu.my).

Corresponding author: Jun Zhao

This research/project is supported by the National Research Foundation,
Singapore under its Strategic Capability Research Centres Funding Initiative.
Any opinions, findings and conclusions or recommendations expressed in this
material are those of the author(s) and do not reflect the views of National
Research Foundation, Singapore.
}
}

\author{\IEEEauthorblockN{Tinghao Zhang, Kwok-Yan Lam, \emph{Senior Member}, IEEE, Jun Zhao, \emph{Member}, IEEE, Feng Li, \emph{Member}, IEEE, Huimei Han, Norziana Jamil}}%\vspace{-1em} 

\markboth{IEEE/ACM Transactions on Networking}%
{Shell \MakeLowercase{\textit{et al.}}: Bare Demo of IEEEtran.cls for Computer Society Journals}

\IEEEtitleabstractindextext{%
\begin{abstract}

Machine learning (ML) is a widely accepted means for supporting customized services for mobile devices and applications. Federated Learning (FL), which is a promising approach to implement machine learning while addressing data privacy concerns, typically involves a large number of wireless mobile devices to collect model training data. Under such circumstances, FL is expected to meet stringent training latency requirements in the face of limited resources such as demand for wireless bandwidth, power consumption, and computation constraints of participating devices. Due to practical considerations, FL selects a portion of devices to participate in the model training process at each iteration. Therefore, the tasks of efficient resource management and device selection will have a significant impact on the practical uses of FL. In this paper, we propose a spectrum allocation optimization mechanism for enhancing FL over a wireless mobile network. Specifically, the proposed spectrum allocation optimization mechanism minimizes the time delay of FL while considering the energy consumption of individual participating devices; thus ensuring that all the participating devices have sufficient resources to train their local models. In this connection, to ensure fast convergence of FL, a robust device selection is also proposed to help FL reach convergence swiftly, especially when the local datasets of the devices are not independent and identically distributed (non-iid). Experimental results show that (1) the proposed spectrum allocation optimization method optimizes time delay while satisfying the individual energy constraints; (2) the proposed device selection method enables FL to achieve the fastest convergence on non-iid datasets.

\end{abstract}

\begin{IEEEkeywords}
Federated Learning, Spectrum Allocation Optimization, Device Selection, Wireless Mobile Networks
\end{IEEEkeywords}}

\maketitle

\IEEEdisplaynontitleabstractindextext

% \fi
% \IEEEpeerreviewmaketitle

% \IEEEraisesectionheading{\section{Introduction}\label{sec:introduction}}
\section{Introduction}~\label{sec:introduction}
Machine learning (ML) has been widely used to help automate decision-making and improve business efficiency. However, traditional ML relies on centralized training approaches, hence leading to data privacy concerns~\cite{Drainakis20,Liu20}. As mobile application services have achieved remarkable successes, there has been an emerging trend of employing Federated Learning (FL) over wireless mobile network~\cite{Helin20,Li21,Zhao21}. As a large-scale data-intensive distributed system, FL involves a large number of wireless mobile devices to train a shared model given a training time budget and a resource budget.

Being a decentralized form of ML, FL performs ML without aggregating user data to a centralized server, thus making it easier to comply with the relevant data privacy regulations~\cite{Viraaji21}. At each global iteration in FL, the wireless devices first download initialized parameters of a global model from a remote server. Then, the devices train local models using their local available data in parallel. After several local iterations, the devices send their parameters to the server, and the global model is updated using the uploaded local parameters. The learning process is repeated until the global model achieves a preset accuracy. Fig.~\ref{FL_Framework} depicts the mechanism of a typical FL paradigm. 
\begin{figure}[t]
  \centering
  \includegraphics[width=8.5cm]{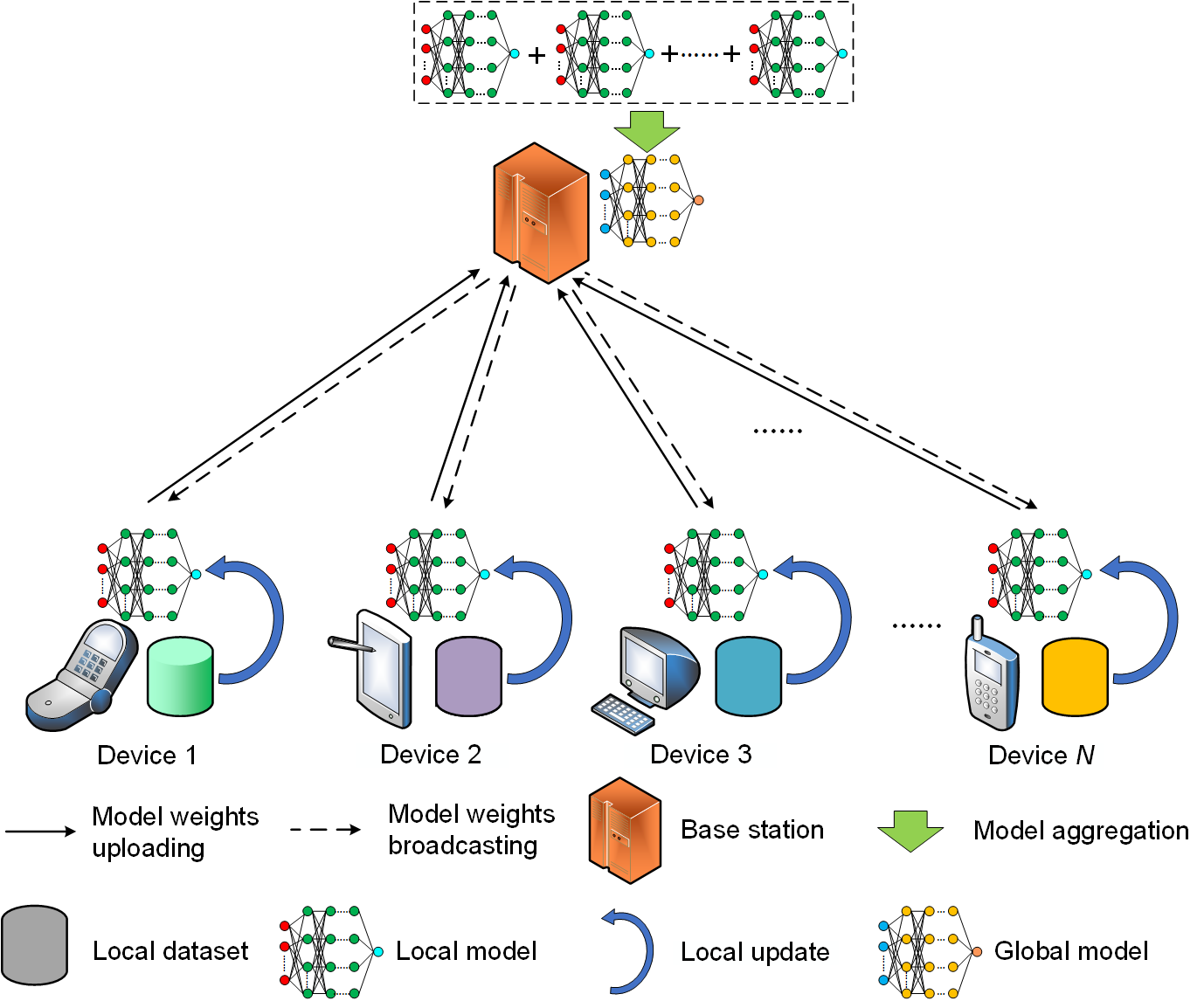}
\caption{Federated learning paradigm. At each iteration, the devices perform the local update on the local datasets. Then, the model weights are transmitted to a base station. Finally, the base station conducts model aggregation and broadcasts the aggregated model weights to the devices.}
\label{FL_Framework}
\end{figure}

Despite the benefits, there are key challenges that limit the applications of FL over wireless mobile networks~\cite{Li20,Aledhari20,Wei20,Brik20}. To date, FL has been applied in a wide variety of fields~\cite{Lim20,Kim20,Theodora18,Meng20}. For example, mobile service providers such as Apple and Google investigate FL for enabling enhanced user experiences through learning from users' behavior~\cite{McMahan20,Hao20}. However, in such important application scenarios, FL needs to deal with practical challenges~\cite{Khan21}. Specifically, FL usually involves a large number of mobile devices to ensure sufficient training data. In this case, it is challenging to deploy FL due to the limited bandwidth resources. In addition, the energy constraints of the individual devices also increase the difficulty of implementing FL. Although most existing spectrum allocation optimization methods mainly focus on total energy consumption, the energy cost of individual devices should be considered. Therefore, an effective spectrum allocation optimization method is critical to the adoption of FL in important applications.

In many scenarios, however, bandwidth resources are still insufficient even though the sophisticated spectrum allocation optimization methods are adopted, hence necessitating device selection methods~\cite{Chen20,Howard20,Amiri20}. As device selection allows FL to choose a subset of wireless devices at each global iteration, the problem incurred by bandwidth limitation is significantly alleviated. Unlike traditional distributed ML which assumes that the datasets are independent and identically distributed (iid)~\cite{Nilsson18}, FL needs to consider the statistical properties of the local datasets and deals with non-iid datasets~\cite{Niknam20}. To handle non-iid datasets in FL, \cite{Wang20} mentioned a benchmark which performs device clustering based on the model weights of the local devices and randomly selects the same number of devices from each cluster. However, this method cannot guarantee that the selected devices are the optimal ones to help FL achieve convergence swiftly. Besides, this method suffers from high training latency as the model weights for training the clustering algorithm are usually high-dimensional.

In this paper, we enhance the performance of FL by proposing an energy-efficient spectrum allocation optimization method and a weight divergence based device selection method. Specifically, the energy-efficient spectrum allocation optimization method aims to optimize the time delay of FL given the energy constraint of each local device. For the device selection problem, we train the K-means to perform device clustering while ensuring low training latency and high clustering performance. Then, we carry out device selection based on the weight divergence between the local model and the global model. The main contributions of our works include:

\begin{enumerate}
    \item We formulate an optimization problem to minimize the time delay of training the entire FL algorithm with the consideration of the energy constraints of the local device. We decompose the optimization problem into two subproblems: spectrum allocation optimization under one global iteration and device selection for each global iteration.

    \item We propose an energy-efficient spectrum allocation optimization method for FL. The proposed method minimizes the sum of the computation and transmission delay under one global iteration while making sure that the energy consumption of each device meets the energy constraint. To this end, we prove that the formulated problem is convex. Then, the problem is solved by KKT conditions and a bisection method, and the optimal solutions of computation capacity and bandwidth are obtained.

    \item For device clustering, we use the model weights of a certain layer in CNNs as the feature vectors to train the K-means algorithm. Compared with training K-means by all the model weights, the proposed training method enables K-means to show better clustering performance. Moreover, the training time is considerably reduced to ensure low latency in FL. 

    \item To overcome the problem incurred by non-iid datasets, we propose a weight divergence based device selection method. At each global iteration, the weight divergence between the local model and the global model is measured. According to the experiments, the global model will gain the highest accuracy on the cluster if the device with the largest weight divergence in this cluster is selected. In this case, for each cluster, the device with the largest weight divergence is chosen to join the local update.

    \item Sufficient numerical experiments show that the proposed spectrum allocation optimization method is capable of minimizing the delay within certain energy budgets and outperforms other baseline approaches. The proposed device selection method is evaluated on multiple datasets. The FL achieves the fastest convergence with the help of the proposed method compared with other device selection methods. The proposed FL framework minimizes the time delay of training the entire FL algorithm.
\end{enumerate}

The rest of the paper is organized as follows. The related work is reviewed in Section~\ref{section2}. System model and problem formulation are defined in Section~\ref{section3}. The weight divergence based device selection method is described in Section~\ref{section5}. The proposed spectrum allocation optimization method is provided in Section~\ref{section4}. Experimental results are provided in Section~\ref{section6} followed by the conclusion in Section~\ref{section7}.

\section{Related Work}~\label{section2}
This paper focuses on spectrum allocation optimization and device selection for FL. While existing works are dealing with spectrum allocation optimization and device selection for FL, our proposed method improves from existing schemes by addressing a number of practical considerations. Specifically, the proposed spectrum resource optimization scheme optimizes the bandwidth and CPU frequency of each local device, and the individual energy constraints are considered in the optimization problem; the proposed device selection method aims to speed up the convergence on non-iid datasets.

For spectrum allocation optimization to enhance FL, \cite{Zhu19} proposed a low-latency multi-access scheme for edge learning to minimize the communication latency of FL. \cite{Yang20} focused on minimizing the FL completion time by determining the tradeoff between computation delay and transmission delay. These works mainly considered minimizing the delay without considering energy consumption, while our proposed method considers not only time delay but also individual energy costs. \cite{Luo20} proposed a hierarchical framework for federated edge learning to jointly optimize computation and communication resources. \cite{Dinh21} illustrated two tradeoffs in FL and proposed FEDL to jointly optimize the FL completion time and energy consumption. These works considered both delay and energy by introducing an importance weight, though did not investigate the means for obtaining the optimal weight given the resources and time constraints. Our proposed method balances the tradeoff between delay and energy adaptively given the individual energy constraint. \cite{Vu20} proposed a novel optimization algorithm to minimize the training time for cell-free massive multiple-input multiple-output systems. \cite{Wang19} studied the convergence bound for FL and propose a control algorithm to optimize the frequency of global aggregation under a resource budget constraint. \cite{Yang21} proposed an iterative algorithm to minimize the total energy consumption of FL and designed a bisection-based algorithm to optimize the time delay. Compared with our proposed method, these works did not consider the limited battery power of the individual local device.

For device selection, \cite{McMahan17} proposed FedAvg that randomly selects a portion of local devices. \cite{Abdulrahman21} proposed a multicriteria-based approach for client selection to reduce total number of global iterations and optimize the network traffic. \cite{Nishio19} proposed a new FL protocol to select devices based on resource conditions to accelerate performance improvement. \cite{shi20} designed a greedy algorithm for device selection to balance the trade-off between the total number of global iterations and the latency. Compared with our proposed method, these works did not consider the data distribution of local datasets or assumed the local datasets to be iid. Recently, \cite{Wang20} proposed Favor that adopts double deep Q-learning~\cite{Hado15} to perform active device selection on non-iid local datasets. However, a new policy network has to be trained from scratch when the datasets are changed, while our approach aims to remain effective when directly applied to other FL tasks with different datasets.

There are some works that investigate both spectrum allocation optimization and device selection~\cite{Chen21,Chen2021,Shi21,Zeng20,Tengchan20,Xu21}. However, the impact of CPU frequency were not considered by \cite{Shi21,Zeng20,Tengchan20,Xu21}; whereas, individual energy constraints were not addressed in \cite{Chen21,Shi21,Zeng20}; the data distribution of the local datasets were not addressed by \cite{Zeng20,Tengchan20,Chen2021}.

\section{System Model and Problem Formulation}~\label{section3} 
We consider a wireless mobile network that consists of one server and $N$ users that are indexed by $\mathcal{N} = \left\{1,2,...,N\right\}$. Each user $n$ contains its local dataset $\mathcal{D}_n = \left\{(\bm{x}_i \in \mathbb{R}^d,y_i \in \mathbb{R}) \right\}^{D_n}_{i=1}$ with $D_n$ data samples, where $\bm{x}_i$ denotes the $i$-th $d$-dimensional input vector, and $y_i$ is its corresponding label. The models trained on the local devices are called the local models, while the model trained on the server is called the global model.

\subsection{FL Training}
We define the model weights as $\bm{w}$. The loss function of $i$-th sample is defined as $f_i(\bm{w})$. For each device $n$ with $D_n$, the loss function for updating the model is
\begin{equation}
    F_n(\bm{w}) = \frac{1}{D_n} \sum_{i=1}^{D_n}f_i(\bm{w}).
\end{equation}
The goal of FL is to minimize the following global loss function $F(\bm{w})$ on the whole dataset~\cite{McMahan17}
\begin{equation}
    \mathop{\rm min}\limits_{\bm{w}}F(\bm{w}) = \sum^{N}_{n}\frac{D_n}{D}F_n(\bm{w}),
    \label{globalFunc}
\end{equation}
where $D = \sum_{n=1}^N D_n$. Algorithm~\ref{FL} provides a typical FL paradigm to solve the problem~\eqref{globalFunc}~\cite{McMahan17,Zeng20,Ren20}. At $k$-th global iteration, a subset of local devices $\mathcal{S}_k \subseteq \mathcal{N}$ with $S_k$ devices is formulated based on a device selection method. According to line 7 in Algorithm~\ref{FL}, the local models adopt the gradient descent (GD) algorithm for local training. As GD requires to span over all the samples in the local dataset, it causes a long training time if the size of the local dataset is large. To solve this problem, stochastic gradient descent (SGD) is an effective method since it only uses part of samples to train the local model~\cite{Bottou12}.

\begin{algorithm}[t]
\textsl{}\setstretch{1.1}
\caption{Federated learning framework}
\begin{algorithmic}[1] 
\Require{Set of local devices $\mathcal{N} = \{1,2,..,N\}$, maximum local iteration number $L$, learning rate $\zeta$, target accuracy $A$. }
\Ensure{Global model $w$}
\State Initialize the global model $\bm{w}^k$, where $k = 1$;
\Repeat
    \State A subset of local devices $\mathcal{S}_k \subset \mathcal{N}$ is formulated;
    \State The server broadcasts the global model $\bm{w}^k$ to the
    \Statex \quad\ \ local devices in $\mathcal{S}_k$;
    \For{each local device $n \in \mathcal{S}_k\  \textbf{in parallel}$}
        \State Initialize $i = 0$;
        \While{$i < L$}
            \State The local device performs local update  \begin{equation}
                    \bm{w}_{n,k}(i+1) = \bm{w}_{n,k}(i) - \zeta \nabla F_n(\bm{w}_{n,k}(i));
                    \label{localtrain}
                    \end{equation}
            \State $i = i+1$;
        \EndWhile
        \State The local device sends $\bm{w}_{n,k}$ to the server;
    \EndFor
    \State The server performs global aggregation    \begin{equation}
        \bm{w}^k = \frac{\sum_{n \in \mathcal{S}_k}D_n\bm{w}_{n,k}}{\sum_{n \in \mathcal{S}_k}D_n}
        \label{aggregation}
    \end{equation}
\State $k = k+1$
\Until{The accuracy of the global model achieves $A$}
\end{algorithmic}  
\label{FL}
\end{algorithm}

\subsection{Computation and Communication Model}
We formulate the computation and communication model to calculate the energy cost and the FL completion time for FL. Let $C_n$ be the number of CPU cycles for device $n$ to compute one sample data. We denote CPU frequency of device $n$ by $f_n$. We assume all sample data have the same size, the total CPU cycles for each global aggregation at device $n$ is $LC_nD_n$, where $L$ denotes the maximum number of local iterations. For local device $n$, the computation time of $L$ local iterations is:
\begin{equation}
    t^{\text{cmp}}_n = \frac{LC_nD_n}{f_n}.
\end{equation}
The energy consumption of $L$ local iterations at local device $n$ is formulated as~\cite{Dinh21}:
\begin{equation}
    e^{\text{cmp}}_n = \frac{\alpha}{2}LC_nD_nf^2_n,
\end{equation}
where $\frac{\alpha}{2}$ denotes the effective capacitance coefficient of device $n$’s computing chipset. 
When finishing local update, the model weights of the local models are transmitted to the server, hence resulting in the energy consumption and the delay. In this paper, a frequency-division multiple access (FDMA) protocol is used for data transmission. The achievable transmission rate of local device $n$ is:
\begin{equation}
    r_n = b_{n}\text{log}_2(1+\frac{h_np_n}{N_0b_{n}}),% $\sum_{n=1}^N b_{n,m} \leq B_m, \text{ for } 1\leq m \leq M$
\end{equation}
where $b_n$ is the bandwidth allocated to device $n$, $h_n$ is the static channel gain of device $n$ during model sharing duration, $N_0$ is background noise, and $p_n$ is the transmission power. Based on this, the communication time of local device $n$ is defined as:
\begin{equation}
    t^{\text{com}}_n = \frac{z_n}{r_n},
\end{equation}
where $z_n$ denotes the size of the model weights of local device $n$. As a result, the energy consumption for transmitting the model weights at local device $n$ is
\begin{equation}
    e^{\text{com}}_n = t^{\text{com}}_{n}p_{n} = \frac{p_{n}z_{n}}{b_n\text{log}_2(1+\frac{h_{n}p_{n}}{N_0b_n})}.
\end{equation}
It can be seen that a larger $f_n$ leads to a shorter time delay while causing more energy consumption. On the other hand, a larger $b_n$ helps the FL to reduce energy consumption and time delay. However, $b_n$ is limited by the total bandwidth $B$. Therefore, the resource allocation in FL necessitates an effective spectrum allocation optimization method. 

As the downlink bandwidth is much higher than that of uplink and the transmit power of BS is much larger than the those of local devices, the latency of broadcasting the global model is not considered~\cite{Dinh21}. Besides, the delays incurred in branch-and-bound-based solution accounted for the latency are not considered in this work, since we assume that the server has powerful computation capability. Thus, the time delay and the energy consumption of FL are derived as follows:

\begin{equation}
    E_k = \sum_{n \in \mathcal{S}_k} \left( e^{\text{com}}_n + e^{\text{cmp}}_n \right),\ E = \sum_{k=1}^{K}E_k,
    \label{E_equa}
\end{equation}
\begin{equation}
    T_k = \max_{n \in \mathcal{S}_k} \left( t^{\text{com}}_n + t^{\text{cmp}}_n \right),\ T = \sum_{k=1}^{K}T_k,
    \label{T_equa}
\end{equation}
where $K$ is the total number of global iterations, $E_k$ and $T_k$ are the energy consumption and the time delay at the $k$-th iteration, respectively. $E$ and $T$ are the total energy consumption and the time delay of training the entire FL algorithm, respectively. 

\begin{figure}[t]
  \centering
  \includegraphics[width=8.5cm]{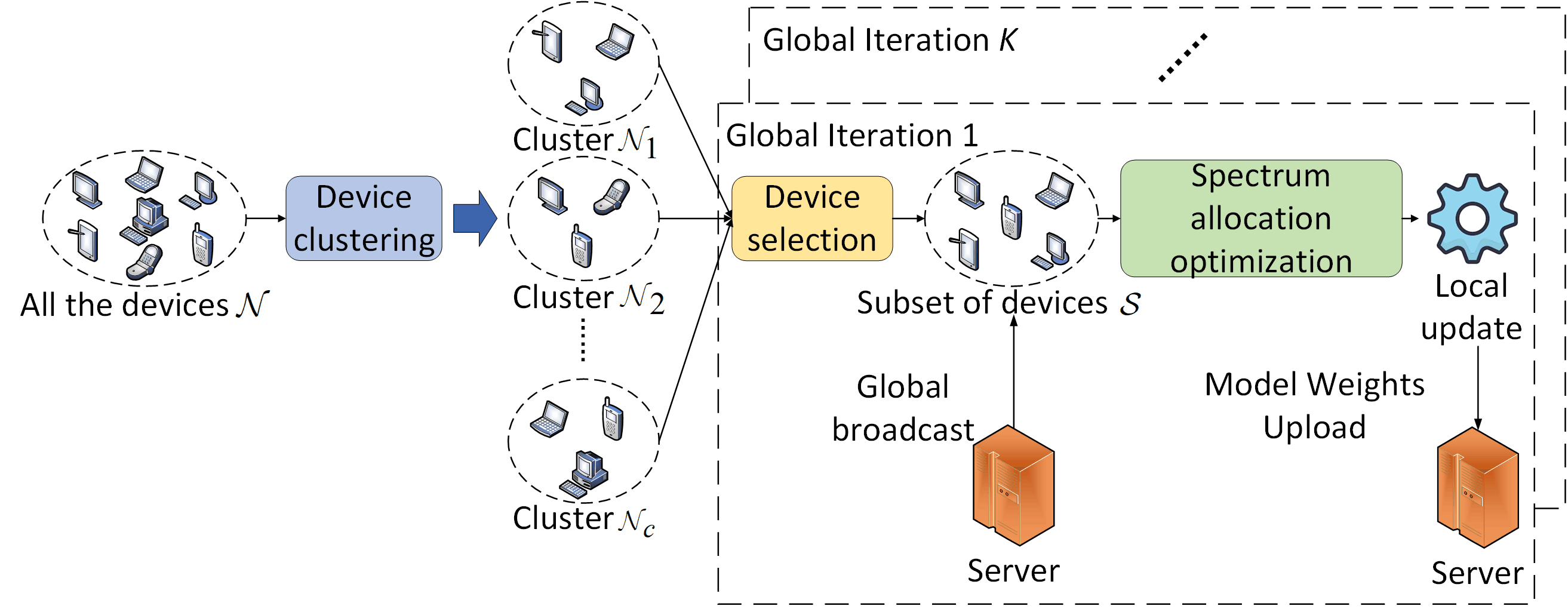}
\caption{Workflow of the proposed FL framework}
\label{WholeFramework}
\end{figure}

\subsection{Problem Formulation}
In this paper, we formulate the optimization problem~\eqref{A} to minimize the time delay of FL while considering the energy constraints of the local devices as follows
\begin{align}
&\underset{\bm{\mathcal{B}}, \bm{\mathcal{F}}, \bm{\mathcal{S}}}{\text{min}}\ \sum_{k=1}^{K}T_k \label{A}\\
&\textrm{s.t.} \quad e^{\text{com}}_n + e^{\text{cmp}}_n \leq  e^{\text{cons}}_n,\ \forall n \in \mathcal{S}_k,\ k = 1,..,K,  \tag{\ref{A}{a}} \label{Aa}\\
&\quad\, \quad t^{\text{com}}_n + t^{\text{cmp}}_n \leq  T_k,\ \forall n \in \mathcal{S}_k,\ k = 1,..,K,  \tag{\ref{A}{b}} \label{Ab}\\
&\quad\ \; \, \sum_{n \in \mathcal{S}_k}b_n\leq B,\ k = 1,..,K, \tag{\ref{A}{c}} \label{Ac}\\
&\quad\ \;\; \,\,  f_n^{\text{min}} \leq f_n \leq f_n^{\text{max}},\ \forall n \in \mathcal{S}_k,\ k = 1,..,K. \tag{\ref{A}{d}} \label{Ad}\\
&\quad\ \;\; \,\,  A_k < A,\ k = 1,..., K-1 \tag{\ref{A}{e}} \label{Ae}\\
&\quad\ \;\; \,\,  A_K \geq A,\tag{\ref{A}{f}} \label{Af}
\end{align}
where $B$ represents total bandwidth, $e^{\text{cons}}_n$ is the energy constraint of local device $n$, $f^{\text{max}}_n$ and $f^{\text{min}}_n$ represent the maximum and minimum CPU frequencies of local device $n$, respectively. $\bm{\mathcal{B}} = \{\bm{b}_k| k = 1,...,K\}$, $\bm{b}_k = \{b_n| n \in \mathcal{S}_k\}$, $\bm{\mathcal{F}} = \{\bm{f}_k| k = 1,...,K\}$, $\bm{f}_k = \{f_n| n \in \mathcal{S}_k\}$, $\bm{\mathcal{S}} = \{\mathcal{S}_k|k = 1,...,K \}$, $A_k$ is the accuracy of the global model at the $k$-th iteration, and $A$ is the target accuracy. \eqref{Aa} ensures that all the devices satisfy the energy requirements. \eqref{Ab} reflects the time delay of FL at each global iteration. \eqref{Ac} and \eqref{Ad} denote the constraints of total bandwidth and the CPU frequency, respectively. \eqref{Ae} and \eqref{Af} denote that $K$ is a design parameter that ensures the global model to achieve the target accuracy. In this paper, the number of the selected devices is fixed at each iteration, which means $S_1 = ... = S_K = S$.

In this work, the local models of FL are convolutional neural networks (CNNs) with the cross-entropy loss function, which indicates that the loss function is non-convex. Therefore, it is challenging to theoretically obtain the upper bound of the total number of global iterations $K$ given an FL task. Moreover, the clustering in the set $\bm{\mathcal{S}}$ makes the problem (12) a combinatorial problem, which further increases the difficulty of solving this problem. To solve the problem~\eqref{A}, we design an FL framework as shown in Fig.~\ref{WholeFramework}. First of all, the device clustering method divides local devices into several clusters based on the data distribution of the model weights. Next, at each global iteration: the device selection method chooses devices from each cluster; the selected devices download the global model from the server, and the spectrum allocation optimization method allocates the bandwidth and CPU frequency for each device; the subset of devices carries out the local update and upload their model weights to the server. The device clustering and selection methods aim to select the optimal devices that enable FL to swiftly achieve the target accuracy. The spectrum allocation optimization method aims to minimize $T_k$ while satisfying the energy constraints. The proposed FL framework is capable of finding a solution candidate ($\bm{\mathcal{B}}, \bm{\mathcal{F}}, \bm{\mathcal{S}}$, $K$) for the problem~\eqref{A} that significantly reduces the time delay $T$.

\section{Weight divergence based Device Selection}~\label{section5}
In this section, we first present the motivation, the implementation details, and the challenges of adopting the K-means algorithm to perform device clustering and selection. Then, a weight divergence based device selection method is proposed to help FL achieve convergence swiftly by counterbalancing the bias introduced by non-iid local datasets.

\subsection{K-means based device clustering and selection}~\label{Kmeanscluster}
In practice, the local devices are usually trained on highly skewed non-iid datasets. In a non-iid local dataset, most of the data samples come from a majority class (also called a dominant class~\cite{Wang20}), while the remaining data samples belong to other classes. It is difficult for the local models to correctly identify the samples unless the samples belong to the majority class. Without a feasible device selection method, the global model is very likely to inherit such detrimental property through model aggregation. Take an image classification task on CIFAR-10~\cite{Krizhevsky09} as an example. If the device selection method rarely selects the devices whose majority class is "dog", the global model will achieve low accuracy in recognizing the samples that belong to the "dog" class.

\begin{algorithm}[t]
\textsl{}\setstretch{1}
\caption{K-means based device clustering}
\begin{algorithmic}[1] 
\Require{A set of local devices $\mathcal{N}$, number of clusters $c$}
\Ensure{$c$ clusters $\{ \mathcal{N}_1,...,\mathcal{N}_c \}$}
\State Create $c$ empty sets $\{\mathcal{N}_1,...,\mathcal{N}_c\}$
\State Initialize global model $\bm{w}^0$ and broadcast $\bm{w}^0$ to the devices in $\mathcal{N}$
\For{each local device $n \in \mathcal{N}$ \textbf{in parallel}}
    \State Perform local update based on~\eqref{localtrain} for $L$ iterations
    \State Transmit model weights $\bm{w}_{n,0}$ to the server
\EndFor
\State The server trains a K-means model with $c$ clusters using $\{\bm{w}_{1,0},...,\bm{w}_{N,0}\}$ based on~\eqref{Kmeans1} and~\eqref{Kmeans2}
\For{each local device $n \in \mathcal{N}$}
    \State The K-means model predicts the cluster label $c_n \in $
    \Statex \quad\  $\{1,...,c\}$ for local device $n$
    \State $\mathcal{N}_{c_n} \leftarrow n$
\EndFor
\State \Return $c$ clusters $\{ \mathcal{N}_1,...,\mathcal{N}_c \}$
\end{algorithmic}  
\label{kmeans}
\end{algorithm}

\begin{algorithm}[t]
\textsl{}\setstretch{1}
\caption{K-means based device selection at the $k$-th global iteration}
\begin{algorithmic}[1] 
\Require{Set of local models $\{\bm{w}_{1,k},\bm{w}_{2,k},...,\bm{w}_{N,k}\}$, $c$ clusters $\{ \mathcal{N}_1,\mathcal{N}_2,...,\mathcal{N}_c \}$, number of devices selected from each cluster $s$}
\State Initialize an empty set $\mathcal{S}_k$;
\For{each cluster $\mathcal{N}_i$ ($i=1,2,...,c$)}
    \State Randomly select $s$ devices from $\mathcal{N}_i$ to $\mathcal{S}_k$;
\EndFor
\State \Return Set of selected devices $\mathcal{S}_k$.
\end{algorithmic}  
\label{KmeansSelect}
\end{algorithm}

To overcome the harmful effect of non-iid datasets, an intuitive wisdom is to make sure that the majority classes of the local datasets in $\mathcal{S}_k$ cover all the classes, in which case the bias incurred by non-iid datasets is significantly balanced. However, the server cannot know the majority class of the local dataset as the datasets are not transmitted to the server. In this case, unsupervised learning is a sensible choice to deal with this problem~\cite{Zhu21}. In this paper, the K-means algorithm~\cite{MacQueen67} is used to perform device clustering based on the data distribution of the model weights. Specifically, given the model weights $\{\bm{w}_n|n = 1,...,N\}$ and the number of clusters $c$, the K-means algorithm iteratively performs two steps as follows
\begin{align}
&\mathcal{N}^{(j)}_i  = \big\{\bm{w}_n \big| \Vert \bm{w}_n - m_i^{(j)} \Vert^2 \leq \Vert \bm{w}_n - m_q^{(j)} \Vert^2, 1\leq q\leq c \big\} \label{Kmeans1}\\
&m_i^{(j+1)}  = \frac{1}{\big|\mathcal{N}^{(j)}_i\big|}\sum_{\bm{w}_n \in \mathcal{N}^{(j)}_i}\bm{w}_n \label{Kmeans2}
\end{align}
where $m_i^{(j)}$ is the centroid of $\mathcal{N}_i$ at the $j$-th iteration. The K-means algorithm converges when the assignments no longer change. Algorithm~\ref{kmeans} and Algorithm~\ref{KmeansSelect} present the details of K-means based device clustering and selection, respectively. In Algorithm~\ref{kmeans}, a K-means model is formulated by the model weights of the local model. The well-trained K-means model is used to predict the cluster label for each local device. As a result, $c$ clusters $\{\mathcal{N}_1\text{\textasciitilde}\mathcal{N}_c\}$ are obtained, where $c$ is the number of clusters. Note that device clustering only performs at the initial global iteration. In the following global iterations, $s$ devices are randomly selected from each cluster to form $\mathcal{S}_k$ according to Algorithm~\ref{KmeansSelect}. As a result, the class imbalance issue is mitigated as model aggregation enables the global model to achieve similar performance on all the classes.

However, K-means based device clustering is bottlenecked by long training latency especially when heavyweight CNN models are used as the local model. This is because the feature vectors for training the K-means model are the model weights which usually contain millions of parameters. In addition, randomly choosing devices from each cluster does not guarantee that the selected devices are the optimal ones to help FL reach fast convergence. There is still room to further enhance K-means based device clustering and selection.

\begin{figure}[t]
  \centering
  \includegraphics[width=8.5cm]{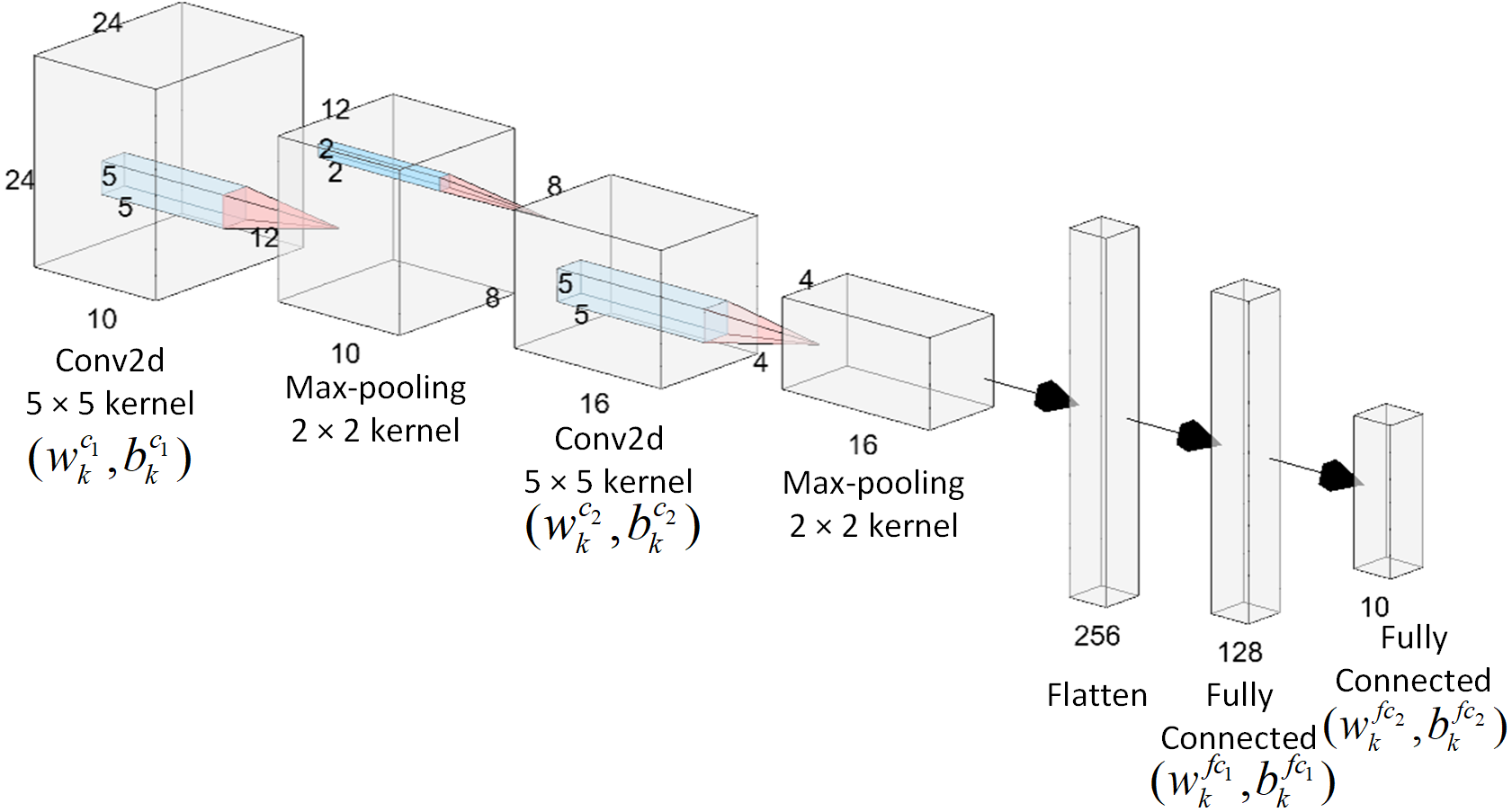}
\caption{CNN architecture of the local model. $\emph{w}^{c_i}$ and $\emph{b}^{c_i}$ $(i = 1$ or $2)$ are the weights and the biases of the $i$-th convolutional layer, respectively. $\emph{w}^{fc_i}$ and $\emph{b}^{fc_i}$ are the weights and the biases of the $i$-th linear layer, respectively.}
\label{CNN}
\end{figure}

\begin{figure*}[htbp]
\centering
\subfigure[$\bm{w}^{c_1}$]{
\begin{minipage}[t]{0.25\linewidth}
\centering
\includegraphics[width=4.2cm]{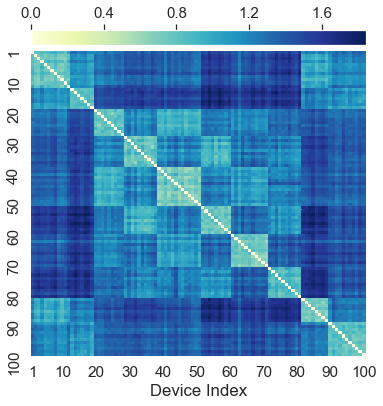}
%\caption{fig1}
\label{confusion_1}
\end{minipage}%
}%
\subfigure[$\bm{b}^{c_1}$]{
\begin{minipage}[t]{0.25\linewidth}
\centering
\includegraphics[width=4.2cm]{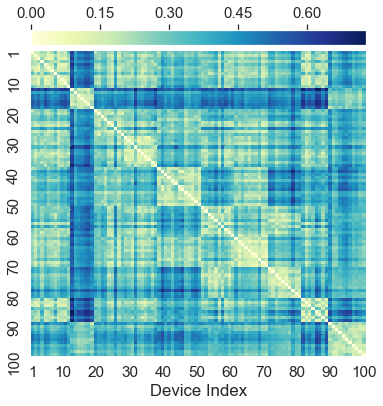}
\label{confusion_22}
%\caption{fig2}
\end{minipage}%
}%
\subfigure[$\bm{w}^{c_2}$]{
\begin{minipage}[t]{0.25\linewidth}
\centering
\includegraphics[width=4.2cm]{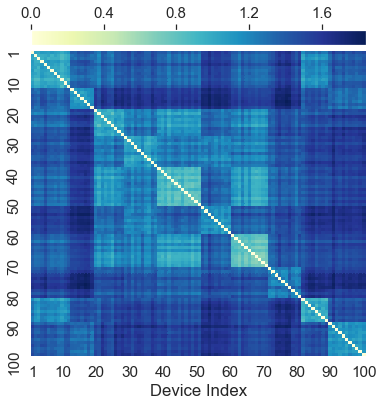}
\label{confusion_3}
%\caption{fig2}
\end{minipage}%
}%
\subfigure[$\bm{b}^{c_2}$]{
\begin{minipage}[t]{0.25\linewidth}
\centering
\includegraphics[width=4.2cm]{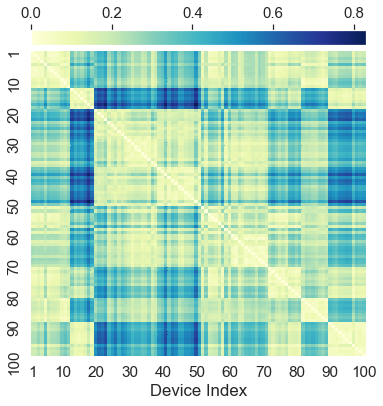}
\label{confusion_4}
%\caption{fig2}
\end{minipage}%
}%

\subfigure[$\bm{w}^{fc_1}$]{
\begin{minipage}[t]{0.25\linewidth}
\centering
\includegraphics[width=4.2cm]{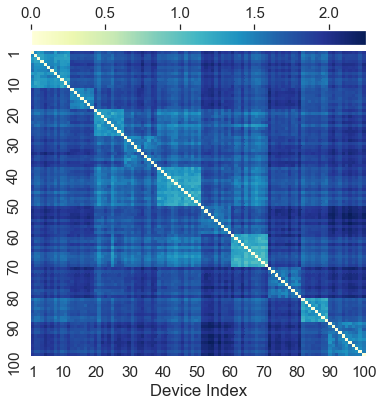}
\label{confusion_55}
%\caption{fig1}
\end{minipage}%
}%
\subfigure[$\bm{b}^{fc_1}$]{
\begin{minipage}[t]{0.25\linewidth}
\centering
\includegraphics[width=4.2cm]{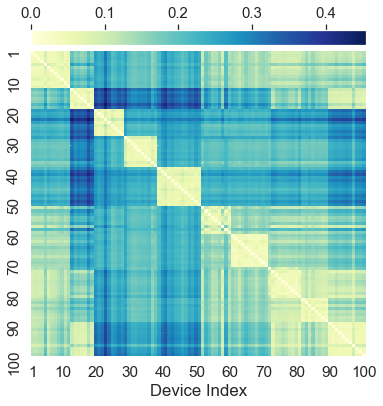}
\label{confusion_6}
%\caption{fig2}
\end{minipage}%
}%
\subfigure[$\bm{w}^{fc_2}$]{
\begin{minipage}[t]{0.25\linewidth}
\centering
\includegraphics[width=4.2cm]{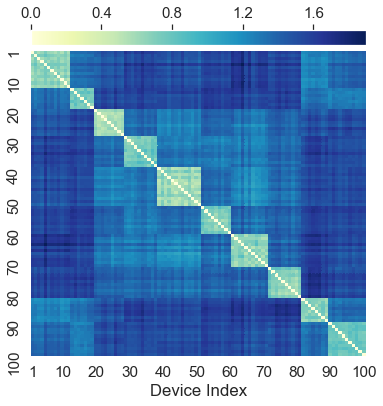}
\label{confusion_7}
%\caption{fig2}
\end{minipage}%
}%
\subfigure[$\bm{b}^{fc_2}$]{
\begin{minipage}[t]{0.25\linewidth}
\centering
\includegraphics[width=4.2cm]{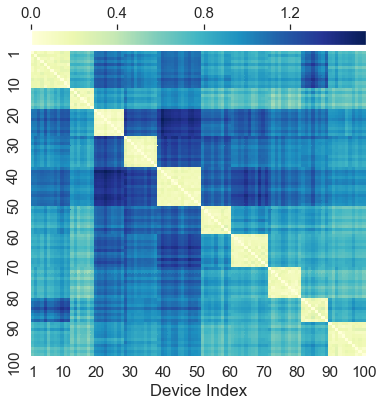}
\label{confusion_8}
%\caption{fig2}
\end{minipage}%
}%

\caption{Euclidean distance between two local devices. Model weights and biases of different layers are used to calculate the Euclidean distance.}
\label{Distance}
\end{figure*}

\subsection{Reducing the computational running time of K-means clustering}

We improve K-means based device clustering by reducing the computational running time without sacrificing the clustering performance. To this end, the size of the feature vector for training the K-means model should be shrunk. Specifically, the K-means model is trained using the model weights of only one layer in CNN instead of all the weights. Since K-means clustering adopts the Euclidean distance to measure the similarities between two samples, we calculate the Euclidean distance between two local models using the model weights of each layer. The clustering performance can be estimated by visualizing the Euclidean distance matrix.

We train a FL model with 100 local devices based on line 2 \textasciitilde\  line 6 in Algorithm~\ref{kmeans}. The CNN architecture of the local models is shown in Fig.~\ref{CNN}. The FL model is trained on CIFAR-10, and non-iid local datasets are generated with a bias $\sigma \in [0,1]$. A local dataset $\mathcal{D}_n$ contains $D_n \times \sigma$ data that belongs to a majority class, while the rest of the data are evenly sampled from other classes. In this section, we set $\sigma = 0.8$.

Fig.~\ref{Distance} depicts the Euclidean distance matrix where the model weights and model biases of different layers are used as the feature vector, respectively. The device indices are successively assigned based on the majority class. For example, the majority class of local device 1\textasciitilde 12 is "airplane", the majority class of local device 13\textasciitilde 19 is "ship", the majority class of local device 20\textasciitilde 28 is "dog", and so on. In Fig.~\ref{Distance}, there are 10 blocks with lighter colors among the distance matrices. It indicates that the Euclidean distance will be much smaller if the two local datasets have the same majority class. Therefore, $c$ is usually set to the number of classes so that each cluster only includes the local models of which datasets belong to the same majority class. Note that this phenomenon is more visible when certain model weights or biases are used to calculate the Euclidean distance (e.g. $\emph{w}^{fc_2}$ and $\emph{b}^{fc_2}$). Furthermore, adding the model weights or biases of some layers into the feature vector may not enhance clustering performance but increasing the training time. Take Fig.~\ref{confusion_22} as an example, the values of the Euclidean distance are very close. Therefore, it is difficult to determine whether two local devices have the same majority class only based on the Euclidean distance calculated by $\emph{b}^{c_1}$. As a result, it is reasonable to use the model weight vector of a certain layer as the feature vector for training the K-means algorithm. In this paper, we find that the K-means algorithm achieves short training time and good performance when using the model weights of the last fully connected layer (i.e., $\emph{w}^{fc_2}$) as the feature vector, and the experimental results are provided in Section~\ref{section6}. 

\subsection{Weight Divergence based Device Selection}
We propose a device selection method to choose devices from each cluster based on the divergence between the local model and the global model. Given the model weights of the local models $\{\bm{w}_{1,k},...,\bm{w}_{N,k}\}$ and the global model $\bm{w}^k$ at the $k$-th global iteration, the weight divergence between each local model and the global model is measured by the Euclidean distance. For each cluster, the device with the largest weight divergence is selected, and the selected devices download the global model from the server. The selected devices perform local updates and send the model weights to the server for global aggregation. According to Section~\ref{Kmeanscluster}, the Euclidean distance of the model weights of some layers will be very small if two local models belong to the same cluster. Therefore, we consider the model weights of all the layers during calculating the weight divergence. Algorithm~\ref{WD} describes the workflow of the proposed device selection method.

\begin{algorithm}[t]
\textsl{}\setstretch{1}
\caption{Weight divergence based device selection at the $k$-th global iteration}  
\begin{algorithmic}[1] 
\Require{Set of local models $\{\bm{w}_{1,k},\bm{w}_{2,k},...,\bm{w}_{N,k}\}$, global model $\bm{w}^k$, $c$ clusters $\{ \mathcal{N}_1,\mathcal{N}_2,...,\mathcal{N}_c \}$, number of devices selected from each cluster $s$}
\State Initialize an empty set $\mathcal{S}_k$;
\For{each cluster $\mathcal{N}_i$ ($i=1,2,...,c$)}
    \State Initialize an empty set $\Omega$;
    \For{each local device $n \in \mathcal{N}_i$};
        \State Calculate the weight divergence $d_n$ between $\bm{w}_{n,k}$
        \Statex \quad\quad\quad and $\bm{w}^k$;
        \State $\Omega \leftarrow d_n$;
    \EndFor
    \State The devices with the top-$s$ values of $d_n$ are selected
    \Statex \quad\; from $\Omega$ to $\mathcal{S}_k$;
\EndFor
\State \Return Set of selected devices $\mathcal{S}_k$.
\end{algorithmic}  
\label{WD}
\end{algorithm}

\begin{table*}[h]
\centering
\captionsetup{justification=centering}
\caption{Weight divergence at the $k$-th iteration and accuracy of the global model on the "dog" cluster at the $(k+1)$-th iteration when selecting different devices from the "dog" cluster to join the $(k+1)$-th iteration. Note that the devices selected from other clusters are fixed.}
\renewcommand\arraystretch{1}
\setlength{\tabcolsep}{1.5mm}{
\begin{tabular}{cccccccccc}\toprule
Index of the selected device    & 20    & 21     & 22     & 23    & 24    & 25    & 26    & 27    & 28    \\\midrule
Weight divergence  & 5.09 & 5.33  & \textbf{16.92} & 6.90 & 5.40 & 5.32 & 7.39 & 6.56 & 5.87 \\
Accuracy  & 52.17\% & 50.27\% & \textbf{53.39\%}  & 51.39\% & 52.21\% & 51.34\% & 52.34\% & 50.90\% & 51.81\% \\ \bottomrule 
\end{tabular}}
\label{increment}
\end{table*}

An experiment is conducted to explain the reason for selecting the device with the highest weight divergence. The local datasets and the local models in Section~\ref{Kmeanscluster} are adopted in this experiment to perform FL. For each global iteration (except the initial global iteration), one device is randomly selected from each cluster for the local update. After $k$ global iterations, the weight divergence between the local models and the global model are calculated. The results of the cluster with the majority class "dog" are listed in Table~\ref{increment}. It can be seen that local device 22 has the largest weight divergence. Next, we investigate which device is optimal for FL training. To this end, the global model $\bm{w}^k$ and the selected devices in other clusters are fixed, and local devices $20\text{\textasciitilde}28$ are selected respectively to participate in the $(k+1)$-th global iteration. At the $(k+1)$-th iteration, the accuracy of the global model on the "dog" cluster is listed in Table~\ref{increment}. It can be observed that the global model achieves the highest accuracy on the "dog" cluster when local device 22 is selected. This is because the weight divergence can reflect the performance gap between the local model and the global model on the corresponding local dataset. Thus, the dataset of local device 22 is most informative for the global model. Selecting this device allows the global model to be improved significantly on this cluster. Note that the above experiments are conducted and evaluated on the training set. To further evaluate the proposed device selection method, Section~\ref{section6} provides the performance of the method on the testing set, which demonstrates its effectiveness.

\section{Spectrum allocation optimization method}~\label{section4}

 We first introduce some notations to simplify the presentation. 
\begin{equation}
    J_n = \frac{h_n p_n}{N_0},  \label{W}
\end{equation}  
\begin{equation}
    U_n =  LC_{n}D_n,  \label{WW}
\end{equation}
\begin{equation}
    G_n = \frac{\alpha}{2} L C_{n} D_n,  \label{WWW}
\end{equation}  
\begin{equation}
    H_n = z_{n}p_{n}. \label{WWWW}
\end{equation}  
% \begin{align}
%     \textcolor{red}{J_n} &= \frac{h_n p_n}{N_0},  \label{W}\\
%     U_n &=  LC_{n}D_n,  \label{WW} \\
%     G_n &= \frac{\alpha}{2} L C_{n} D_n,  \label{WWW}\\
%     H_n &= z_{n}p_{n}. \label{WWWW}
% \end{align}  
The goal of the spectrum allocation optimization method is to minimize the time delay at the $k$-th global iteration (i.e., $T_k$) given a subset of devices $\mathcal{S}_k$ and energy constraints $e^\text{cons}_n$. The optimization problem is formulated as follows:
\begin{align}
&\underset{\bm{b}_k, \bm{f}_k}{\text{min}}\quad  T_k \label{X}\\
&\textrm{s.t.} \quad G_n f^2_n + \frac{H_n}{b_n\text{log}_2(1+\frac{J_n}{b_n})} \leq  e^{\text{cons}}_n,\ \forall n \in \mathcal{S}_k,  \tag{\ref{X}{a}} \label{Xa}\\
&\quad \quad \frac{z_n}{b_n \text{log}_2(1+\frac{J_n}{b_n})} + \frac{U_n}{f_n} \leq  T_k,\ \forall n \in \mathcal{S}_k,  \tag{\ref{X}{b}} \label{Xb}\\
&\quad\ \; \, \sum_{n \in \mathcal{S}_k}b_n\leq B,\tag{\ref{X}{c}} \label{Xc}\\
&\quad\ \;\; \,\,  f_n^{\text{min}} \leq f_n \leq f_n^{\text{max}},\ \forall n \in \mathcal{S}_k. \tag{\ref{X}{d}} \label{Xd}
\end{align}

To solve the optimization problem~\eqref{X}, we first provide Lemma~\ref{convex_proof}. 
\begin{lemma}
The optimization problem~\eqref{X} is a convex problem. 
\label{convex_proof}
\end{lemma}
\begin{proof}
See Appendix A.
\end{proof}

With the help of Lemma~\ref{convex_proof}, the following theorem can be derived.
\begin{theorem}
The optimal solutions of the problem \eqref{X} satisfy:
\begin{align}
\frac{z_n}{b^\ast_n\text{log}_2(1+\frac{J_{n}}{b^\ast_n})} + \frac{U_n}{f^\ast_n} - T_k^\ast = 0, \quad n \in \mathcal{S}_k, \label{T_opt} \\
G_n (f^\ast_n)^2 + \frac{H_n}{b^\ast_n\text{log}_2(1+\frac{J_{n}}{b^\ast_n})} - e^{\text{cons}}_n = 0, \quad n \in \mathcal{S}_k, \label{E_opt} \\
\sum_{n \in \mathcal{S}_k}b^\ast_n - B = 0. \quad\quad\quad\  \label{B_opt}
\end{align}
\label{opti_theo}
\end{theorem}
\begin{proof}
See Appendix B.
\end{proof}

\begin{algorithm}[t]
\textsl{}\setstretch{1}
\caption{Energy-efficient spectrum allocation optimization}  
\begin{algorithmic}[1] 
\Require{$B$, $J_n$, $U_n$, $G_n$, $H_n$, $e^{\text{cons}}_n$, $f^\text{min}_n$, $f^\text{max}_n$, $\varepsilon_0$, $b^\text{max}$, $\mathcal{S}_k$}
\Ensure{$T_k^\ast, \bm{b}_k^\ast$, $\bm{f}_k^\ast$}
\State Let $T_{\text{min}} = \text{max}\{\ln2\frac{z_n}{J_n}+\frac{U_n}{f^\text{max}_n}\}$, and give $T_\text{max}$ a big enough value
\State $ratio = 0$, $T_k = \frac{T_{\text{min}}+T_{\text{max}}}{2}$
\While {Not $1-\varepsilon_0 \leq ratio \leq 1$}
    \State $\bm{b}_k^\ast \gets \{\}$
    \For{$n \in \mathcal{S}_k$}
        \State Calculate $f_n$ from~\eqref{cal_f} using bisection method
        \State Clip $f_n$ to the range $[f^\text{min}_n,f^\text{max}_n]$
        \State Calculate $b_n$ from~\eqref{E_opt} using bisection method
        \State $b_n = \min\{b_n, b^\text{max}\}$
        \State $\bm{b}_k^\ast \gets b_n$
    \EndFor
    \State $ratio = \frac{\sum_{n \in \mathcal{S}_k}b^\ast_n}{B}$
    \If{$ratio > 1$}
        \State $T_\text{min} = T_k$
        \State $T_k = \frac{T_\text{max}+T_k}{2}$
    \ElsIf{$ratio < 1 - \varepsilon_0$} 
        \State $T_\text{max} = T_k$
        \State $T_k = \frac{T_\text{min}+T_k}{2}$
    \EndIf  
\EndWhile
\State  Recalculate $\bm{f}_k^\ast = \{f^\ast_1,f^\ast_2,...,f^\ast_{S_k}\}$ using $\bm{b}_k^\ast = \{b^\ast_1,b^\ast_2,...,b^\ast_{S_k}\}$ based on~\eqref{E_opt}
\State  Recalculate $T_k^\ast = \underset{ n \in \mathcal{S}_k}{\text{max}}{\{\frac{z_n}{b^\ast_n\text{log}_2(1+\frac{J_n}{b^\ast_n})} + \frac{U_n}{f^\ast_n}\}}$
\State \Return $T_k^\ast, \bm{b}_k^\ast$, $\bm{f}_k^\ast$
\end{algorithmic} 
\label{weidaima}
\end{algorithm}

The theoretical results provided by Theorem~\ref{opti_theo} are intuitive. In terms of~\eqref{T_opt}, all the devices are expected to have the same time delay, since the bandwidth can always be assigned from the devices with lower latency to other devices that have a higher delay. In terms of~\eqref{E_opt} and~\eqref{B_opt}, the devices should make full use of the bandwidth and the battery power to minimize the time delay. However, it is difficult to directly obtain the explicit results of Theorem~\ref{opti_theo} due to the high dimensional space of $\mathcal{S}_k$. It can be noted that there is no feasible solution when $T_k < T_k^\ast$, where $T_k^\ast$ is the optimal solution of the optimization problem~\eqref{X}. In contrast, a feasible solution always exists when $T_k \geq T_k^\ast$. Therefore, we design a binary search method to search for $T_k^\ast$.

% \begin{lemma}
% Provided $T_k^\ast$ is the optimal solution of the optimization problem~\eqref{X}, there is no feasible solution when $T_k < T_k^\ast$, while a feasible solution always exists when $T_k \geq T_k^\ast$.
% \label{T_proof}
% \end{lemma}
% \begin{proof}
% See Appendix C.
% \end{proof}

\begin{lemma}
Given $x>0$, $Q_n(x) = x\log_2(1+\frac{J_n}{x})$ is a monotonically increasing function with an upper bound $\frac{J_n}{\ln2}$. 
\label{b_proof}
\end{lemma}
\begin{proof}
See Appendix C.
\end{proof}

\begin{lemma}
$M(f_n) = f_n^3 + (\frac{H_n T_k}{z_n G_n}-\frac{e^{\text{cons}}_n}{G_n})f_n - \frac{H_n U_n}{z_n G_n}$ has only one root in $(0, +\infty)$.
\label{f_proof}
\end{lemma}
\begin{proof}
See Appendix D.
\end{proof}

First of all, Equation~\eqref{cal_f} is derived from Equation~\eqref{T_opt} and Equation~\eqref{E_opt} by removing $b^\ast_n\text{log}_2(1+\frac{J_{n}}{b^\ast_n})$:
\begin{equation}
    (f^\ast_n)^3 + (\frac{H_n T^\ast_k}{z_n G_n}-\frac{e^{\text{cons}}_n}{G_n})f^\ast_n - \frac{H_n U_n}{z_n G_n} = 0.
\label{cal_f}
\end{equation}
According to Lemma~\ref{f_proof}, $f_n$ can be obtained from Equation~\eqref{cal_f} using a bisection method with a known $T_k$, and a clipping function is employed to ensure $f_n \in [f_n^{\text{min}},f_n^{\text{max}}]$. Next, since $Q_n(b_n)$ is a monotonically increasing function of $b_n$ according to Lemma~\ref{b_proof}, $b_n$ can be calculated from Equation~\eqref{E_opt} using a bisection method. Note that $b_n$ will become extremely large or even not exist when $T_k$ is too small due to the upper bound of $Q_n(b_n)$. To solve this problem, we use a clipping method for $b_n$ with a clipping threshold $b^{\text{max}}$. After value clipping, $f_n$ and $T_k$ are recalculated using~\eqref{E_opt} and~\eqref{T_opt}, respectively. The details of the proposed method is shown in Algorithm~\ref{weidaima}. 

The complexity of Algorithm~\ref{weidaima} is $\mathcal{O}\left( (S_k)^2\log_2{\left( \frac{1}{\varepsilon_0}\right)}\log_2{\left( \frac{1}{\varepsilon_1}\right)}\log_2{\left( \frac{1}{\varepsilon_2}\right)} \right)$, where $\varepsilon_0$ is the tolerance of searching for $T_k^\ast$, $\varepsilon_1$ and $\varepsilon_2$ are the accuracies of the bisection methods for calculating $b_n$ and $f_n$, respectively. 

In this paper, we adopt centralized optimization for solving the problem~\eqref{X}. Before FL training, all the devices send local information, such as $f_n^{\text{min}},\ f_n^{\text{max}},\ \text{and}\ D_n$, to the server. At each global iteration, the server selects a subset of devices using the proposed device selection method. Then, the server carries out Algorithm~\ref{weidaima} within the selected devices. Finally, the server broadcasts the model weights and spectrum allocation results (i.e., $f_n$ \text{and} $b_n$) to the selected devices. Due to the sufficient downlink bandwidth and the high transmit power of BS, the latency of transmitting allocation results is negligible.

\begin{table*}[t]
\centering
\caption{Number of model weights of different layers in CNNs and model sizes of different CNN architectures}
\renewcommand\arraystretch{1}
\setlength{\tabcolsep}{2mm}{
\begin{tabular}{ccccccccccc}\toprule
\multirow{2}{*}{Dataset} & \multicolumn{9}{c}{Number of model weights}                                                                                                                                & \multirow{2}{*}{Model size $z_n$} \\ \cline{2-10}
                         & $\emph{w}^{c_1}$ & $\emph{b}^{c_1}$ & $\emph{w}^{c_2}$ & $\emph{b}^{c_2}$ & $\emph{w}^{fc_1}$ & $\emph{b}^{fc_1}$ & $\emph{w}^{fc_2}$ & $\emph{b}^{fc_2}$ & All parameters &       \\\midrule
MNIST                    & 375              & 15               & 10500            & 28               & 100352            & 224               & 2240              & 10                & 113744         & 448 KB                            \\
CIFAR-10                 & 1125             & 15               & 10500            & 28               & 210000            & 300               & 3000              & 10                & 224978         & 882 KB                            \\
FashionMNIST             & 250              & 10               & 3000             & 12               & 15360             & 80                & 800               & 10                & 19522          & 79 KB              \\ \bottomrule               
\end{tabular}}
\label{numberweight}
\end{table*}

\renewcommand{\thefigure}{5}
\begin{figure}[t]
\centering
\subfigure[Energy cost of each device]{
\begin{minipage}[t]{1\linewidth}
\centering
\includegraphics[width=8.5cm]{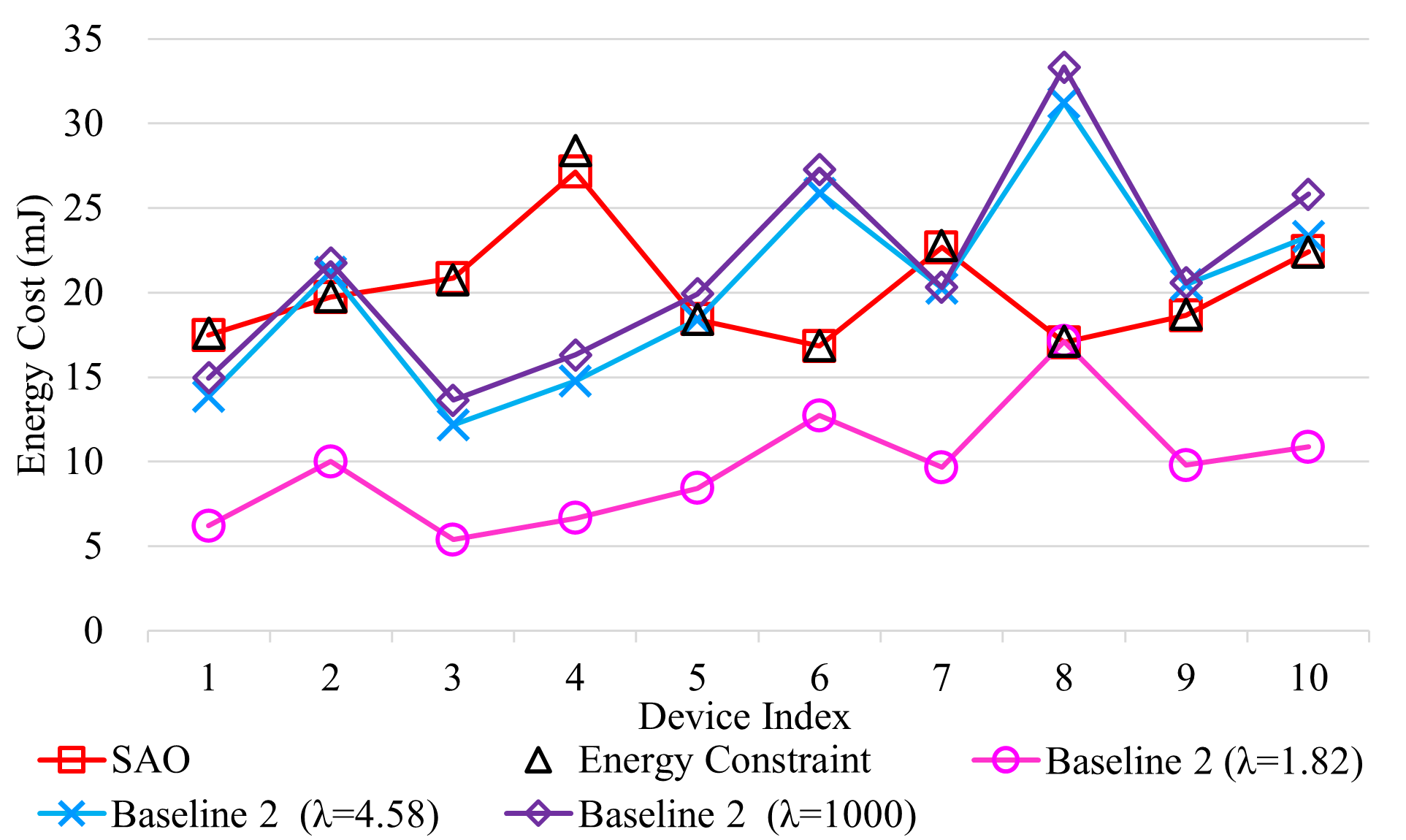}
\label{lmd}
\end{minipage}%
}%

\subfigure[Total energy cost and completion time of FL at one global iteration]{
\begin{minipage}[t]{1\linewidth}
\centering
\includegraphics[width=8.5cm]{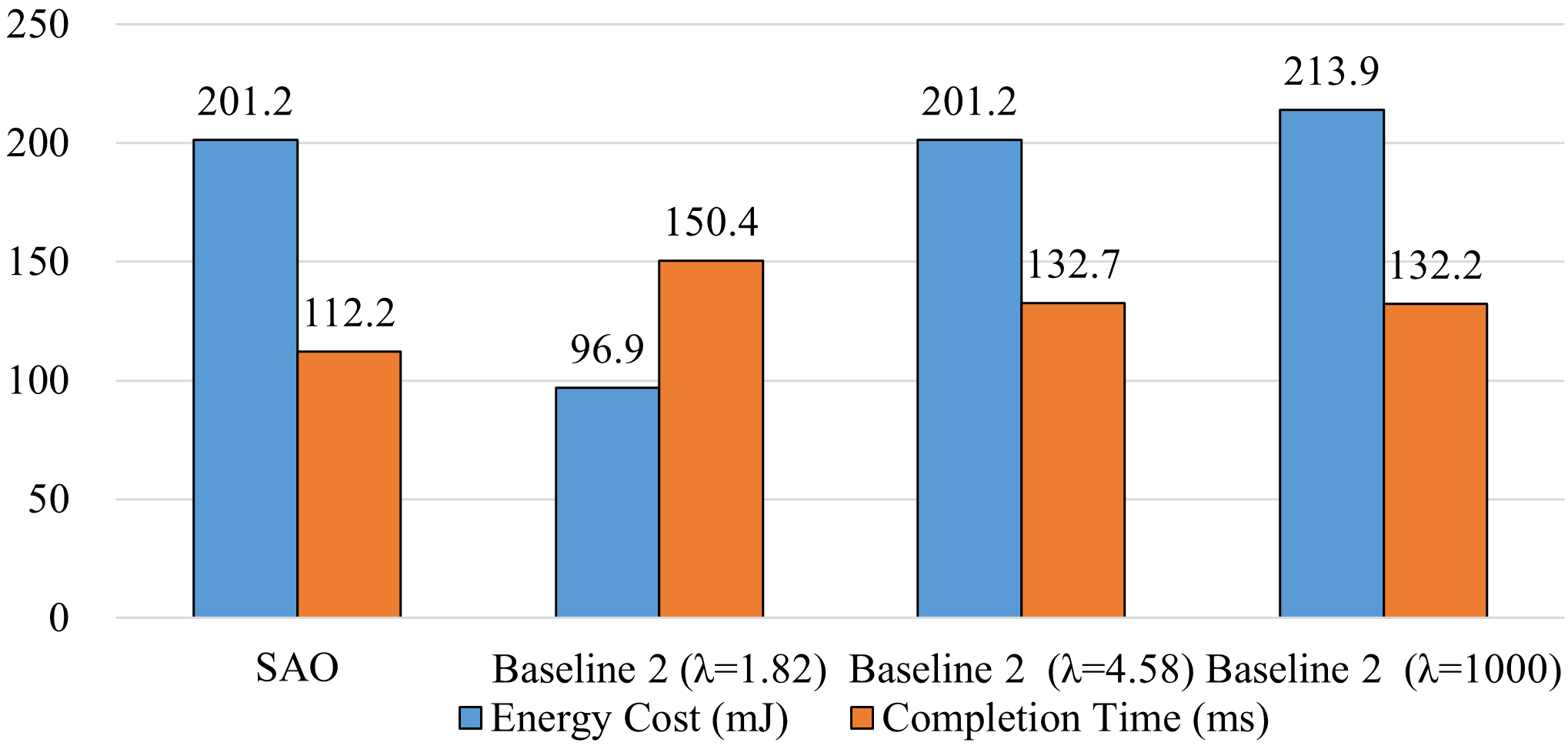}
\label{ET}
\end{minipage}%
}%
\caption{Time delay and energy consumption of FL under one global iteration with $S_k = 10$ devices. When $\lambda = 1.82$, Baseline 2 ensures all the devices to satisfy the energy constraints. When $\lambda = 4.58$, Baseline 2 has the same total energy consumption as that of SAO. When $\lambda = 1000$, Baseline 2 mainly focuses on minimizing the time delay.}
\label{eachdevice}
\end{figure}

\renewcommand{\thefigure}{6}
\begin{figure}[t]
  \centering
  \includegraphics[width=8.5cm]{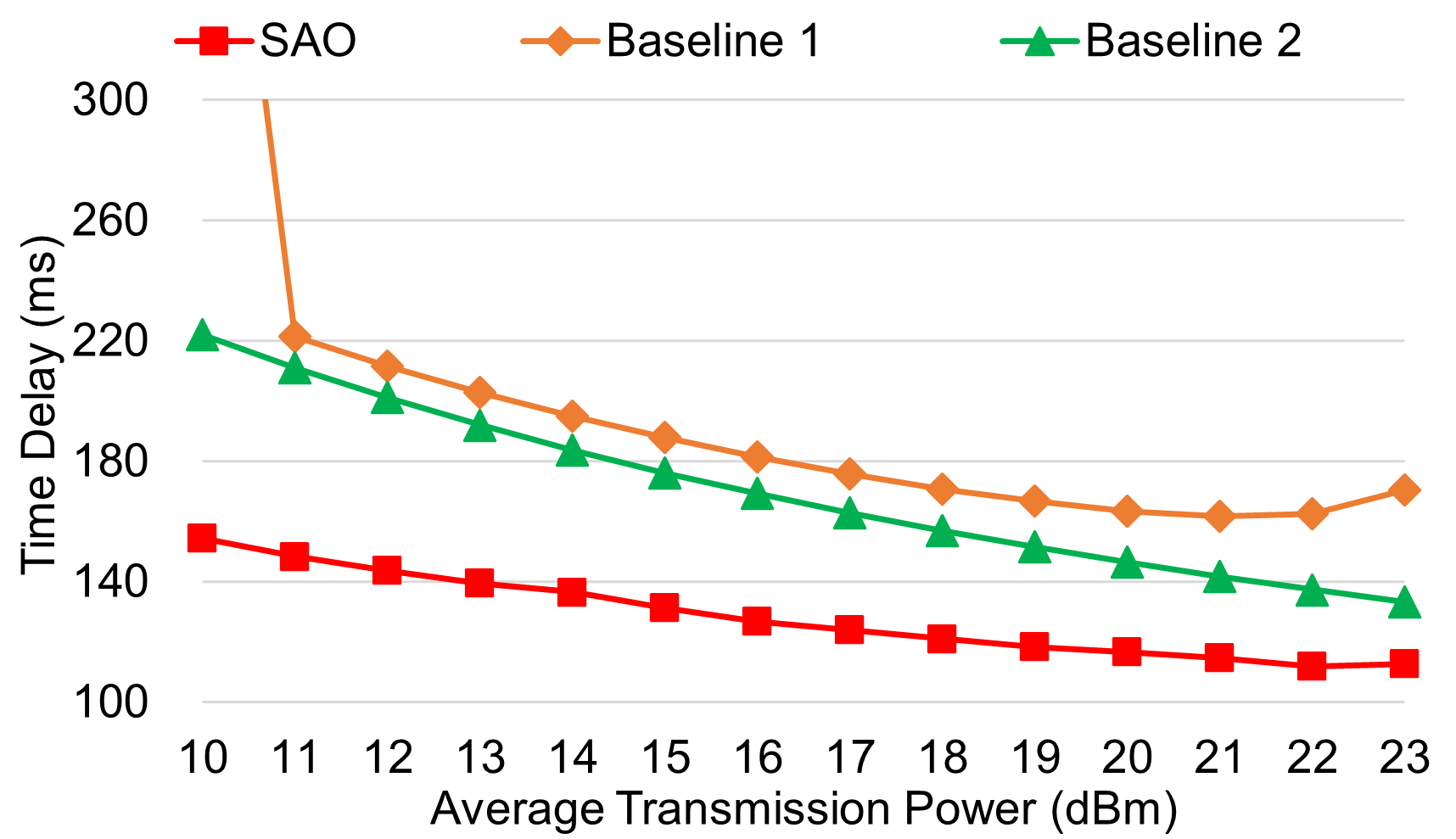}
\caption{Time delay versus average transmission power ($e^{\text{cons}} = 30\ \text{mJ}$).}
\label{P}
\end{figure}
\renewcommand{\thefigure}{7}
\begin{figure}[t]
  \centering
  \includegraphics[width=8.5cm]{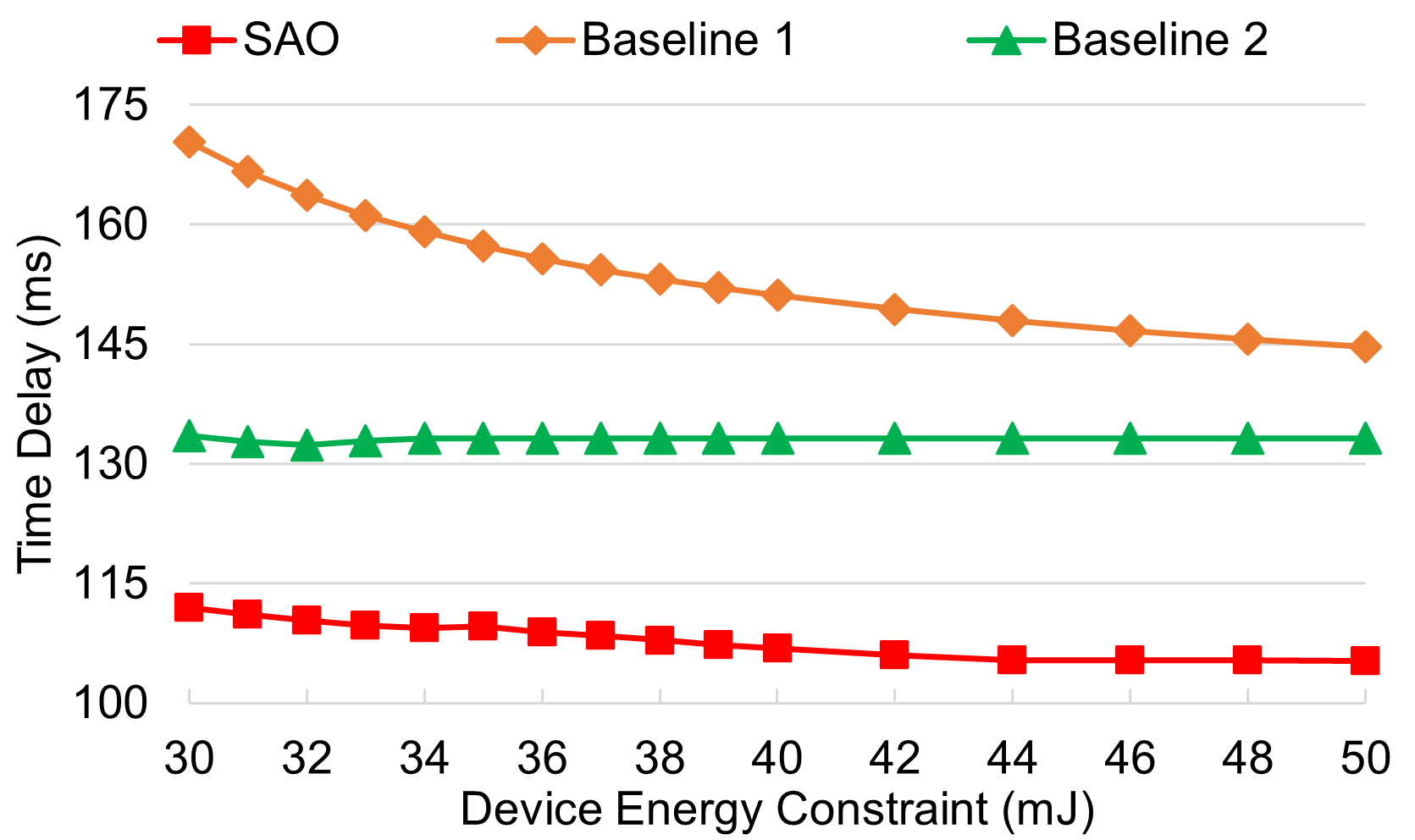}
\caption{Time delay versus energy constraint ($p = 23\ \text{dBm}$).}
\label{E}
\end{figure}

\section{PERFORMANCE EVALUATION}~\label{section6}
We consider $N = 100$ local devices randomly distributed in a cell of radius $R = 300$m, and a server is located at the center of the area. The path loss model is $128.1 + 37.6\log_{10}d$(km), and the standard deviation of shadow fading is $8$dB~\cite{Yang2021}. The power spectrum density of the additive Gaussian noise is $N_0 = -174$ dBm/Hz. 

Section~\ref{SAO_eval} evaluates spectrum allocation optimization with $S = 10$ devices. We have an equal maximum CPU frequency $f^\text{max}_1 = f^\text{max}_2 = ... = f^\text{max}_{N} = 2\ \text{GHz}$, an equal minimum CPU frequency $f^\text{min}_1 = f^\text{min}_2 = ... = f^\text{min}_{N} = 0.2\ \text{GHz}$, and an equal model size $z_1 = z_2 = ... = z_n = 448\ \text{KB}$.

In terms of device selection, the FL model is trained with 100 local devices on three datasets: MNIST~\cite{Cun90}, CIFAR-10, and FashionMNIST~\cite{xiao2017}. For MNIST and CIFAR-10, the model consists of two $5\times5$ convolution layers, and the output dimension of the first layer and the second layer are 15 and 28, respectively. Each layer is followed by $2\times2$ max pooling. The output of the pooling layer is flattened and then fed into the two linear layers. For FashionMNIST, the model consists of two $5\times5$ convolution layers, and the output dimension of the first layer and the second layer are 10 and 12, respectively. Each layer is followed by $2\times2$ max pooling. The output of the pooling layer is flattened and then fed into the two linear layers. The number of model weights of different layers in CNNs is provided in Table~\ref{numberweight}. FL training is conducted under different levels of non-iid data (i.e., $\sigma = 0.5, 0.8,$ and $H$) with a learning rate of 0.05. The meanings of $\sigma = 0.8$ and $\sigma = 0.5$ are illustrated in Section~\ref{Kmeanscluster}. $\sigma = H$ means that the data in each local dataset only cover two labels. 80\% of data belong to a majority class, and others belong to a secondary class. When $\sigma = 0.5$ or 0.8, the FL model is considered to achieve convergence if the accuracies of the global model on MNIST, CIFAR-10, and FashionMNIST are 99\%, 55\%, and 87\%, respectively. For $\sigma = H$, the target accuracies of the global model on MNIST, CIFAR-10, and FashionMNIST are 98.5\%, 52\%, and 85\%, respectively.

\renewcommand{\thefigure}{8}
\begin{figure*}[h]
\centering
\subfigure[MNIST]{
\begin{minipage}[t]{0.33\linewidth}
\centering
\includegraphics[width=6cm]{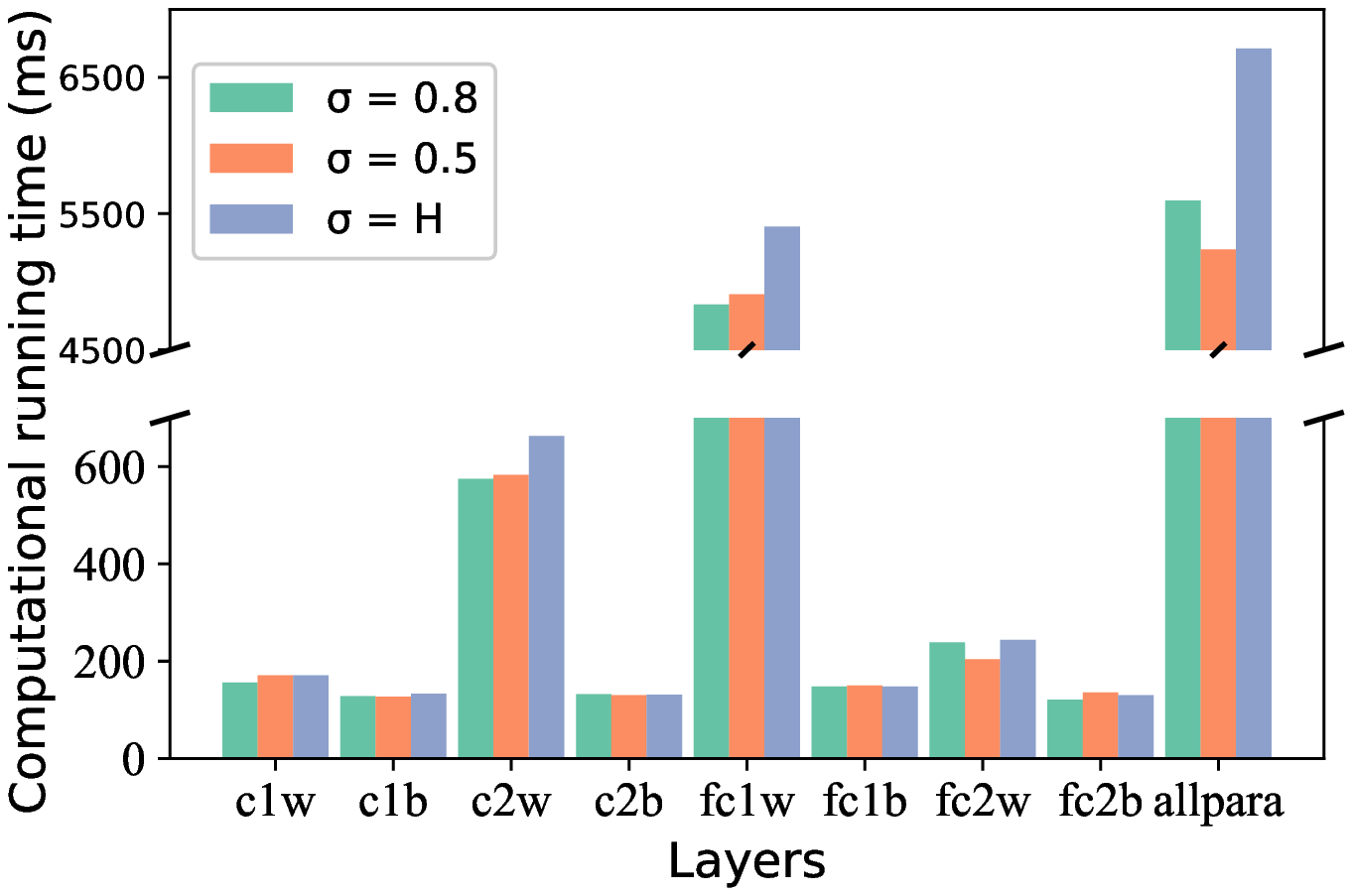}
%\caption{fig1}
\label{confusion_1}
\end{minipage}%
}%
\subfigure[CIFAR-10]{
\begin{minipage}[t]{0.33\linewidth}
\centering
\includegraphics[width=6cm]{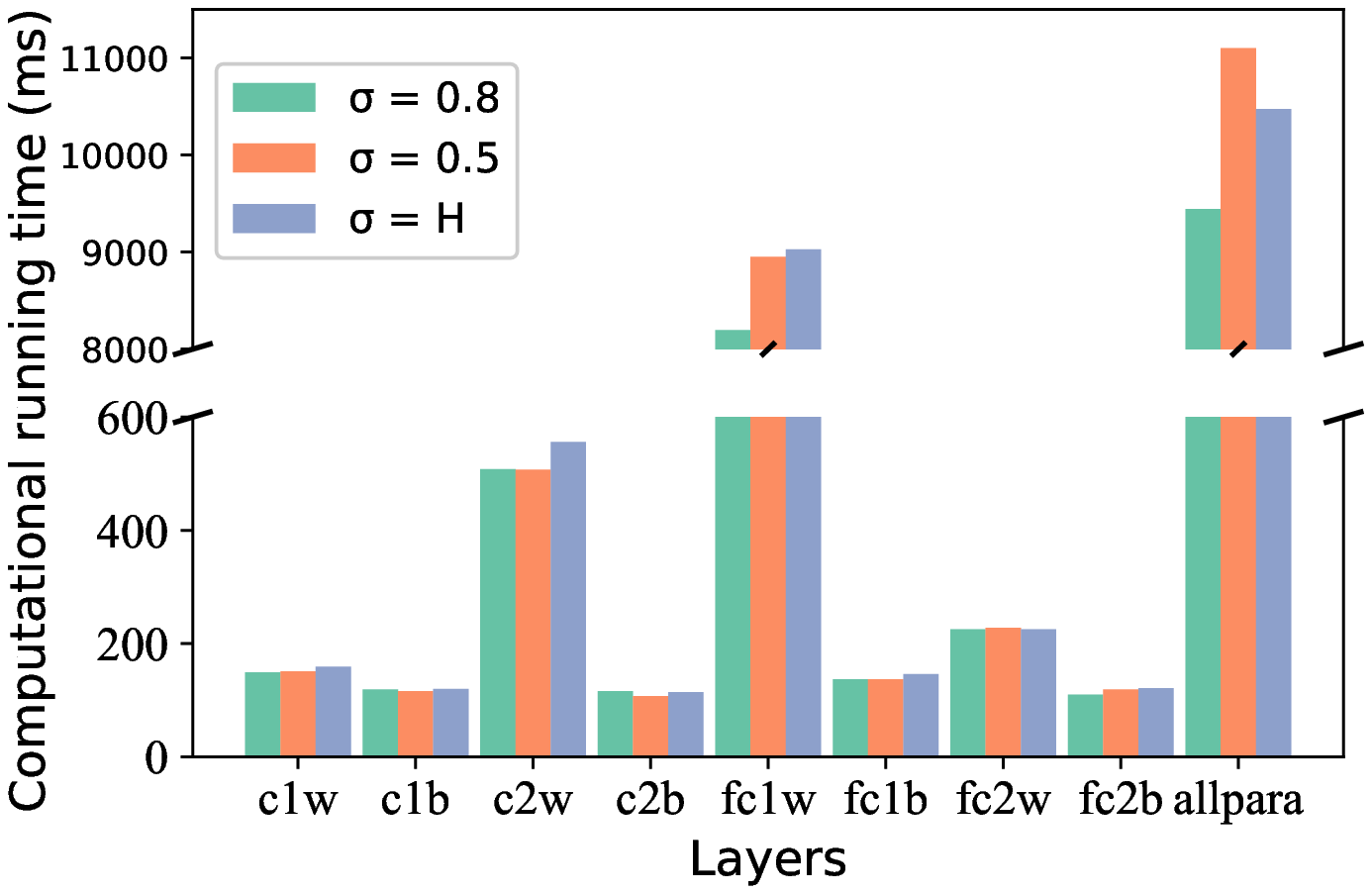}
\label{confusion_2}
%\caption{fig2}
\end{minipage}%
}%
\subfigure[FashionMNIST]{
\begin{minipage}[t]{0.33\linewidth}
\centering
\includegraphics[width=6cm]{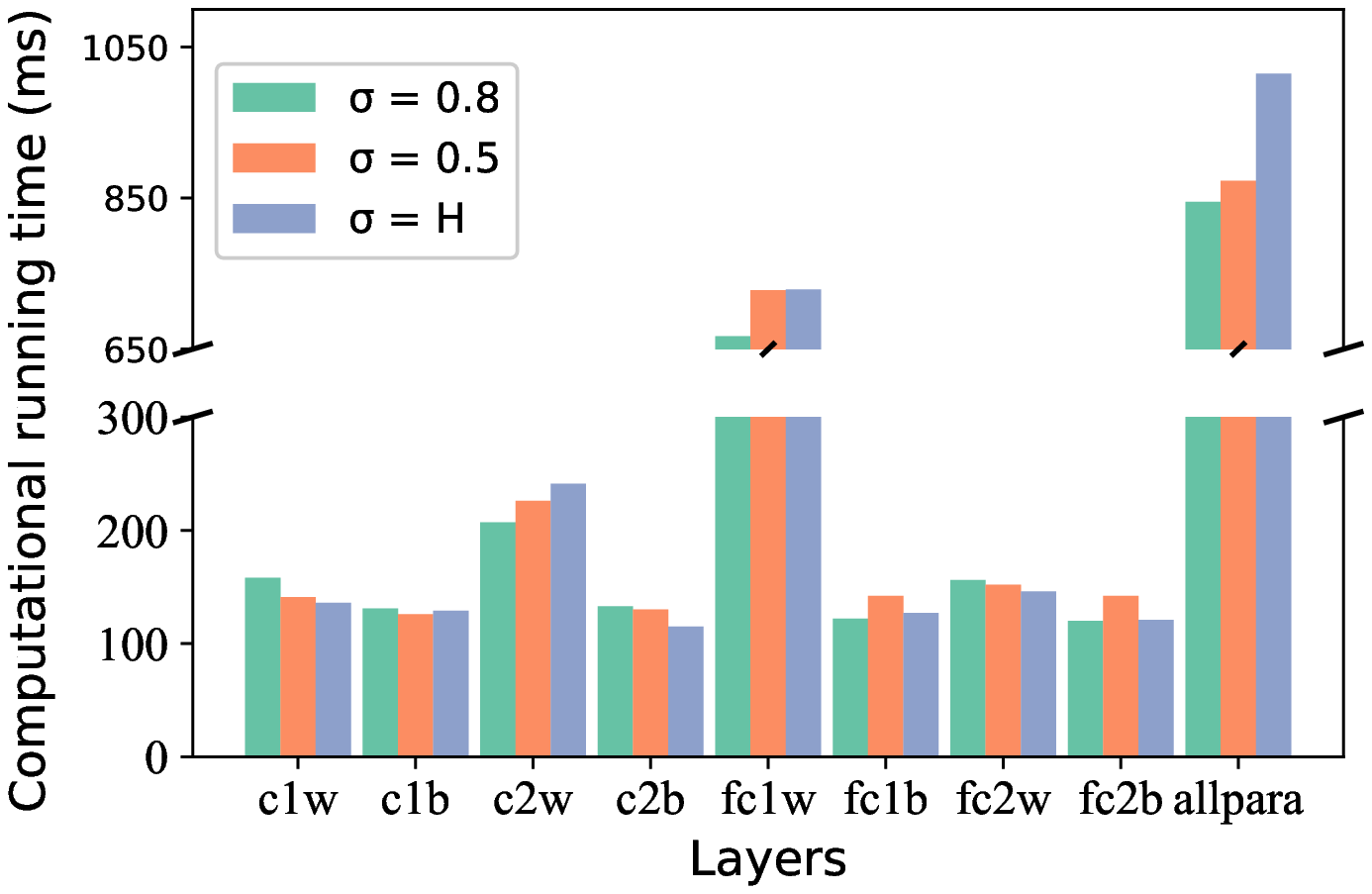}
\label{confusion_3}
%\caption{fig2}
\end{minipage}%
}%
\caption{Computational running time for training the K-means algorithm.}
\label{T_Kmeans}
\end{figure*}
\renewcommand{\thefigure}{9}
\begin{figure*}[h]
\centering
\subfigure[MNIST]{
\begin{minipage}[t]{0.33\linewidth}
\centering
\includegraphics[width=6cm]{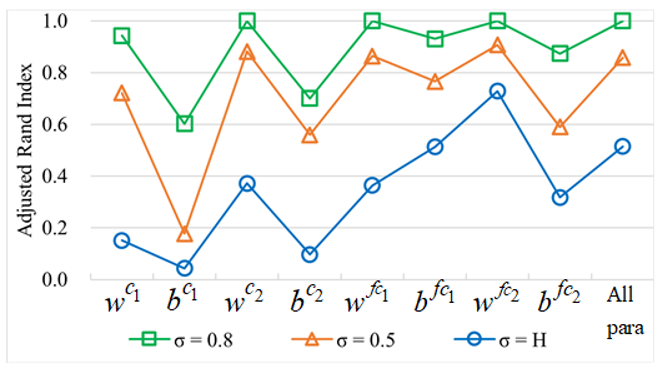}
%\caption{fig1}
\label{confusion_1}
\end{minipage}%
}%
\subfigure[CIFAR-10]{
\begin{minipage}[t]{0.33\linewidth}
\centering
\includegraphics[width=6cm]{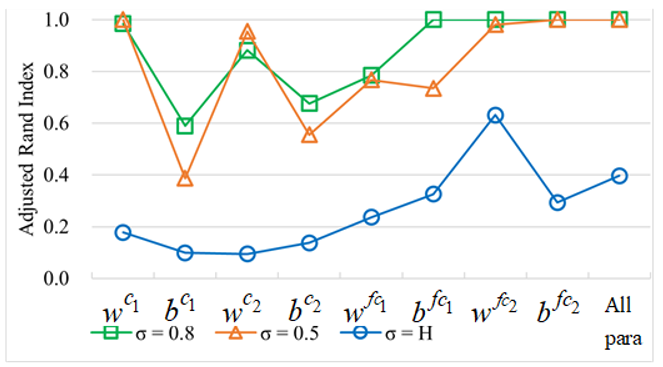}
\label{confusion_2}
%\caption{fig2}
\end{minipage}%
}%
\subfigure[FashionMNIST]{
\begin{minipage}[t]{0.33\linewidth}
\centering
\includegraphics[width=6cm]{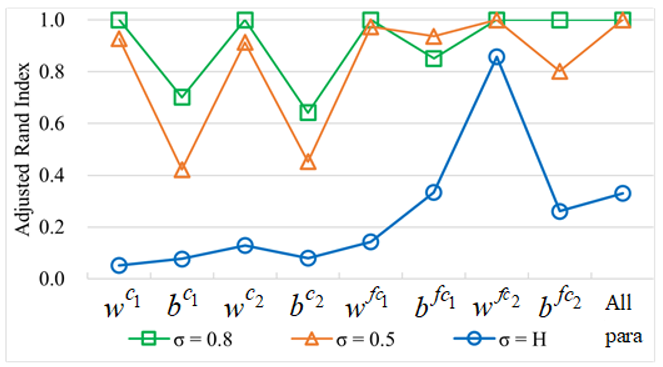}
\label{confusion_3}
%\caption{fig2}
\end{minipage}%
}%
\caption{Adjusted Rand index (ARI) of the K-means algorithm.}
\label{ARI}
\end{figure*}

\subsection{Evaluation of spectrum allocation optimization}~\label{SAO_eval}
The proposed spectrum allocation optimization method (labeled as SAO) is compared with two baselines. Baseline 1 sets all the devices with equal bandwidth, which means $b_1 = b_2 = ... = b_{N} = \frac{B}{N}$. Baseline 2, so called FEDL~\cite{Dinh21}, jointly optimizes the total energy consumption and the time delay $E + \lambda $, where $\lambda$ is an importance weight.

First, SAO is compared with Baseline 2 with different $\lambda$. We set $p_1 = p_2 = ... = p_{N} = 23\ \text{dBm}$ and $B = 20\ \text{MHz}$. The energy constraints are randomly distributed between $15\ \text{mJ}$ to $30\ \text{mJ}$. The total energy consumption, the energy consumption of each device, and the time delay are shown in Fig.~\ref{eachdevice}. When $\lambda = 1.82$, Baseline 2 makes sure that all the local devices meet their corresponding energy constraints. However, the time delay of Baseline 2 is much higher than that of SAO. Baseline 2 with $\lambda = 4.58$ achieves the same total energy consumption while a higher time delay compared with SAO. Whereas, four local devices exceed the energy constraints. The time delay of Baseline 2 with $\lambda = 1000$ is higher than that of SAO, and the energy constraints are not met. To sum up, Baseline 2 will not provide feasible solutions for the problem~\eqref{X} if $\lambda > 1.82$. Although Baseline 2 can attain feasible solutions when $\lambda \leq 1.82$, these solutions achieve a higher time delay compared with SAO. In contrast, SAO enables FL to achieve the minimum time delay while meeting the energy constraints.

Furthermore, we investigate the completion time under different average transmit power and energy constraints, respectively. In the following experiments, we choose an equal energy constraint $e^{\text{cons}}_1 = ... = e^{\text{cons}}_N = e^{\text{cons}}$. $\lambda$ in Baseline 2 is tuned to make the device with the highest energy cost just meet the energy constraint. In this case, Baseline 2 allows all individual energy costs to be less than or equal to the constraints. Given the energy constraint $e^{\text{cons}}= 30 \ \text{mJ}$, Fig.~\ref{P} depicts how the time delay changes with the increasing of average transmission power. As shown in Fig.~\ref{P}, SAO has the lowest time delay compared with the other two baselines. 

Finally, we set average transmit power of all devices to be a fixed value $23\ \text{dBm}$, and the energy constraints are changed from $30\ \text{mJ}$ to $50\ \text{mJ}$. Fig.~\ref{E} shows the time delay versus different energy constraints. As the energy constraints are relaxed, SAO always achieves the lowest time delay compared with Baseline 1 and Baseline 2.

\renewcommand{\thefigure}{10}
\begin{figure*}[htbp]
\centering
\subfigure[MNIST ($\sigma = 0.5$)]{
\begin{minipage}[t]{0.33\linewidth}
\centering
\includegraphics[width=6cm]{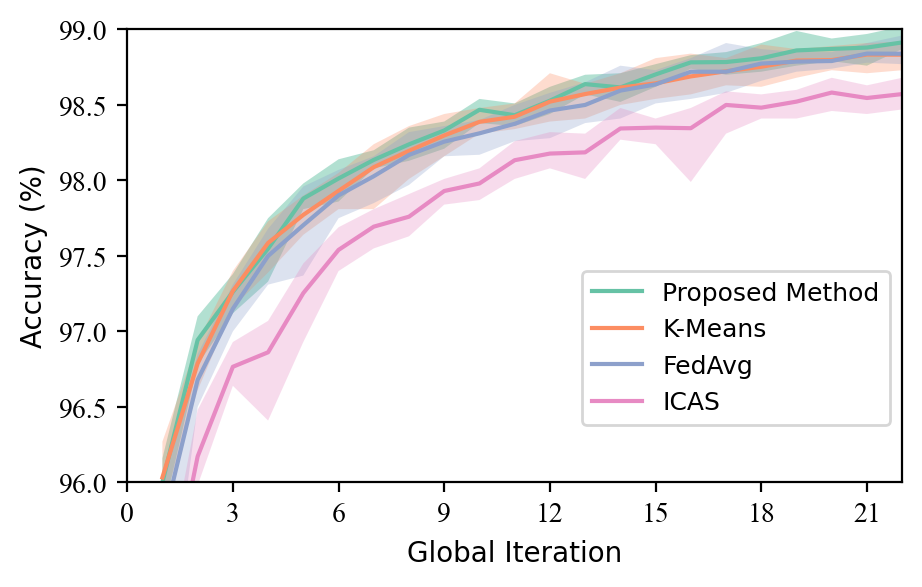}
%\caption{fig1}
\label{Accuracy_11}
\end{minipage}%
}%
\subfigure[MNIST ($\sigma = 0.8$)]{
\begin{minipage}[t]{0.33\linewidth}
\centering
\includegraphics[width=6cm]{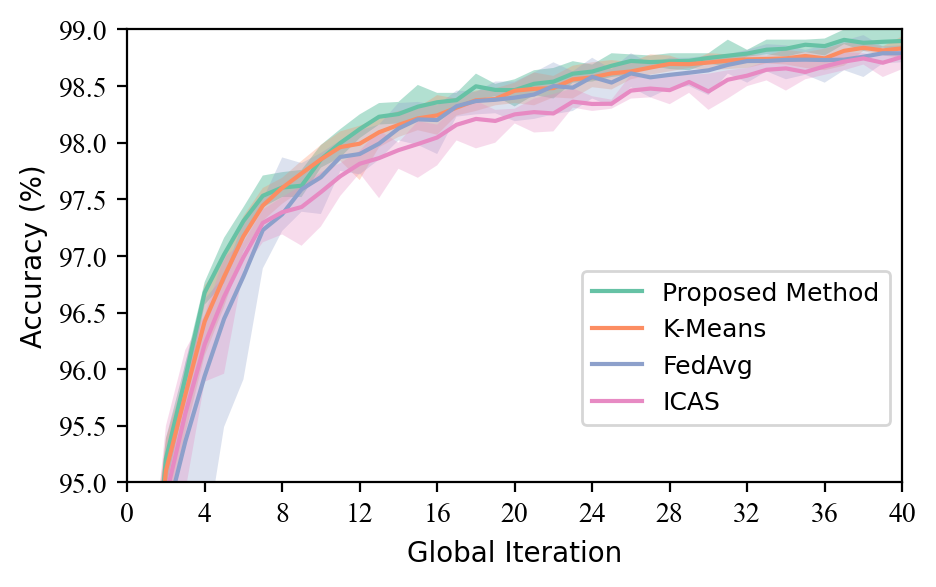}
\label{Accuracy_22}
%\caption{fig2}
\end{minipage}%
}%
\subfigure[MNIST ($\sigma = H$)]{
\begin{minipage}[t]{0.33\linewidth}
\centering
\includegraphics[width=6cm]{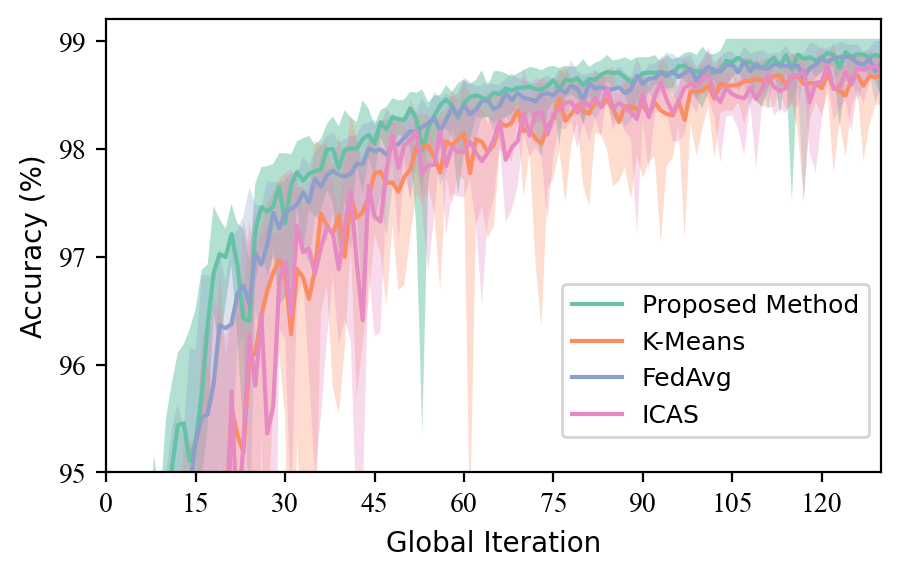}
\label{Accuracy_33}
%\caption{fig2}
\end{minipage}%
}%

\subfigure[FashionMNIST ($\sigma = 0.5$)]{
\begin{minipage}[t]{0.33\linewidth}
\centering
\includegraphics[width=6cm]{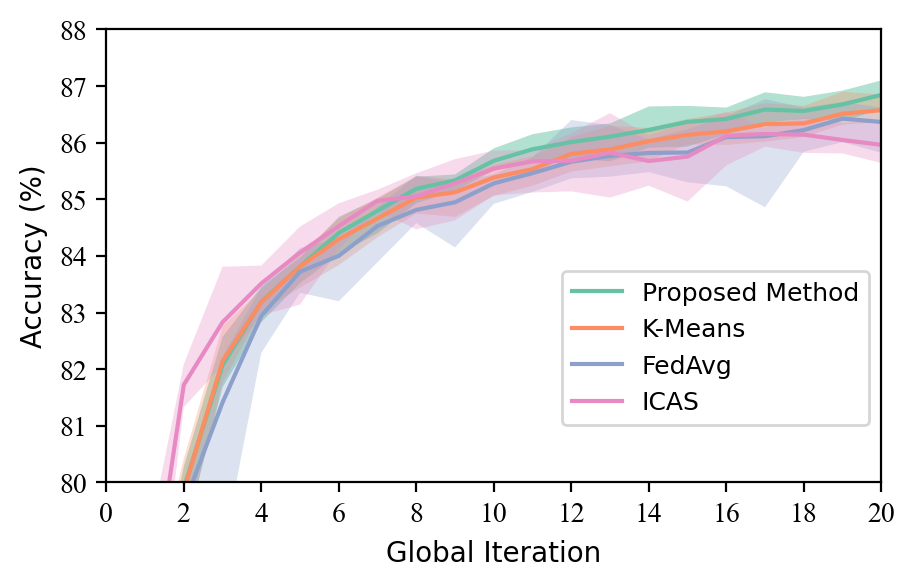}
\label{Accuracy_1}
%\caption{fig1}
\end{minipage}%
}%
\subfigure[FashionMNIST ($\sigma = 0.8$)]{
\begin{minipage}[t]{0.33\linewidth}
\centering
\includegraphics[width=6cm]{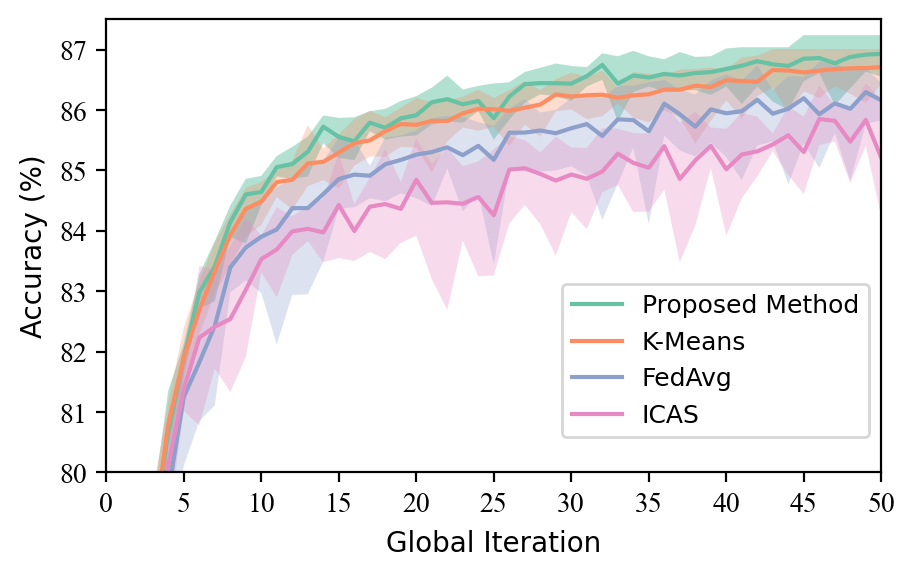}
\label{Accuracy_2}
%\caption{fig2}
\end{minipage}%
}%
\subfigure[FashionMNIST ($\sigma = H$)]{
\begin{minipage}[t]{0.33\linewidth}
\centering
\includegraphics[width=6cm]{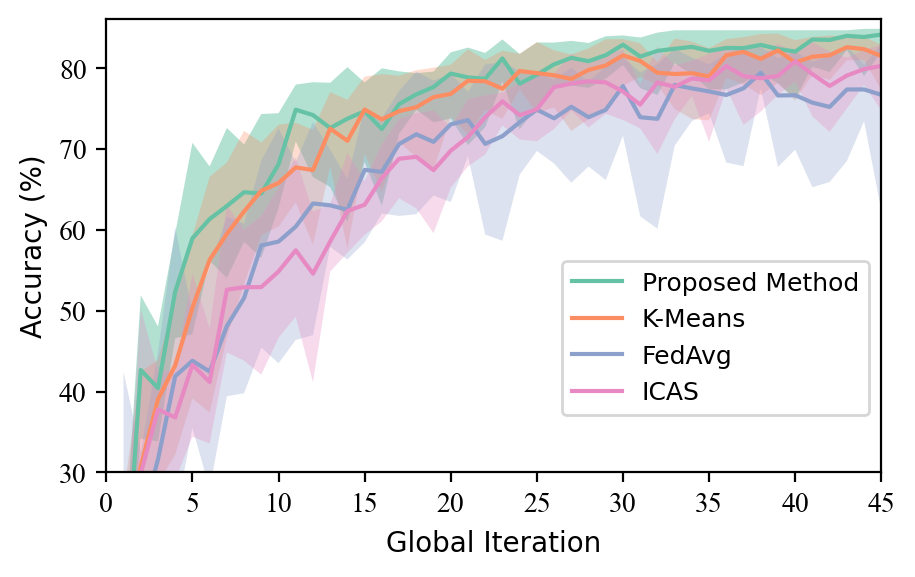}
\label{Accuracy_3}
%\caption{fig2}
\end{minipage}%
}%

\subfigure[CIFAR-10 ($\sigma = 0.5$)]{
\begin{minipage}[t]{0.33\linewidth}
\centering
\includegraphics[width=6cm]{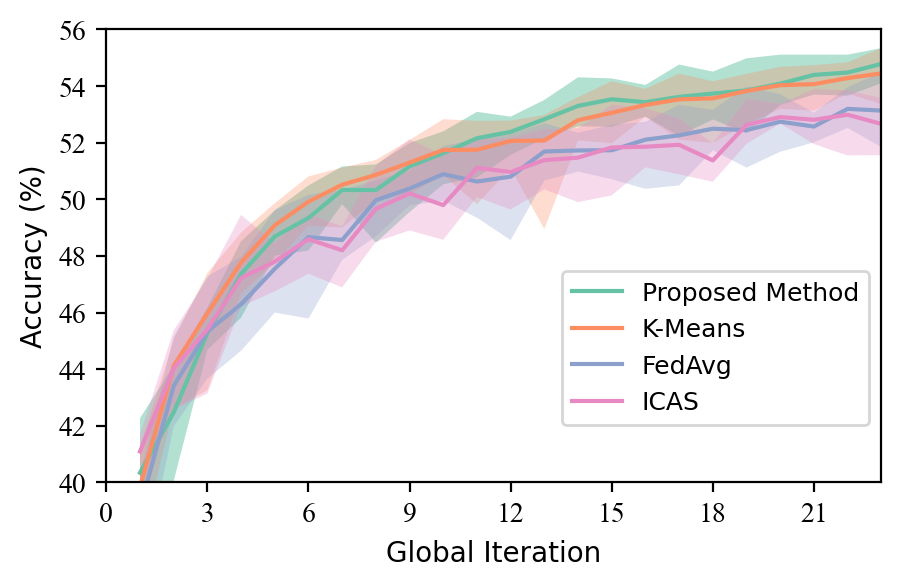}
\label{Accuracy_4}
%\caption{fig1}
\end{minipage}%
}%
\subfigure[CIFAR-10 ($\sigma = 0.8$)]{
\begin{minipage}[t]{0.33\linewidth}
\centering
\includegraphics[width=6cm]{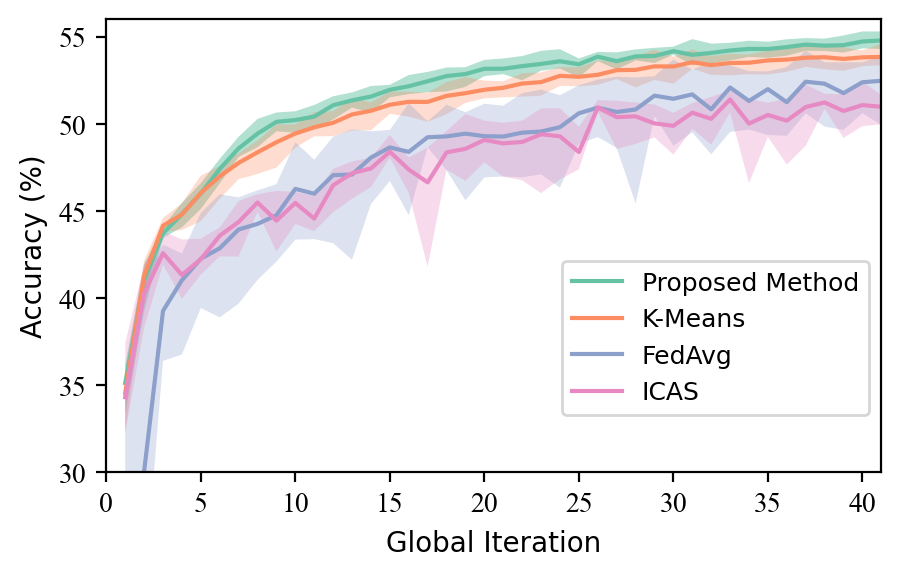}
\label{Accuracy_5}
%\caption{fig2}
\end{minipage}%
}%
\subfigure[CIFAR-10 ($\sigma = H$)]{
\begin{minipage}[t]{0.33\linewidth}
\centering
\includegraphics[width=6cm]{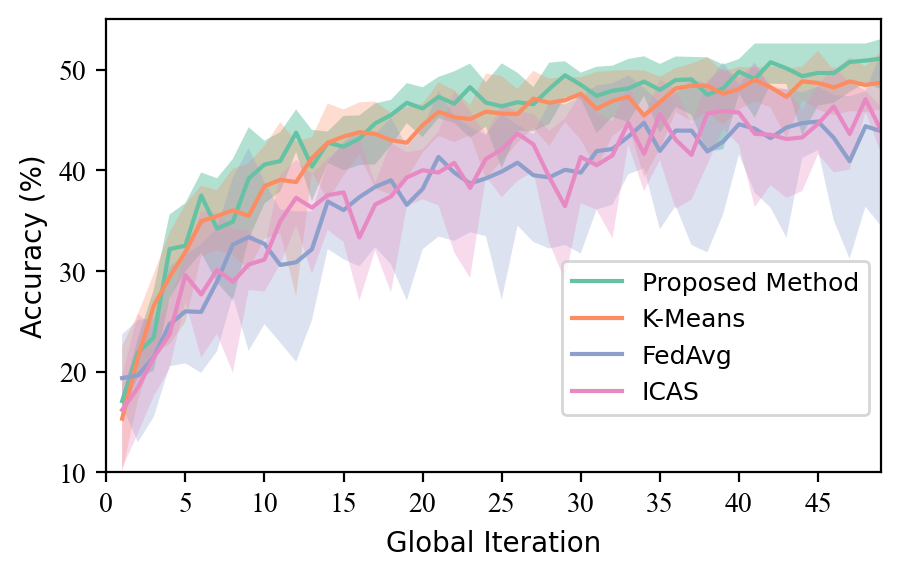}
\label{Accuracy_6}
%\caption{fig2}
\end{minipage}%
}%
\caption{Accuracy on the testing set under different $\sigma\ (S = 10)$.}
\label{Accuracy}
\end{figure*}

\renewcommand{\thefigure}{11}
\begin{figure*}[htbp]
\centering
\subfigure[MNIST ($\sigma = 0.5$)]{
\begin{minipage}[t]{0.33\linewidth}
\centering
\includegraphics[width=6cm]{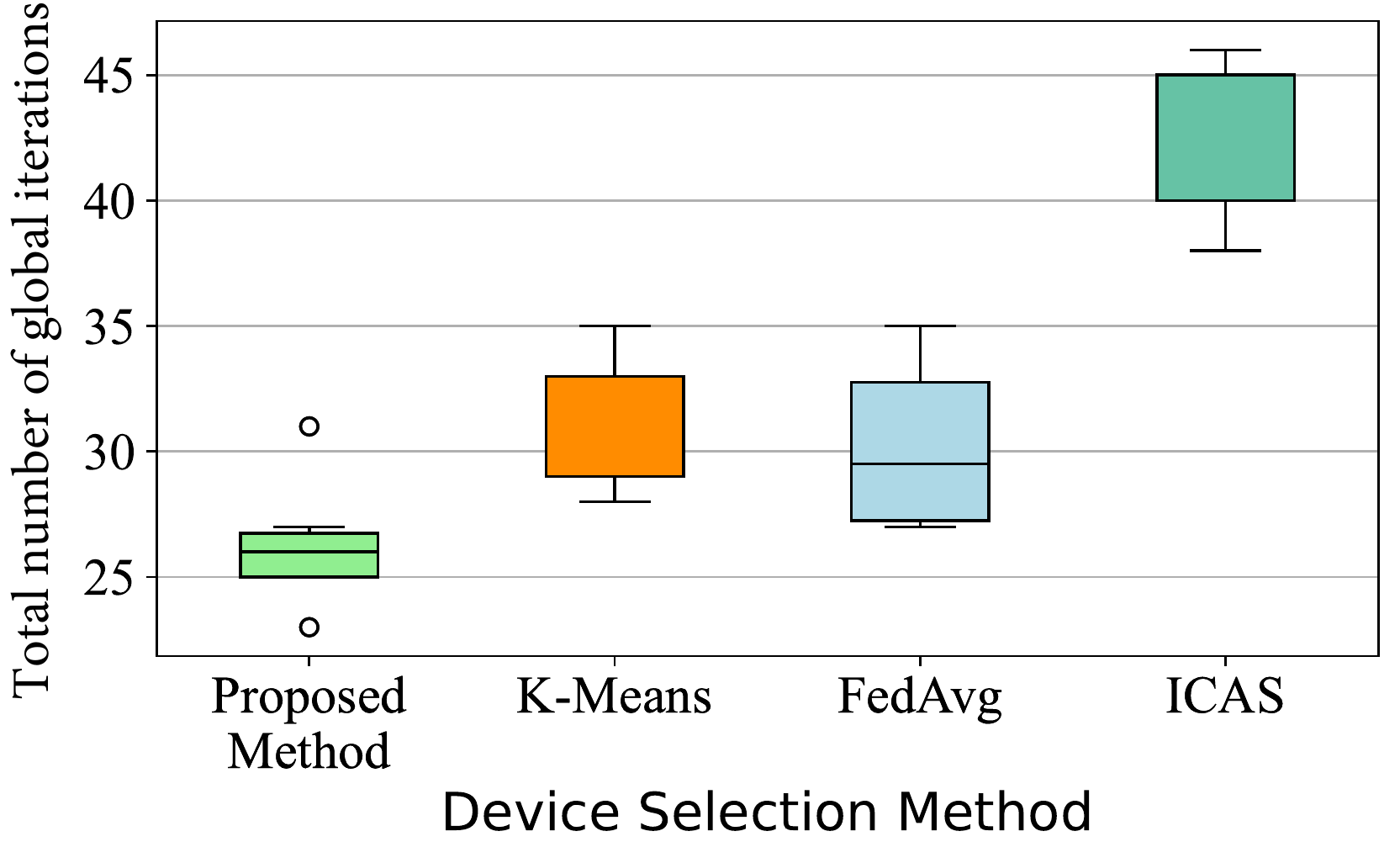}
%\caption{fig1}
\label{confusion_1}
\end{minipage}%
}%
\subfigure[MNIST ($\sigma = 0.8$)]{
\begin{minipage}[t]{0.33\linewidth}
\centering
\includegraphics[width=6cm]{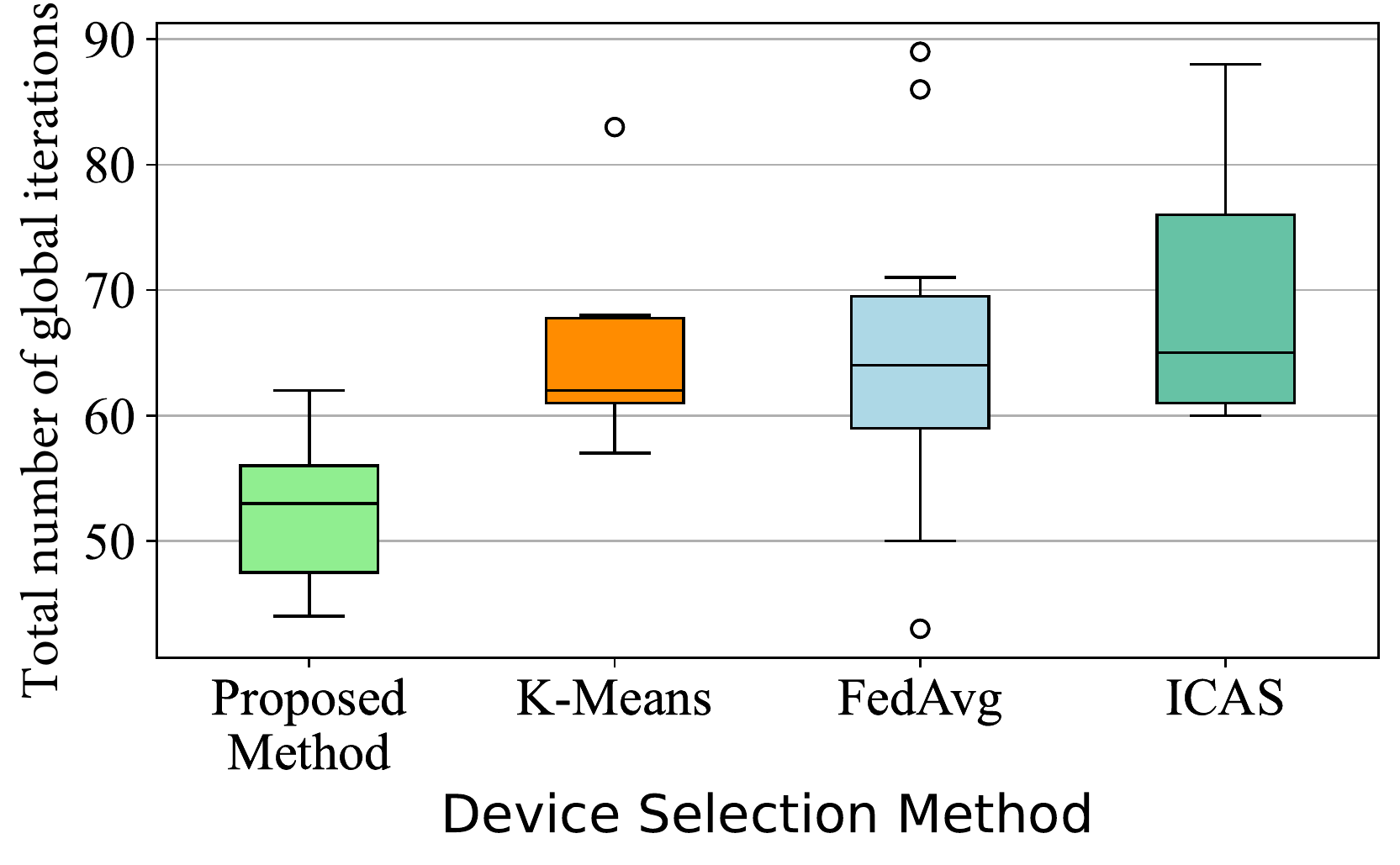}
\label{confusion_2}
%\caption{fig2}
\end{minipage}%
}%
\subfigure[MNIST ($\sigma = H$)]{
\begin{minipage}[t]{0.33\linewidth}
\centering
\includegraphics[width=6cm]{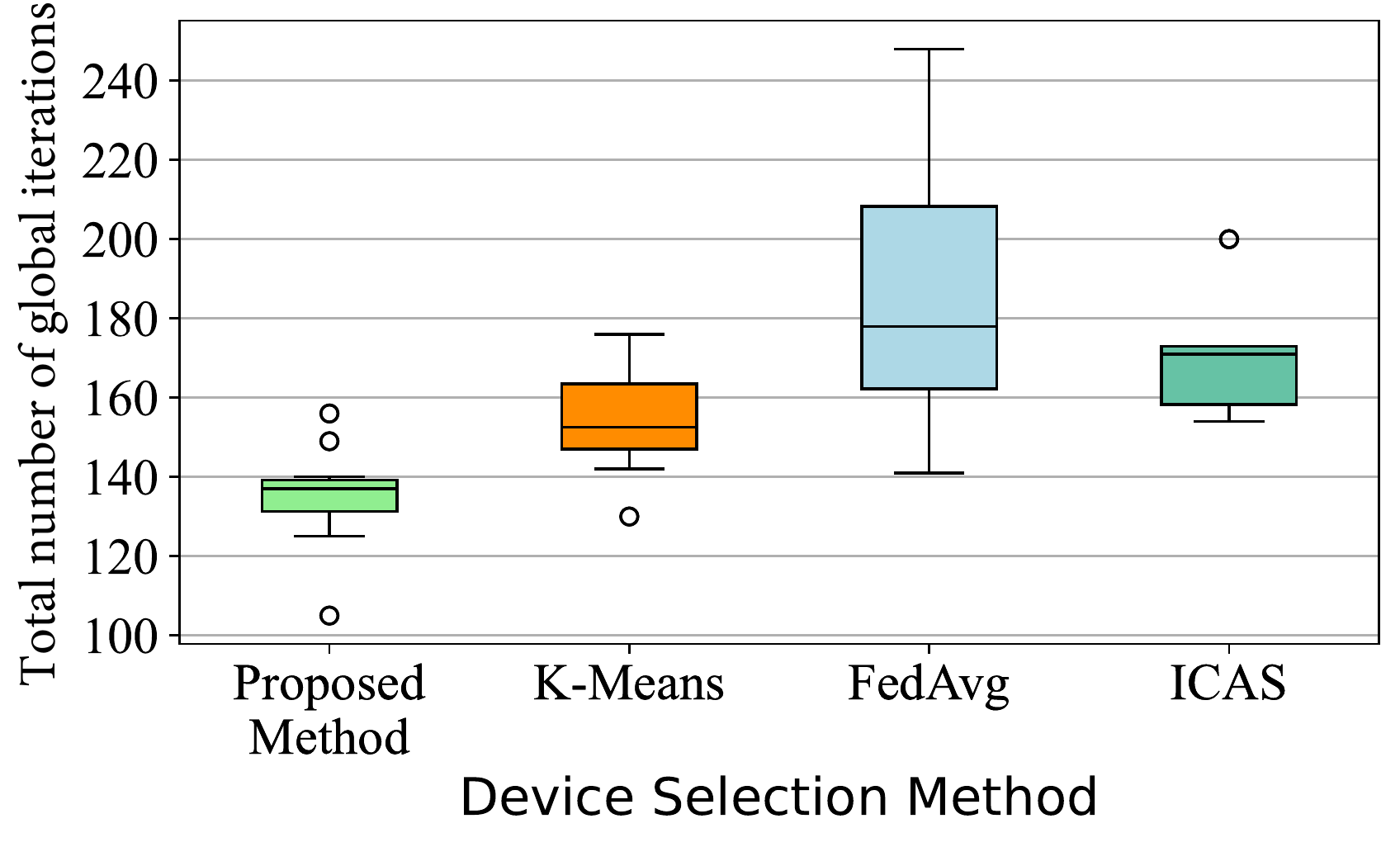}
\label{confusion_3}
%\caption{fig2}
\end{minipage}%
}%

\subfigure[FashionMNIST ($\sigma = 0.5$)]{
\begin{minipage}[t]{0.33\linewidth}
\centering
\includegraphics[width=6cm]{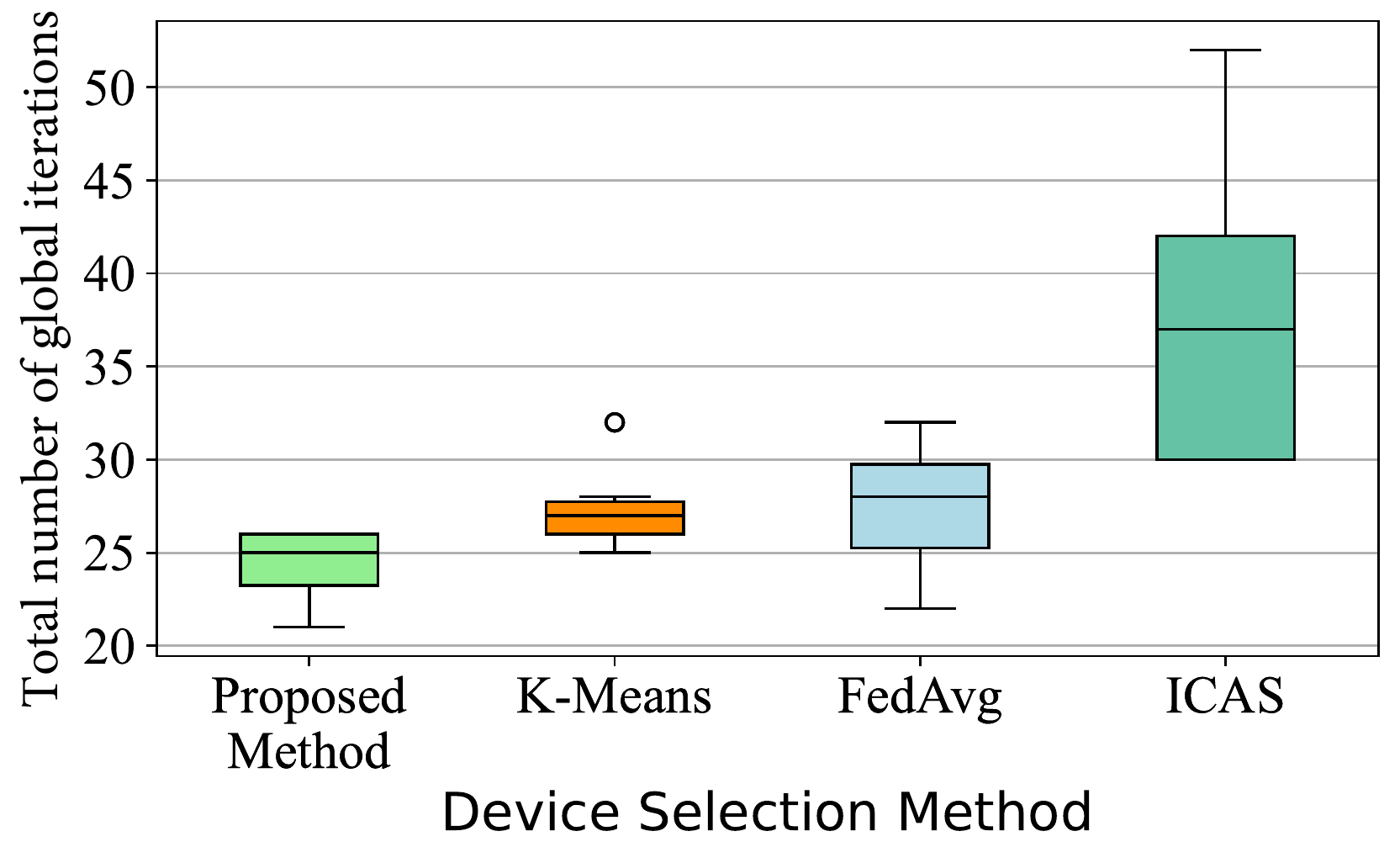}
\label{confusion_4}
%\caption{fig1}
\end{minipage}%
}%
\subfigure[FashionMNIST ($\sigma = 0.8$)]{
\begin{minipage}[t]{0.33\linewidth}
\centering
\includegraphics[width=6cm]{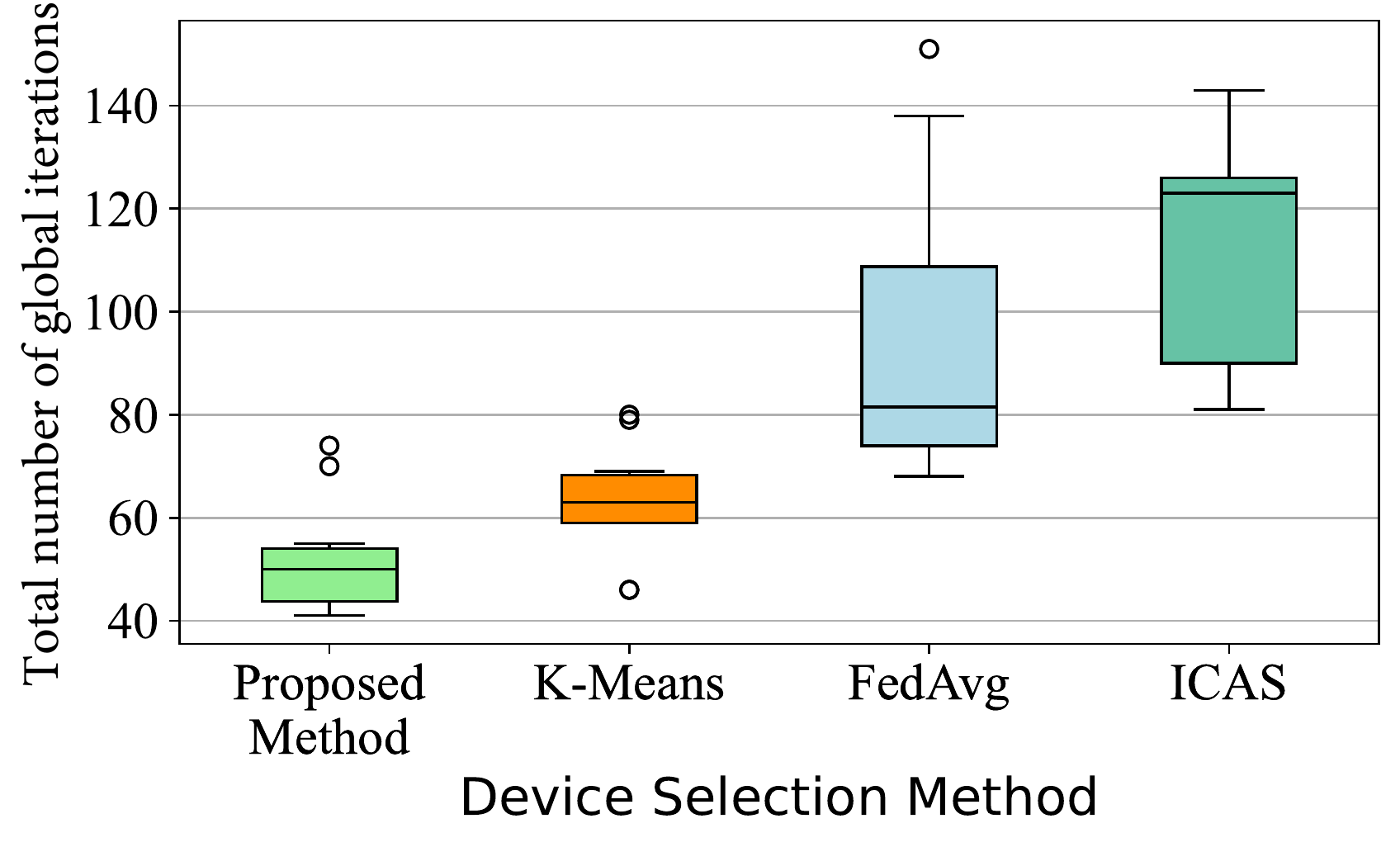}
\label{confusion_5}
%\caption{fig2}
\end{minipage}%
}%
\subfigure[FashionMNIST ($\sigma = H$)]{
\begin{minipage}[t]{0.33\linewidth}
\centering
\includegraphics[width=6cm]{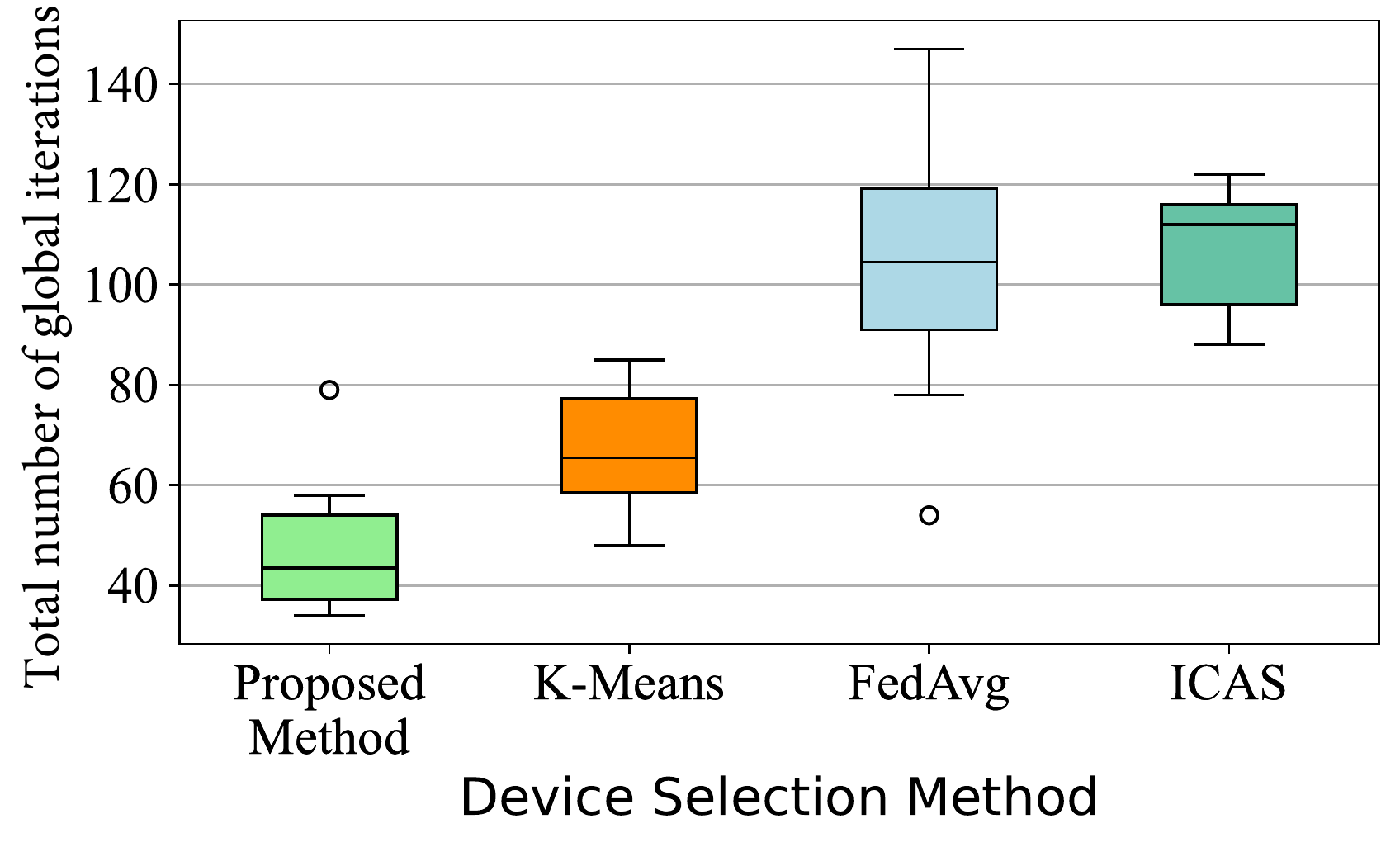}
\label{confusion_6}
%\caption{fig2}
\end{minipage}%
}%

\subfigure[CIFAR-10 ($\sigma = 0.5$)]{
\begin{minipage}[t]{0.33\linewidth}
\centering
\includegraphics[width=6cm]{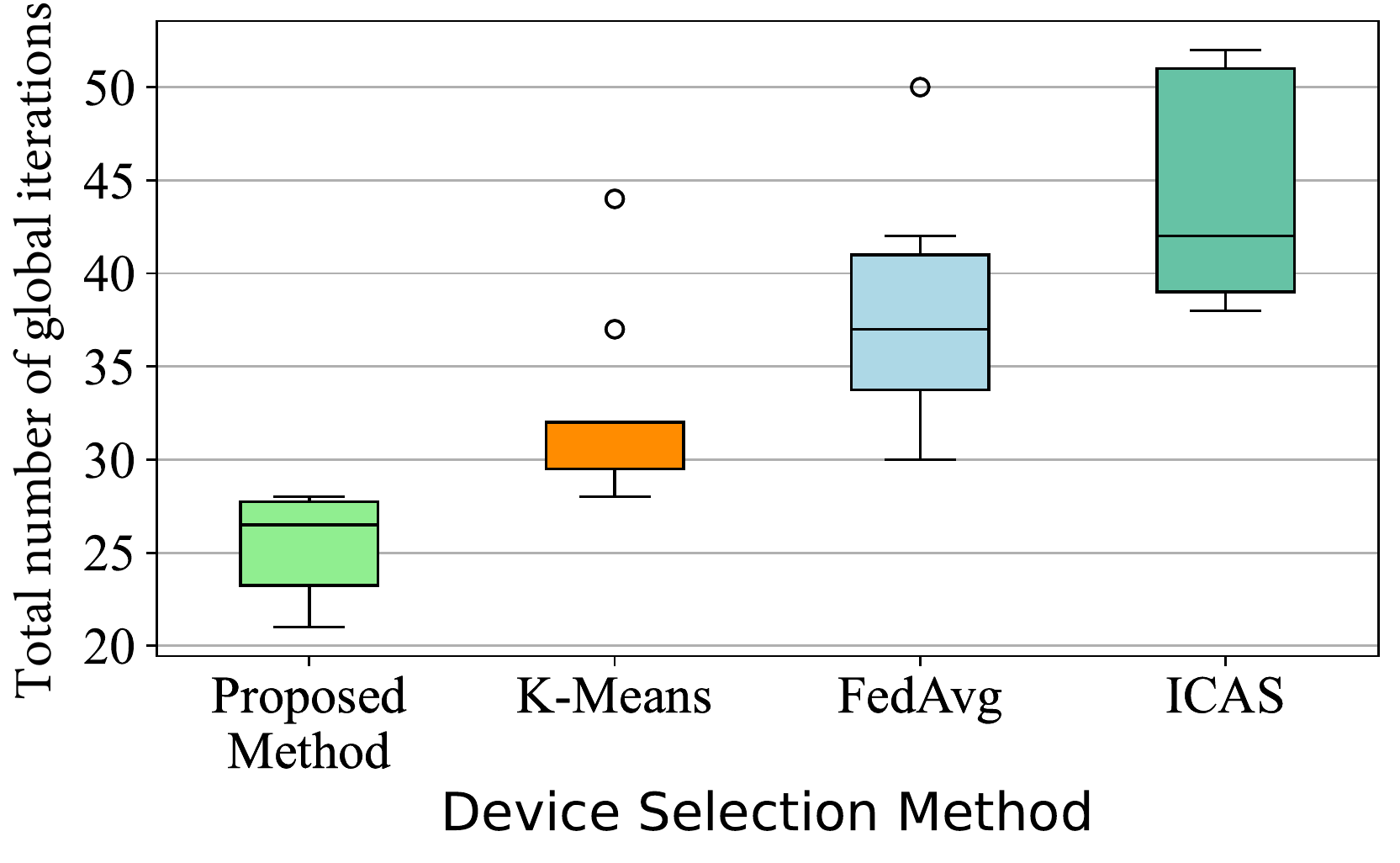}
\label{confusion_7}
%\caption{fig1}
\end{minipage}%
}%
\subfigure[CIFAR-10 ($\sigma = 0.8$). The accuracy of ICAS does not
achieve the target accuracy]{
\begin{minipage}[t]{0.33\linewidth}
\centering
\includegraphics[width=6cm]{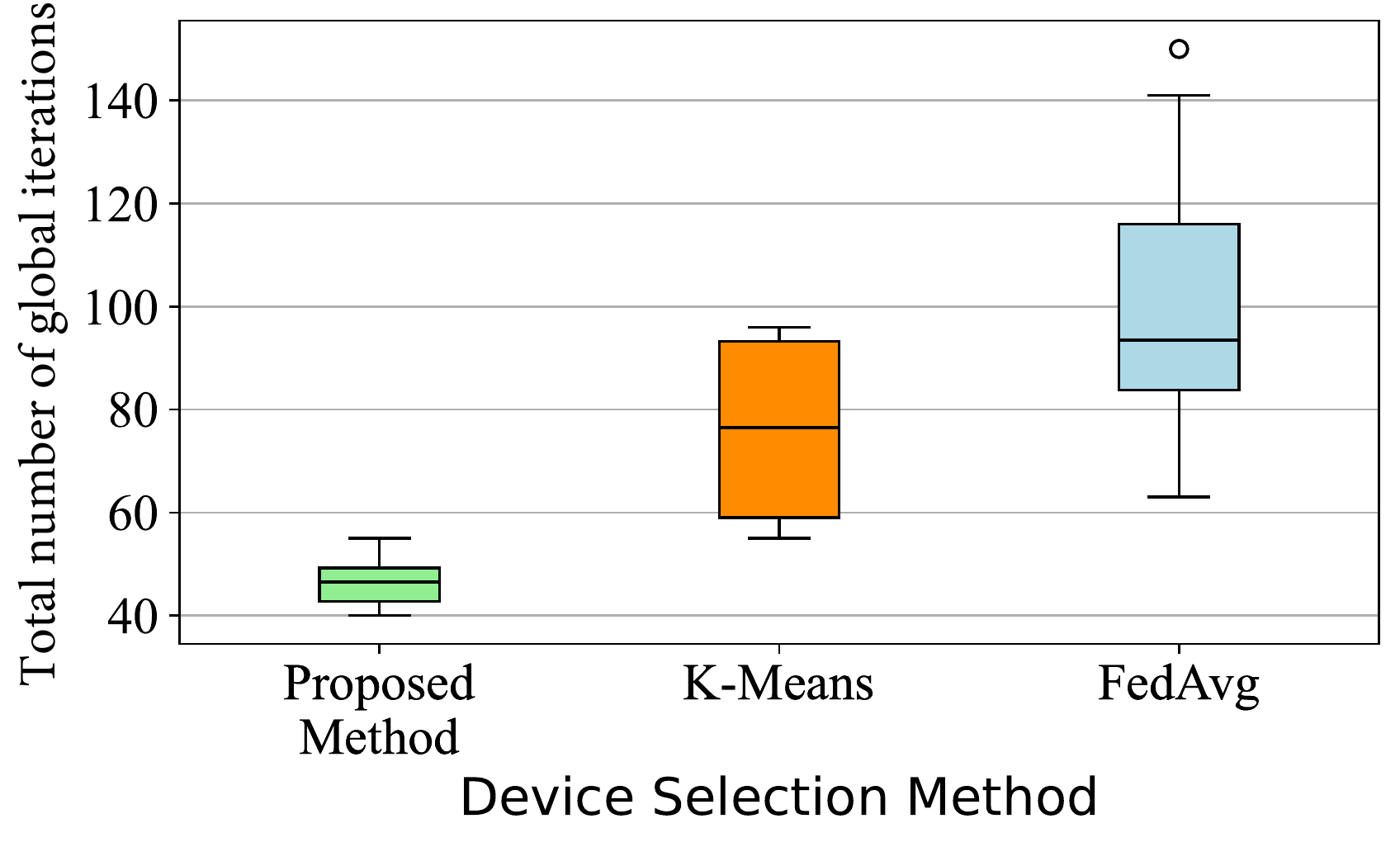}
\label{confusion_8}
%\caption{fig2}
\end{minipage}%
}%
\subfigure[CIFAR-10 ($\sigma = H$)]{
\begin{minipage}[t]{0.33\linewidth}
\centering
\includegraphics[width=6cm]{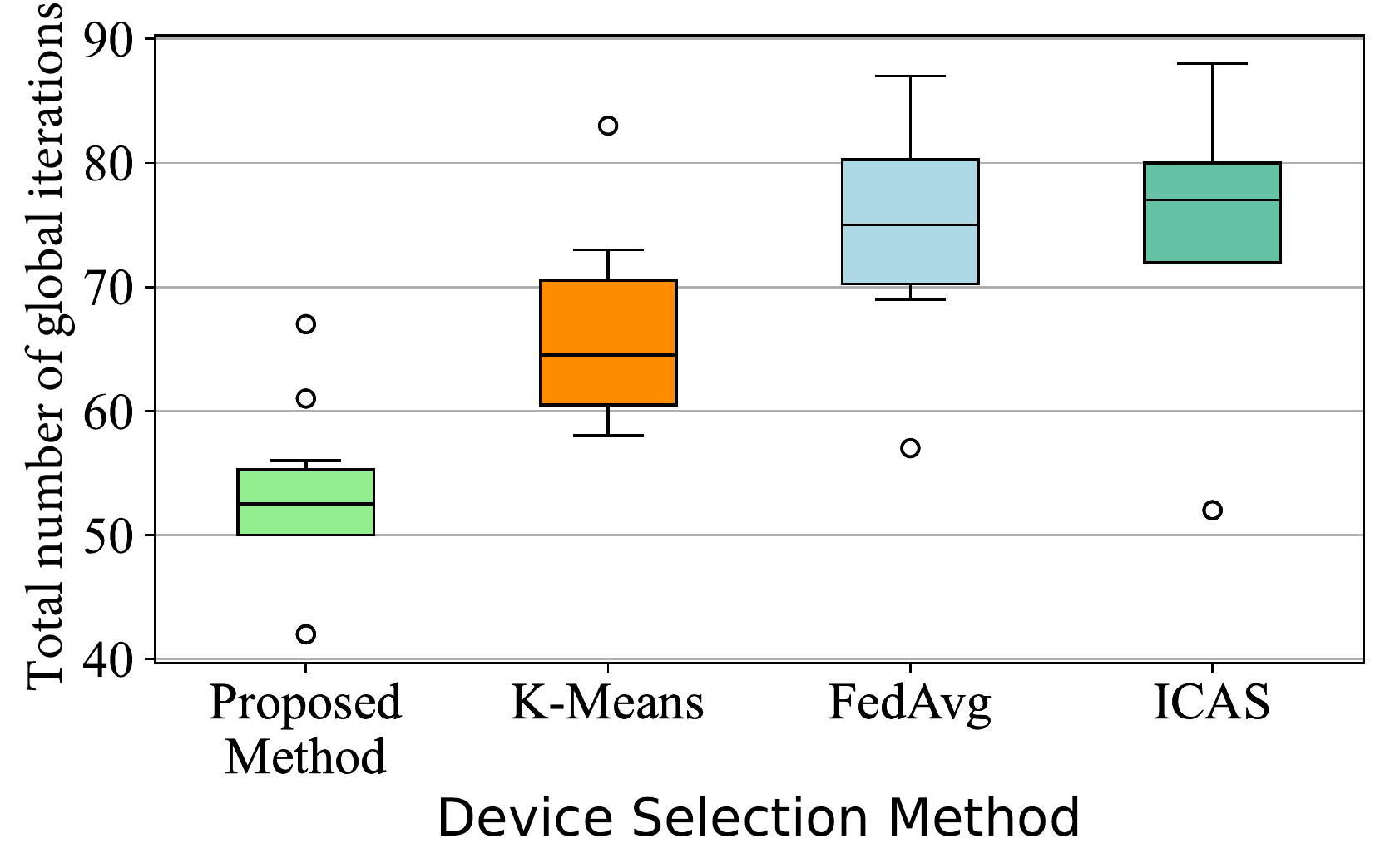}
\label{confusion_9}
%\caption{fig2}
\end{minipage}%
}%
\caption{Total number of global iterations on three datasets under different $\sigma\ (S = 10)$.}
\label{box}
\end{figure*}

\subsection{Evaluation of K-means in device clustering}
The model weights of different layers are used to train the K-means algorithm, respectively. Adjusted Rand index (ARI)~\cite{Lawrence85,Nguyen10} and the computational running time for training the K-means algorithm are two criteria for evaluating the performance of the algorithm. The ARI measures similarity between the predicted cluster labels and the ground truth (i.e., the majority labels of the clusters):
\begin{align}
\begin{split}
ARI(\bm{\mathcal{N}}, \bm{\mathcal{U}}) &= [2(\beta_{00}\beta_{11}-\beta_{01}\beta_{10})]/[(\beta_{00}+\beta_{01})(\beta_{01}+\beta_{11})\\&+(\beta_{00}+\beta_{10})(\beta_{10}+\beta_{11})],
\end{split}%\tag{$22$}
\label{ARI_equa}
\end{align}
where $\bm{\mathcal{N}} = \{\mathcal{N}_1,...,\mathcal{N}_c\}$ is the predicting clustering result, $\bm{\mathcal{U}} = \{\mathcal{U}_1,...,\mathcal{U}_c\}$ is the ground truth, $\beta_{11}$ is the number of pairs that are in the same cluster in both $\bm{\mathcal{N}}$ and $\bm{\mathcal{U}}$, $\beta_{00}$ is the number of pairs that are in different clusters in both $\bm{\mathcal{N}}$ and $\bm{\mathcal{U}}$, $\beta_{01}$ is the number of pairs that are in the same cluster in $\bm{\mathcal{N}}$ but in different clusters in $\bm{\mathcal{U}}$, and $\beta_{10}$ is the number of pairs that are in different clusters in $\bm{\mathcal{N}}$ but in the same cluster in $\bm{\mathcal{U}}$. The ARI will be close to 0 if two data clusters do not have any overlapped pair of samples and exactly 1 if the clusters are the same. 

Fig.~\ref{T_Kmeans} provides the computational running time of training the K-means algorithm with the model weights of different layers. According to Table~\ref{numberweight} and Fig.~\ref{T_Kmeans}, the number of model weights directly affect the training time. The highest time delay is obtained when all the model weights are used for training the K-means algorithm. In contrast, the time delay becomes significantly short when the K-means algorithm is trained by low-dimensional model weights or biases. 

Fig.~\ref{ARI} shows the ARI of the K-means algorithm on three datasets. The K-means algorithm achieves the highest overall ARI under $\sigma=0.8$ while obtaining the smallest overall ARI under $\sigma=H$. This is because a high $\sigma$ indicates that most of the samples in the local dataset belong to the majority class. As a result, the model weights trained on two local datasets with different majority classes are distinguishable, thus resulting in better cluster performance. In terms of $\sigma=H$, the secondary classes cause a lower ARI, since some local devices are grouped into the clusters of which labels are the secondary classes of those local datasets.

From Fig.~\ref{T_Kmeans} and Fig.~\ref{ARI}, it can be seen that the K-means algorithm trained by $\emph{w}^{fc_2}$ achieves both short training time and high ARI no matter how $\sigma$ changes. Therefore, $\emph{w}^{fc_2}$ is adopted as the feature vector for training the K-means algorithm.

\renewcommand{\thefigure}{12}
\begin{figure*}[t]
\centering
\begin{minipage}[t]{0.333\linewidth}
\centering
\includegraphics[width=6cm]{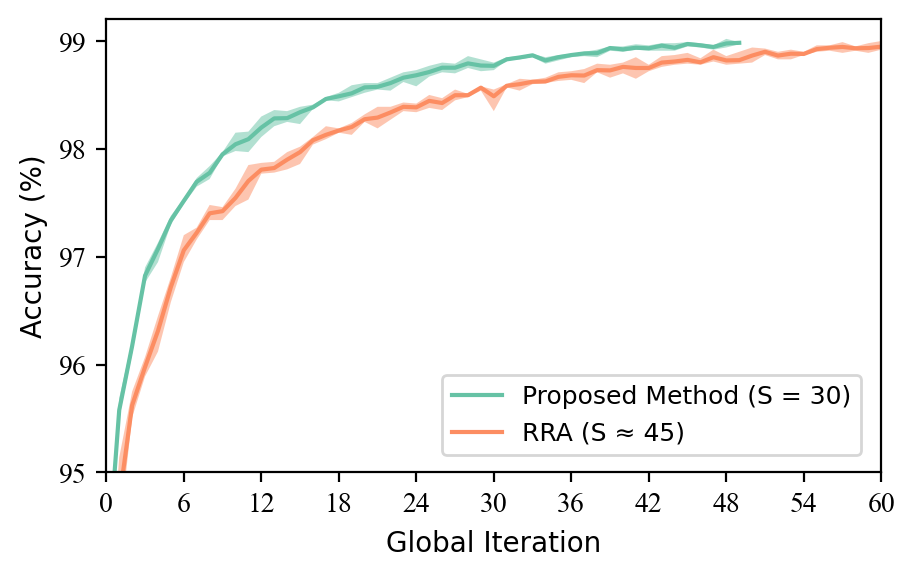}
\end{minipage}%
\begin{minipage}[t]{0.333\linewidth}
\centering
\includegraphics[width=6cm]{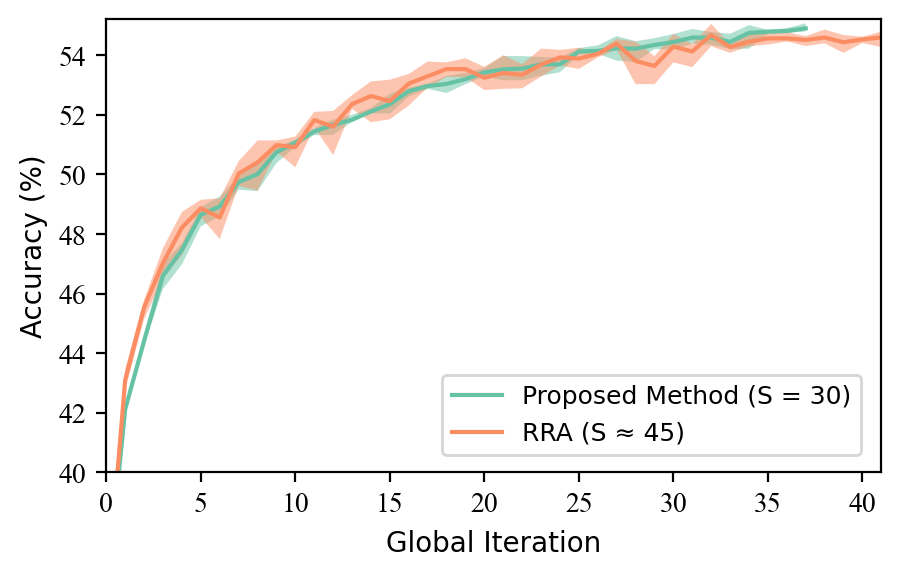}
\end{minipage}%
\begin{minipage}[t]{0.333\linewidth}
\centering
\includegraphics[width=6cm]{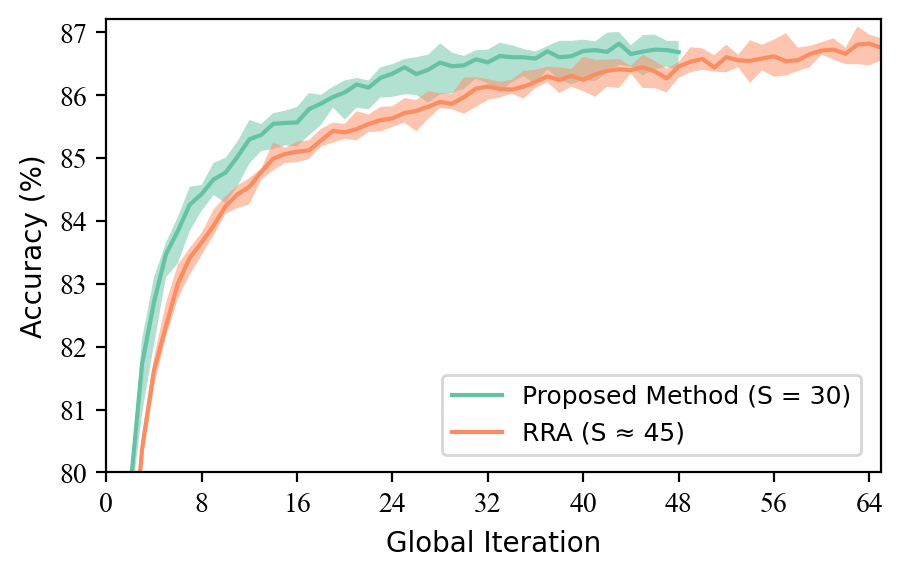}
\end{minipage}%

\begin{minipage}[t]{0.333\linewidth}
\centering
\includegraphics[width=6cm]{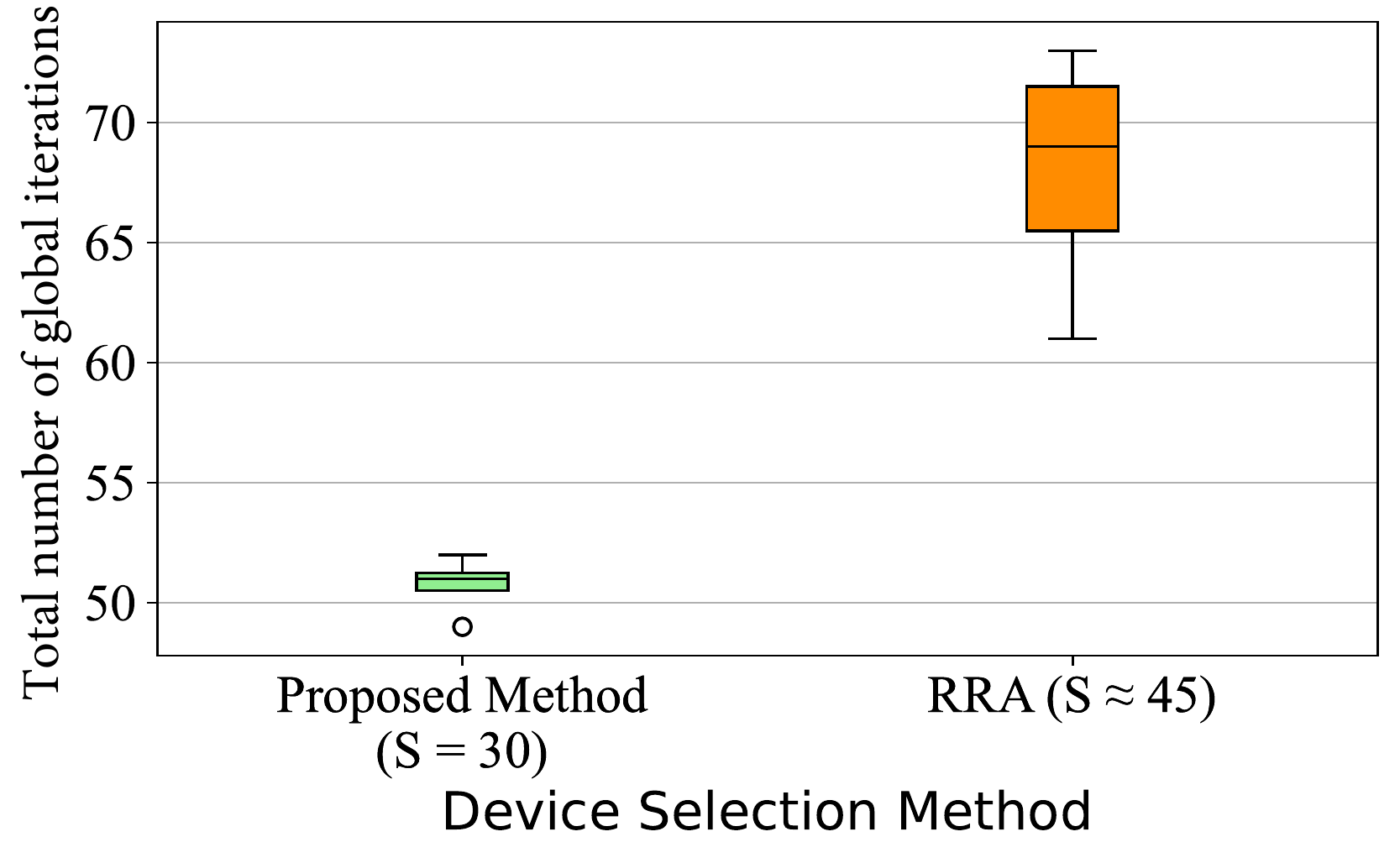}
\label{radio_1}
\caption*{(a) MNIST}
\end{minipage}%
\begin{minipage}[t]{0.333\linewidth}
\centering
\includegraphics[width=6cm]{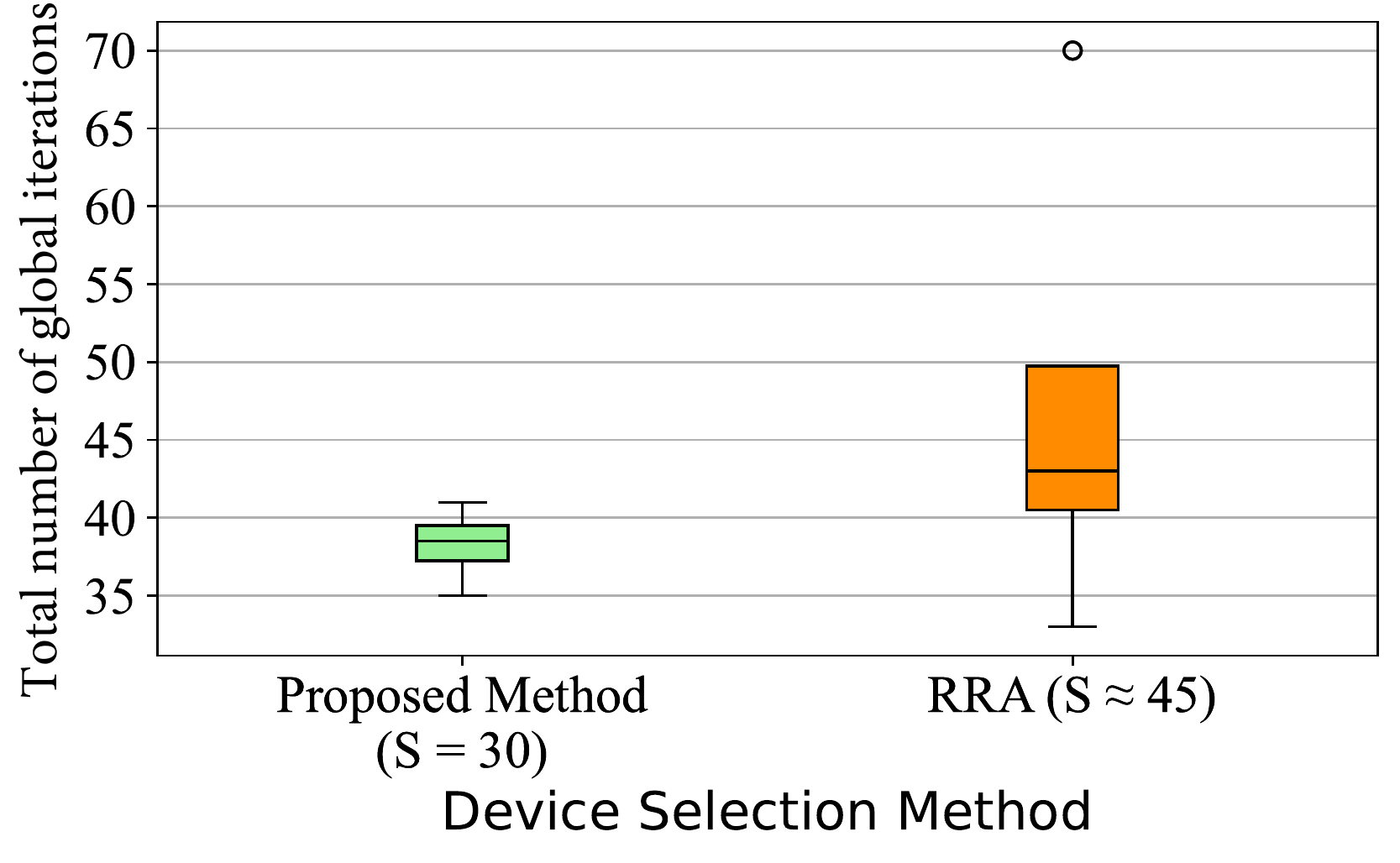}
\label{radio_2}
\caption*{(b) CIFAR-10}
\end{minipage}%
\begin{minipage}[t]{0.333\linewidth}
\centering
\includegraphics[width=6cm]{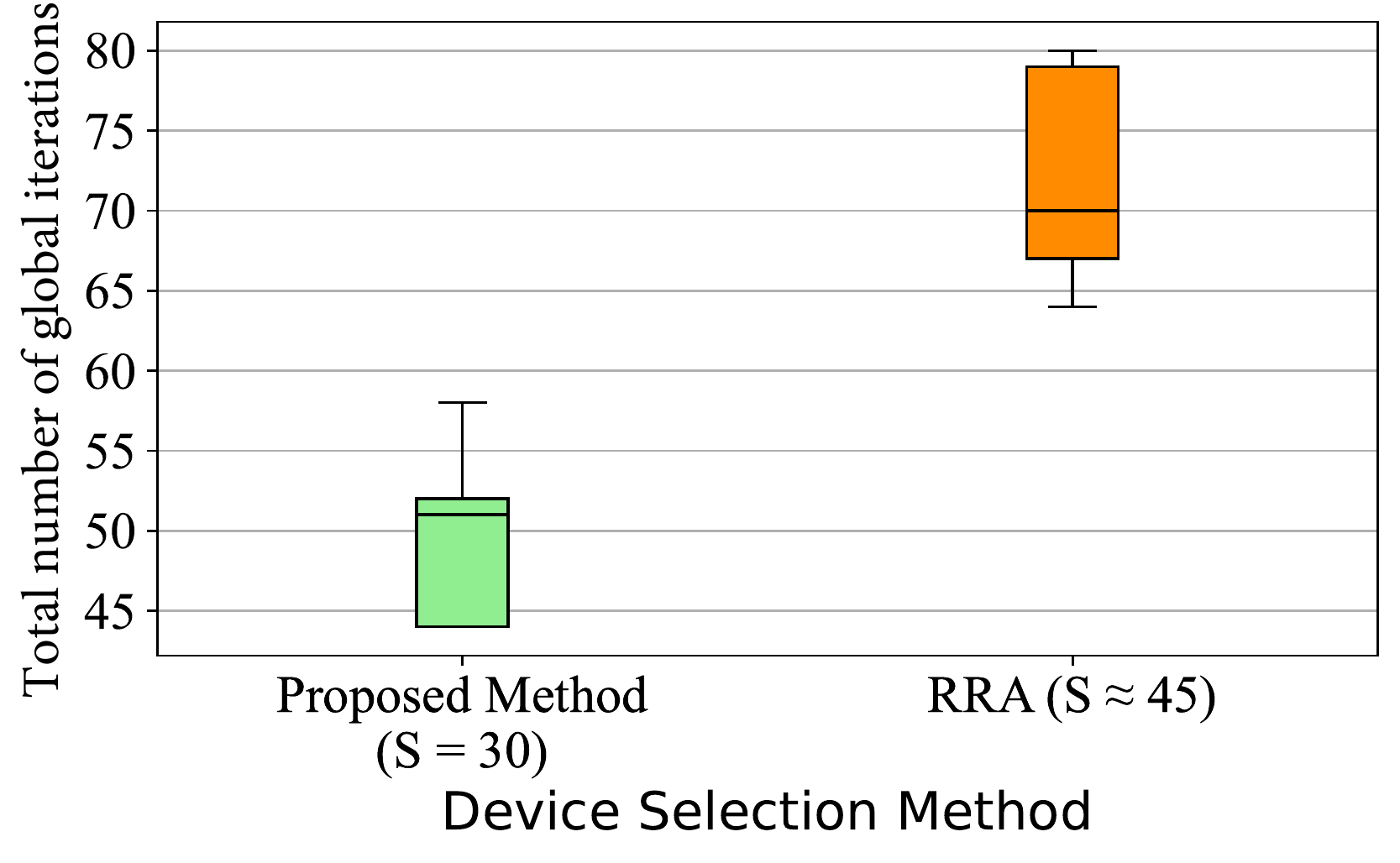}
\label{radio_3}
\caption*{(c) FashionMNIST}
\end{minipage}%
\caption{Accuracy and total number of global iterations  on three datasets under $\sigma = 0.8$.}
\label{radio}
\end{figure*}

\subsection{Evaluation of the device selection method}
At the first global iteration, all the local devices participate in the local update and global aggregation. Then, the weight divergence based device selection method is adopted to choose one local device from each cluster for the following iteration. First, we implement FedAvg~\cite{McMahan17}, the K-means method, and ICAS~\cite{Ren20} to perform device selection for comparison. FedAvg chooses ten devices randomly at each global iteration, while the K-means method randomly chooses one device from each cluster. ICAS conducts device selection based on the importance of the local learning update. We repeat the FL training ten times under $S = 10$ and then evaluate the convergence performance of four device selection methods. 

The accuracy curves on the testing set are depicted in Fig.~\ref{Accuracy}. The solid curves represent the mean, and the shaded regions correspond to the minimum and maximum returns over the ten trials. In most cases, the accuracy curves of the proposed method and the K-means method are above the curve of FedAvg. This is because the proposed method and the K-means method divide the local devices into several clusters based on the majority class. By choosing the local devices from each cluster, the bias incurred by non-iid datasets is significantly alleviated. Besides, we observe that the accuracy curve of the proposed method increases fastest among these device selection methods. This experimental result shows that the local devices chosen by the proposed method enable the FL model to reach convergence swiftly. Furthermore, the shaded regions of our method are the smallest, which demonstrates the stability of our method. Fig.~\ref{box} provides the total number of global iterations of FL under different device selection methods. Note that Fig.~\ref{confusion_8} does not include ICAS. This is because given the target accuracy 55\%, the accuracy of FL is only about 53\% when FL reaches convergence. Thus, we cannot obtain the total number of global iterations of ICAS. First, the confidence intervals of the proposed method and the K-means method are smaller than that of FedAvg. It further demonstrates that device clustering speeds up the FL training process. Secondly, the medians values of the proposed method are smaller than the two baselines. Thirdly, ICAS does not obtain good performance because the mechanism of ICAS is on the basis of the assumption: the loss function is Lipschitz continuous and strongly convex. Nevertheless, our paper adopts a non-convex function, hence degrading the performance of ICAS. The experimental results prove the effectiveness and feasibility of the weight divergence based device selection method. 
\begin{table*}[htbp]
\centering
\caption{Improvement scores over FedAvg}
\renewcommand\arraystretch{1.1}
\setlength{\tabcolsep}{3mm}{
\begin{tabular}{ccccccc}\toprule
\multirow{2}{*}{$\sigma$} & \multicolumn{2}{c}{MNIST}           & \multicolumn{2}{c}{FashionMNIST} & \multicolumn{2}{c}{CIFAR-10}   \\
                        & Favor~\cite{Wang20}            & Proposed Method  & Favor~\cite{Wang20}    & Proposed Method  & Favor~\cite{Wang20}    & Proposed Method  \\\midrule
0.5                     & 0.150         & \textbf{0.228} & 0.181 & \textbf{0.810} & 0.736 & \textbf{1.133} \\
0.8                     & \textbf{0.955} & 0.157         & 0.209 & \textbf{0.232} & 0.426 & \textbf{0.486} \\
H                    & 0                & \textbf{0.388}  & 0.187   & \textbf{1.204} & 0.340     & \textbf{0.641}\\ \bottomrule 
\end{tabular}}
\label{improve}
\end{table*}

\renewcommand{\thefigure}{13}
\begin{figure*}[t]
\centering
\begin{minipage}[t]{0.333\linewidth}
\centering
\includegraphics[width=6cm]{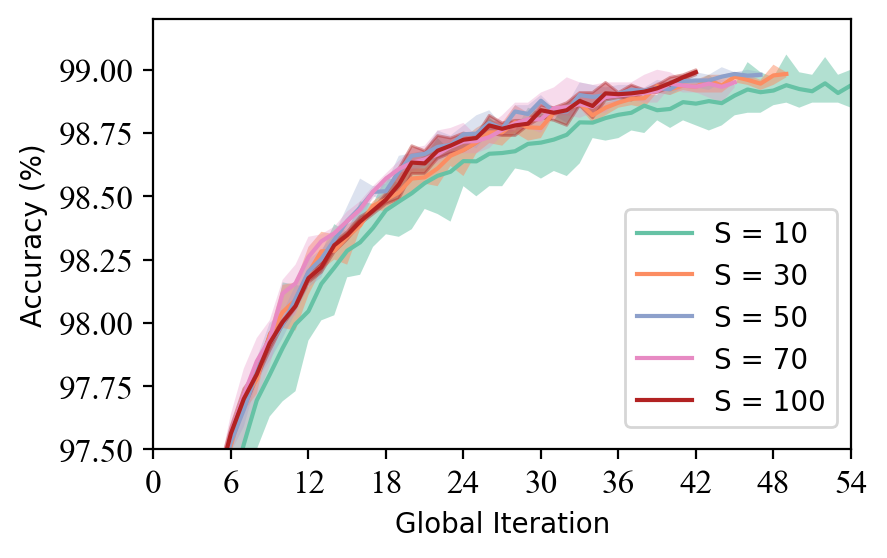}
\end{minipage}%
\begin{minipage}[t]{0.333\linewidth}
\centering
\includegraphics[width=6cm]{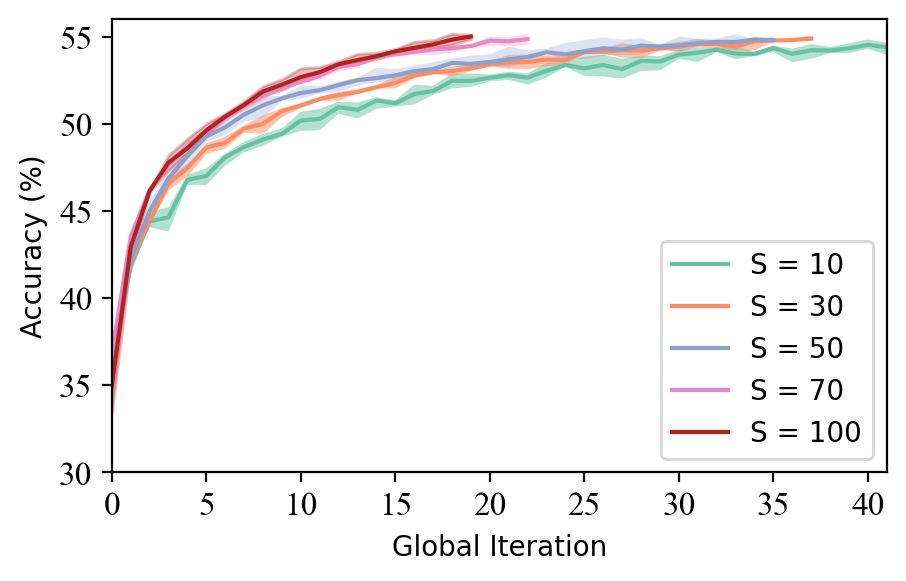}
\end{minipage}%
\begin{minipage}[t]{0.333\linewidth}
\centering
\includegraphics[width=6cm]{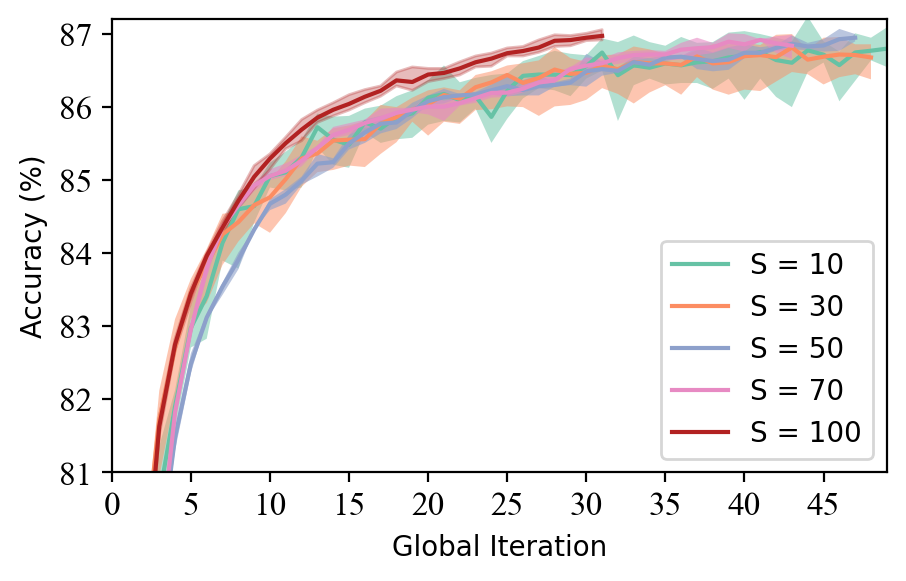}
\end{minipage}%

\begin{minipage}[t]{0.333\linewidth}
\centering
\includegraphics[width=6cm]{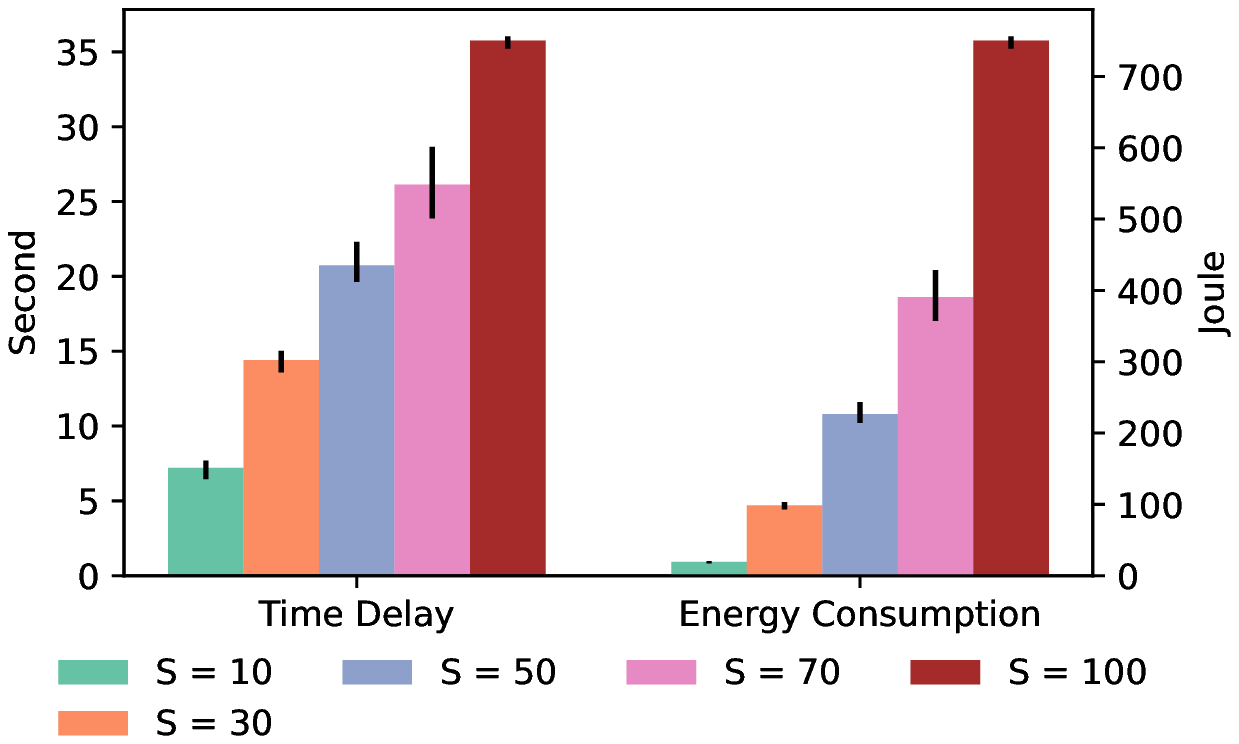}
\label{diffs_mnist}
\caption*{(a) MNIST}
\end{minipage}%
\begin{minipage}[t]{0.333\linewidth}
\centering
\includegraphics[width=6cm]{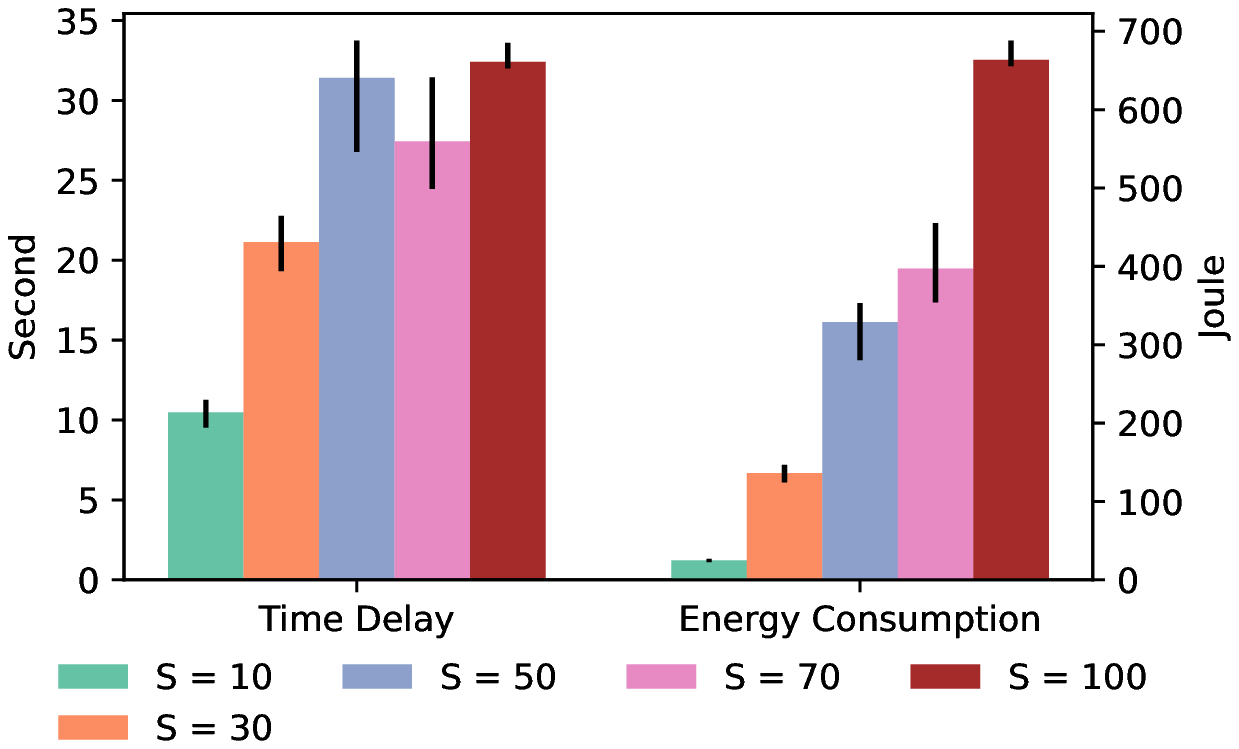}
\label{diffs_cifar}
\caption*{(b) CIFAR-10}
\end{minipage}%
\begin{minipage}[t]{0.333\linewidth}
\centering
\includegraphics[width=6cm]{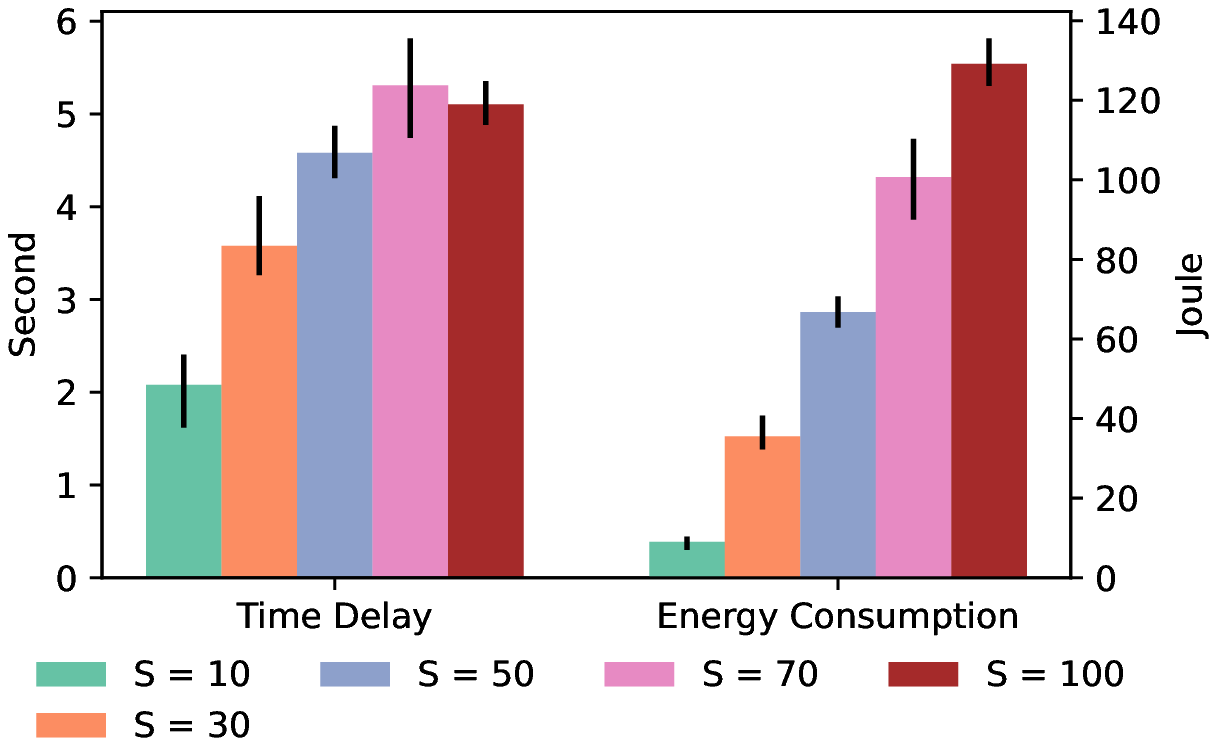}
\label{diffs_fashion}
\caption*{(c) FashionMNIST}
\end{minipage}%
\caption{Accuracy, time delay $T$, and total energy consumption $E$ of FL under different $S$.}
\label{diffs_all}
\end{figure*}

We also compare the proposed method with RRA~\cite{Zeng20}. RRA aims to reduce the total energy consumption while ensuring learning performance. For comparison, FL is trained using the proposed device selection method with $S = 30$ and RRA under $\sigma = 0.8$, respectively. Note that the value of $S_k$ in RRA varies at each iteration, and the average number of the selected devices is 45. Fig.~\ref{radio} depicts the accuracy and total number of global iterations of FL. It can be observed that our method enables FL to reach faster convergence with fewer devices. Besides, our method is more flexible than RRA as $S$ of our method is controllable.

In addition to the above baseline methods, a state-of-art device selection method, Favor~\cite{Wang20}, is used to compare with our method. Favor applies the reinforcement learning algorithm to select the optimal devices for training FL. However, the CNN architectures of the local models, the local datasets, and other parameter settings in Favor are different from ours. Therefore, we cannot directly cite the total number of global iterations of Favor into our work for comparison. As FedAvg was also implemented in Favor as a baseline, the performance of Favor and the proposed method can be compared by evaluating how much these methods improve the FL training over FedAvg. To this end, we introduce an improvement score as follows.
\begin{equation}
    score = \frac{R_\text{eval}}{R_\text{FedAvg}}-1,
\label{score}
\end{equation}
where $R_\text{eval}$ is the total number of global iterations of Favor or the proposed method, and $R_\text{FedAvg}$ is the corresponding total number of global iterations of FedAvg in each paper. In this work, the median values of total number of global iterations in Fig.~\ref{box} are used as $R_\text{eval}$. The improvement scores are shown in Table~\ref{improve}. We observe that the proposed method achieves higher improvement scores except on MNIST under $\sigma = 0.8$. Besides, the well-trained reinforcement learning models in Favor are no longer effective when dealing with new datasets. In contrast, the proposed method is flexible and can be directly employed even though the datasets are changed. In conclusion, the experimental results demonstrate that the weight divergence based device selection method outperforms Favor.

\subsection{Interplay between spectral allocation optimization and device selection}
To analyze the interplay between SAO and weight divergence based device selection, we train the FL models with different $S$ based on Fig.~\ref{WholeFramework}. Given the non-iid datasets with $\sigma = 0.8$, the accuracy curves on the testing sets are shown in Fig.~\ref{diffs_all}. Besides, the time delay $T$ and the total energy consumption $E$ are respectively calculated based on~\eqref{T_equa} and~\eqref{E_equa}, which are also depicted in Fig.~\ref{diffs_all}.

From Fig.~\ref{diffs_all} and the problem~\eqref{X}, we can observe that the number of the selected devices (i.e., $S$) affects $T$ and $E$ from two aspects. On the one hand, increasing $S$ enables FL to learn more knowledge at one global iteration, thus leading to fast convergence (i.e., a smaller $K$). Without device selection (i.e., $S = 100$), FL reaches convergence with the smallest $K$. On the other hand, a larger $S$ results in higher $E_k$ and $T_k$ at each global iteration due to the limited bandwidth resource. According to Fig.~\ref{diffs_all}, it can be noted that a smaller $S$ usually allows FL to achieve lower $E$ and $T$. Besides, the optimal $S$ for solving the problem~\eqref{A} is 10, where FL always attains the lowest $T$.

\section{Conclusions}~\label{section7}
In this paper, we studied the spectrum allocation optimization and device selection problems for enhancing FL over a wireless mobile network. We formulated the computation and communication models to calculate the energy consumption and time delay of FL. We formulated the optimization problem to minimize the time delay of training the entire FL algorithm given the individual energy constraint and the total bandwidth. We proposed an energy-efficient spectrum allocation optimization method and a weight divergence based device selection method to solve this problem. The proposed spectrum allocation optimization method aims to minimize the time delay under one global iteration while warranting the energy constraints. For device selection, we conducted device clustering by training the K-means algorithm with part of model weights in CNNs. We have proposed the weight divergence based device selection method to select the devices from each cluster. According to our experimental results, the proposed spectrum allocation optimization method optimizes the time delay of FL while satisfying the energy constraints; the proposed device clustering method realizes a fast training process and high clustering performance; the proposed device selection method helps FL to achieve convergence swiftly and outperforms other baselines.

\appendices
\section{Proof of LEMMA 1}
\begin{proof}
The objective function of the problem (19) is a convex function. Let $I_n(x) = \frac{1}{x\text{log}_2(1+\frac{J_n}{x})}$, where $x > 0$ and $\forall n \in \mathcal{S}_k$. Note that $x\text{log}_2(1+\frac{y}{x})$ is concave in $(x,y)$ as $\text{log}_2(1+y)$ is concave. Therefore, $x\text{log}_2(1+\frac{J_n}{x})$ is positive, increasing, and concave by fixing $y = J_n$. As the reciprocal of concave functions in the $R^+$ is convex, $\frac{1}{x\text{log}_2(1+\frac{J_n}{x})}$ is convex. As a result, the left-hand side of the constraints (19a)-(19d) are convex functions, and their epigraphs are convex sets. Therefore, (19a)-(19d) are convex constraints. This completes the proof.
\end{proof}

%1.the left hand side of 19 is convex 他的epigraph就是一个convex set, 所以这个不等式才是一个convex

\section{Proof of THEOREM 1}
\begin{proof}
As the optimization problem (19) is a convex problem, it satisfies the Karush-Kuhn-Tucker (KKT) conditions. By introducing Lagrange multipliers $\lambda^\ast_n \geq 0$, $\mu^\ast_n \geq 0$ and $\gamma^\ast \geq 0$, the partial Lagrange function $L$ is defined as
\begin{equation}
 \begin{aligned}
    &L = T_k^\ast + \sum_{n \in \mathcal{S}_k}\lambda^\ast_n (\frac{z_n}{b^\ast_n\text{log}_2(1+\frac{J_{n}}{b^\ast_n})} + \frac{U_n}{f^\ast_n} - T_k^\ast)+ \sum_{n \in \mathcal{S}_k}\mu^\ast_n  \\
      & (G_n (f^\ast_n)^2 + \frac{H_n}{b^\ast_n\text{log}_2(1+\frac{J_{n}}{b^\ast_n})} - e^{\text{cons}}_n)   + \gamma^\ast(\sum_{n \in \mathcal{S}_k}b^\ast_n - B) 
\end{aligned}   %\tag{$23$}
\end{equation}
where $b^\ast_n$, $f^\ast_n$, and $T_k^\ast$ are optimal solutions. The KKT conditions of the problem (19) can be written as follows:
\begin{align}
&\frac{\partial L }{\partial T_k^\ast} = 1 - \sum_{n \in \mathcal{S}_k}\lambda^\ast_n = 0, \ \forall n \in \mathcal{S}_k  \label{T} \\
&\frac{\partial L }{\partial b^\ast_n} = -\frac{(\lambda^\ast_n z_n + \mu^\ast_n H_n)\Theta_n(b^\ast_n)}{\left(b^\ast_n\text{log}_2(1+\frac{J_{n}}{b^\ast_n})\right)^2} + \gamma^\ast = 0,\ \forall n \in \mathcal{S}_k \label{bn}\\
% \qquad\qquad\qquad\qquad\qquad\qquad\qquad\qquad\qquad,
&\frac{\partial L }{\partial f^\ast_n} = -\frac{F_{n}\lambda^\ast_n}{{f^\ast_n}^2}+2\mu^\ast_n G_n f^\ast_n = 0,\ \forall n \in \mathcal{S}_k \label{fn}\\
&\lambda^\ast_n (\frac{z_n}{b^\ast_n\text{log}_2(1+\frac{J_{n}}{b^\ast_n})} + \frac{U_n}{f^\ast_n} - T_k^\ast) = 0, \ \forall n \in \mathcal{S}_k  \label{lambda}\\
&\mu^\ast_n (G_n (f^\ast_n)^2 + \frac{H_n}{b^\ast_n\text{log}_2(1+\frac{J_{n}}{b^\ast_n})} - e^{\text{cons}}_n) = 0, \ \forall n \in \mathcal{S}_k  \label{mu}\\
&\gamma^\ast(\sum_{n \in \mathcal{S}_k}b^\ast_n - B) = 0. \quad\quad\quad\  & \label{gamma}
\end{align}
where $\Theta_n(b^\ast_n) = \text{log}_2(1+\frac{J_n}{b^\ast_n}) + \frac{J_n}{\text{ln}2}\frac{1}{b^\ast_n + J_n}$. To have (20)-(22), we will prove that all Lagrange multipliers $\lambda^\ast_n$, $\mu^\ast_n$, and $\gamma^\ast$ are larger than 0.

Specifically, if there exists $n^\dag \in \mathcal{S}_k$ that satisfies $\lambda^\ast_{n^\dag} = 0$, we will have $\mu^\ast_{n^\dag} = \frac{F_{n^\dag}\cdot 0}{f^\ast_{n^\dag}}\frac{1}{2G_{n^\dag}f^\ast_{n^\dag}} = 0$ based on~\eqref{fn}. As $\mu^\ast_{n^\dag}$ and $\lambda^\ast_{n^\dag}$ are both 0, we have $\gamma^\ast = (0 \cdot z_{n^\dag} + 0 \cdot H_{n^\dag})\frac{\Theta_n(b^\ast_{n^\dag})}{\left(b^\ast_{n^\dag}\text{log}_2\left(1+J_{n^\dag}/b^\ast_{n^\dag}\right)\right)^2} = 0$ according to~\eqref{bn}. If we plug $\gamma^\ast = 0$ in~\eqref{bn}, we have $(\lambda^\ast_n z_n + \mu^\ast_n H_n)\frac{\Theta_n(b^\ast_n)}{\left(b^\ast_n\text{log}_2(1+\frac{J_{n}}{b^\ast_n})\right)^2} = \gamma^\ast = 0$, $\forall n \in \mathcal{S}_k$. As $\frac{\Theta_n(b^\ast_n)}{\left(b^\ast_n\text{log}_2(1+\frac{J_{n}}{b^\ast_n})\right)^2}$ is always larger than 0, the term $\lambda^\ast_n z_n + \mu^\ast_n H_n$ is equal to 0. As a result, for each $n \in \mathcal{S}_k$, $\lambda^\ast_n$ and $\mu^\ast_n$ are equal to 0 due to the fact that $z_n > 0$ and $H_n > 0$. According to~\eqref{T}, however, there must exist at least one positive $\lambda^\ast_{n}$, which contradicts the conclusion "$\lambda^\ast_n$ = 0, $\forall n \in \mathcal{S}_k$". 

Therefore, $\lambda^\ast_{n} > 0$, $\forall n \in \mathcal{S}_k$. According to~\eqref{bn} and~\eqref{fn}, $\mu^\ast_{n}$ and $\gamma^\ast$ are both positive variables. As a result, we can have (20), (21), and (22) according to~\eqref{lambda},~\eqref{mu}, and~\eqref{gamma}.
\end{proof}

% \textcolor{red}{
% \section{Proof of LEMMA 2}}

% \begin{proof}
% \textcolor{red}{
% If there exists $\widehat{T}_k < T^\ast_k$ that enables the optimization problem (19) to have a feasible solution ($\widehat{\bm{b}}, \widehat{\bm{f}}$), $T^\ast_k$ will be no longer the optimal solution as the objective function of the problem (19) is smaller with $\widehat{T}_k$. It contradicts the fact that $T^\ast_k$ is the optimal solution of the problem (19).}
% \textcolor{red}{
% When $T_k > T^\ast_k$, there always exists a feasible solution ($\bm{b}, \bm{f}$) to satisfy the constraints (19a)\textasciitilde (19d).}
% \end{proof}

\section{Proof of LEMMA 2}
\begin{proof}
The first and second order derivatives of $Q_n(x)$ are derived as
\begin{align}
    Q_n'(x) &= \log_2(1+\frac{J_n}{x}) - \frac{J_n}{\ln2 (x+J_n)} = h_n(x) \label{Qd}\\
    Q_n''(x) &= h_n'(x) = \frac{-J_n^2}{x(x+J_n)^2\ln2} < 0 \label{Qdd}
\end{align}
According to~\eqref{Qdd}, $Q_n'(x)$ is a decreasing function. As $\text{lim}_{x \rightarrow +\infty}Q_n'(x) = 0$, we have $Q_n'(x) > 0$. Therefore, $Q_n(x)$ is an increasing function. By leveraging the inequality $\frac{1}{\ln2}(x-1) > \text{log}_2(x)$, we have $x\log_2(1+\frac{J_n}{x}) < x\frac{1}{\ln 2}\frac{J_n}{x} = \frac{J_n}{\ln2}$.
\end{proof}

\section{Proof of LEMMA 3}
\begin{proof}

Let $X_n = \frac{H_n T_k}{z_n G_n}-\frac{e^{\text{cons}}_n}{G_n}$, $Y_n = \frac{H_n U_n}{z_n G_n} > 0$ . Therefore, $M(f_n)$ can be presented as $f_n^3 + X_nf_n - Y_n$. The first order derivative of $M(f_n)$ is:
\begin{align}
     M'(f_n) = 3f_n^2 + X_n. 
\end{align}

If $X_n \geq 0$, $M'(f_n) > 0$ in $(0, +\infty)$, and $M(f_n)$ is a monotonically increasing function. As $M(0) = -Y_n < 0$, $M(f_n)$ will have only one root in $(0, +\infty)$.

If $X_n < 0$, $M'(f_n) < 0$ in $(0, \sqrt{-\frac{X_n}{3}})$ while $M'(f_n) > 0$ in $(\sqrt{-\frac{X_n}{3}}, +\infty)$. Therefore, $M(f_n)$ is a monotonically decreasing function in $(0, \sqrt{-\frac{X_n}{3}})$ and a monotonically increasing function in $(\sqrt{-\frac{X_n}{3}}, +\infty)$. Since $M(0) = -Y_n < 0$, there will exist only one root $f_n^\dag \in (\sqrt{-\frac{X_n}{3}}, +\infty)$. $M(f_n) < 0$ in $(0, f_n^\dag)$ while $M(f_n) > 0$ in $(f_n^\dag, +\infty)$.

\end{proof}

\section{Optimize the transmit power $p$}

The problem (19) considers the transmit power $p_n$ as a constant. To further improve energy efficiency of FL, we propose a preliminary method to optimize the transmit power given an upper bound $p^\text{max}$ and a lower bound $p^\text{min}$. According to Equation (19a) and (19b), a larger $p_n$ reduces the time delay $T_k$ and increases the energy consumption $E_k$. Besides, increasing $p_n$ will reduce $f_n$ due to the energy constraint. It is crucial to balance $f_n$ and $p_n$ for solving the problem (19). To this end, we design a binary search algorithm to obtain the optimal transmit power $p^\ast$ as shown in Algorithm 6. Note that Algorithm 6 assumes all the devices have the same transmit power (i.e., $p_1 = ... = p_n = p$).

\begin{algorithm}[t]
\textsl{}\setstretch{1}
\caption*{\textbf{Algorithm 6} Optimal Transmit Power of FL}
\begin{algorithmic}[1] 
\Require{$p^\text{max}$, $p^\text{min}$}
\Ensure{$p^\ast$}
\State $p^\text{up} = p^\text{max}$, $p^\text{low} = p^\text{min}$
\State $epoch = 0$, $\mathcal{T} \leftarrow \{\}$, $p = p^\text{low}$
\While {1 - $\frac{p^\text{low}}{p^\text{up}} > \varepsilon_3$}
    \State Obtain $T_k$ using Algorithm 5
    \If{$epoch > 0$}
        \If{$T_k \leq \min{\mathcal{T}}$}
            \State $p^\text{low} = p$
        \Else
            \State $p^\text{up} = p$
        \EndIf
    \EndIf
    \State $p = \frac{p^\text{up}+p^\text{low}}{2}$
    \State $\mathcal{T} \leftarrow T_k$
    \State $epoch = epoch + 1$
\EndWhile
\State \Return $p^\ast$
\end{algorithmic} 
\label{optimalP}
\end{algorithm}

\renewcommand{\thefigure}{14}
\begin{figure}[t]
  \centering
  \includegraphics[width=7.5cm]{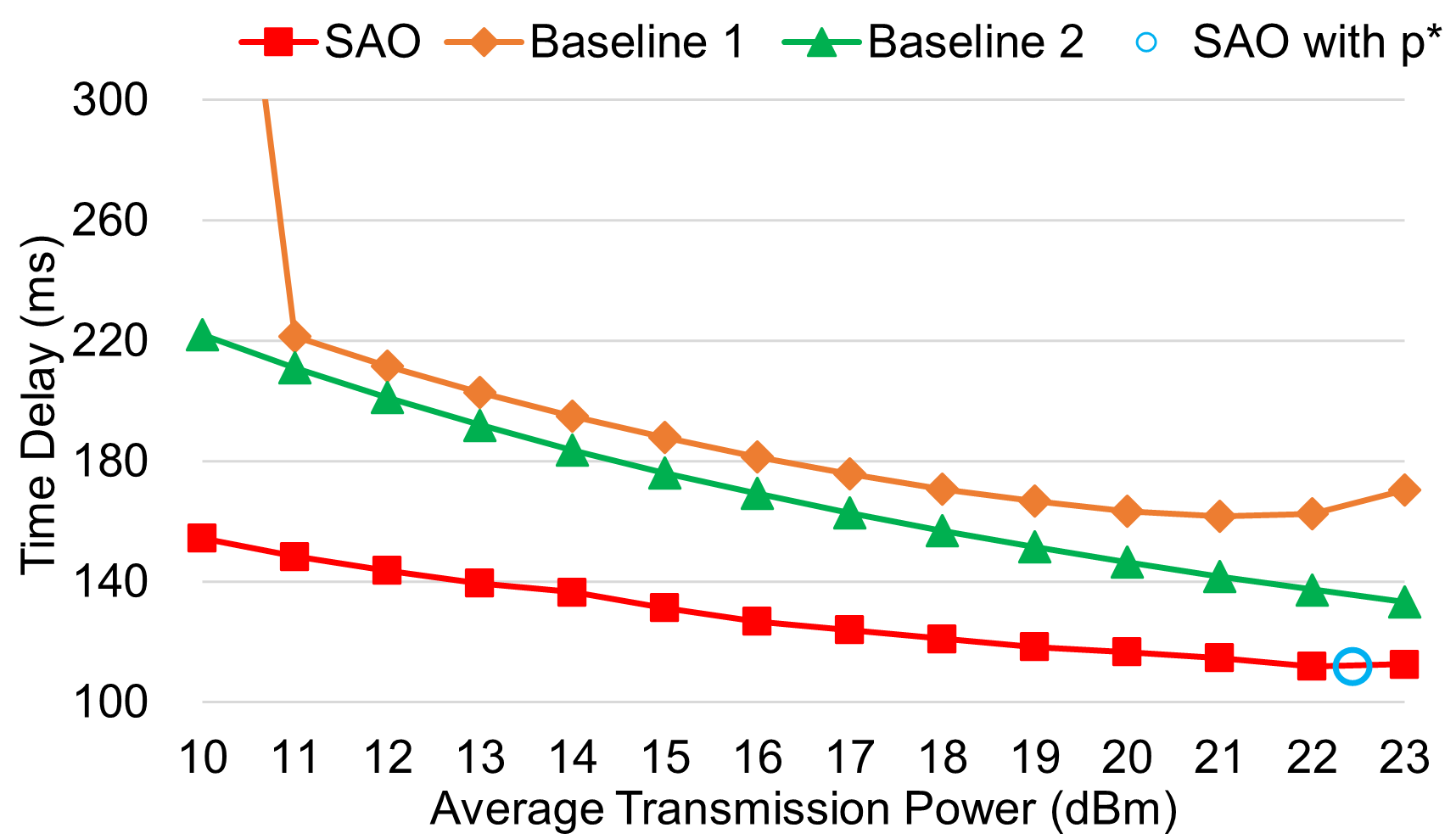}
\caption{Time delay versus average transmission power. $p^\text{min} = 10$ dBm, and $p^\text{max} = 23$ dBm.}
\label{optimal_p}
\end{figure}

As shown in Fig.~\ref{optimal_p}, without Algorithm 6, SAO achieves the lowest time delay (111.84 ms) when $p = 22$ dBm. The optimal transmit power $p^\ast$ derived from Algorithm 6 is 22.45 dBm, where the time delay is 111.49 ms. The result demonstrates that Algorithm 6 enables FL to find the optimal transmit power.

The complexity of SAO with $p^\ast$ is $\mathcal{O}\left( (S_k)^2\log_2{\left( \frac{1}{\varepsilon_0}\right)}\log_2{\left( \frac{1}{\varepsilon_1}\right)}\log_2{\left( \frac{1}{\varepsilon_2}\right)}\log_2{\left( \frac{1}{\varepsilon_3}\right)}  \right)$, where $\varepsilon_3$ is the accuracy of searching for $p^\ast$.

\bibliographystyle{IEEEtran}
\bibliography{sample-base}

\begin{thebibliography}{10}
\providecommand{\url}[1]{#1}
\csname url@samestyle\endcsname
\providecommand{\newblock}{\relax}
\providecommand{\bibinfo}[2]{#2}
\providecommand{\BIBentrySTDinterwordspacing}{\spaceskip=0pt\relax}
\providecommand{\BIBentryALTinterwordstretchfactor}{4}
\providecommand{\BIBentryALTinterwordspacing}{\spaceskip=\fontdimen2\font plus
\BIBentryALTinterwordstretchfactor\fontdimen3\font minus
  \fontdimen4\font\relax}
\providecommand{\BIBforeignlanguage}[2]{{%
\expandafter\ifx\csname l@#1\endcsname\relax
\typeout{** WARNING: IEEEtran.bst: No hyphenation pattern has been}%
\typeout{** loaded for the language `#1'. Using the pattern for}%
\typeout{** the default language instead.}%
\else
\language=\csname l@#1\endcsname
\fi
#2}}
\providecommand{\BIBdecl}{\relax}
\BIBdecl

\bibitem{Drainakis20}
G.~{Drainakis}, K.~V. {Katsaros}, P.~{Pantazopoulos}, V.~{Sourlas}, and
  A.~{Amditis}, ``Federated vs. centralized machine learning under
  privacy-elastic users: A comparative analysis,'' in \emph{2020 IEEE 19th
  International Symposium on Network Computing and Applications (NCA)}, 2020,
  pp. 1--8.

\bibitem{Liu20}
Y.~Liu, J.~J.~Q. Yu, J.~Kang, D.~Niyato, and S.~Zhang, ``Privacy-preserving
  traffic flow prediction: A federated learning approach,'' \emph{IEEE Internet
  of Things Journal}, vol.~7, no.~8, pp. 7751--7763, 2020.

\bibitem{Helin20}
H.~Yang, J.~Zhao, Z.~Xiong, K.-Y. Lam, S.~Sun, and L.~Xiao,
  ``Privacy-preserving federated learning for uav-enabled networks:
  Learning-based joint scheduling and resource management,'' 2020.

\bibitem{Li21}
X.~Li, L.~Cheng, C.~Sun, K.-Y. Lam, X.~Wang, and F.~Li,
  ``Federated-learning-empowered collaborative data sharing for vehicular edge
  networks,'' \emph{IEEE Network}, vol.~35, no.~3, pp. 116--124, 2021.

\bibitem{Zhao21}
Y.~Zhao, J.~Zhao, M.~Yang, T.~Wang, N.~Wang, L.~Lyu, D.~Niyato, and K.-Y. Lam,
  ``Local differential privacy-based federated learning for internet of
  things,'' \emph{IEEE Internet of Things Journal}, vol.~8, no.~11, pp.
  8836--8853, 2021.

\bibitem{Viraaji21}
V.~Mothukuri, R.~M. Parizi, S.~Pouriyeh, Y.~Huang, A.~Dehghantanha, and
  G.~Srivastava, ``A survey on security and privacy of federated learning,''
  \emph{Future Generation Computer Systems}, vol. 115, pp. 619--640, 2021.

\bibitem{Li20}
T.~{Li}, A.~K. {Sahu}, A.~{Talwalkar}, and V.~{Smith}, ``Federated learning:
  Challenges, methods, and future directions,'' \emph{IEEE Signal Processing
  Magazine}, vol.~37, no.~3, pp. 50--60, 2020.

\bibitem{Aledhari20}
M.~{Aledhari}, R.~{Razzak}, R.~M. {Parizi}, and F.~{Saeed}, ``Federated
  learning: A survey on enabling technologies, protocols, and applications,''
  \emph{IEEE Access}, vol.~8, pp. 140\,699--140\,725, 2020.

\bibitem{Wei20}
W.~Y.~B. Lim, N.~C. Luong, D.~T. Hoang, Y.~Jiao, Y.-C. Liang, Q.~Yang,
  D.~Niyato, and C.~Miao, ``Federated learning in mobile edge networks: A
  comprehensive survey,'' \emph{IEEE Communications Surveys Tutorials},
  vol.~22, no.~3, pp. 2031--2063, 2020.

\bibitem{Brik20}
B.~Brik, A.~Ksentini, and M.~Bouaziz, ``Federated learning for uavs-enabled
  wireless networks: Use cases, challenges, and open problems,'' \emph{IEEE
  Access}, vol.~8, pp. 53\,841--53\,849, 2020.

\bibitem{Lim20}
W.~Y.~B. Lim, S.~Garg, Z.~Xiong, D.~Niyato, C.~Leung, C.~Miao, and M.~Guizani,
  ``Dynamic contract design for federated learning in smart healthcare
  applications,'' \emph{IEEE Internet of Things Journal}, pp. 1--1, 2020.

\bibitem{Kim20}
H.~Kim, J.~Park, M.~Bennis, and S.-L. Kim, ``Blockchained on-device federated
  learning,'' \emph{IEEE Communications Letters}, vol.~24, no.~6, pp.
  1279--1283, 2020.

\bibitem{Theodora18}
T.~S. Brisimi, R.~Chen, T.~Mela, A.~Olshevsky, I.~C. Paschalidis, and W.~Shi,
  ``Federated learning of predictive models from federated electronic health
  records,'' \emph{International Journal of Medical Informatics}, vol. 112, pp.
  59--67, 2018.

\bibitem{Meng20}
M.~Hao, H.~Li, X.~Luo, G.~Xu, H.~Yang, and S.~Liu, ``Efficient and
  privacy-enhanced federated learning for industrial artificial intelligence,''
  \emph{IEEE Transactions on Industrial Informatics}, vol.~16, no.~10, pp.
  6532--6542, 2020.

\bibitem{McMahan20}
\BIBentryALTinterwordspacing
D.~R. B.~McMahan, ``Federated learning: Collaborative machine learning without
  centralized training data,'' \emph{Google AI Blog}, 2020. [Online].
  Available:
  \url{https://ai.googleblog.com/2017/04/federated-learning-collaborative.html}
\BIBentrySTDinterwordspacing

\bibitem{Hao20}
\BIBentryALTinterwordspacing
K.~Hao, ``How apple personalizes siri without hoovering up your data,''
  \emph{Technology Review}, 2020. [Online]. Available:
  \url{https://www.technologyreview.com/2019/12/11/131629/apple-ai-personalizes-siri-federated-learning/}
\BIBentrySTDinterwordspacing

\bibitem{Khan21}
L.~U. Khan, W.~Saad, Z.~Han, E.~Hossain, and C.~S. Hong, ``Federated learning
  for internet of things: Recent advances, taxonomy, and open challenges,''
  \emph{IEEE Communications Surveys Tutorials}, vol.~23, no.~3, pp. 1759--1799,
  2021.

\bibitem{Chen20}
M.~Chen, H.~V. Poor, W.~Saad, and S.~Cui, ``Wireless communications for
  collaborative federated learning,'' \emph{IEEE Communications Magazine},
  vol.~58, no.~12, pp. 48--54, 2020.

\bibitem{Howard20}
H.~H. Yang, Z.~Liu, T.~Q.~S. Quek, and H.~V. Poor, ``Scheduling policies for
  federated learning in wireless networks,'' \emph{IEEE Transactions on
  Communications}, vol.~68, no.~1, pp. 317--333, 2020.

\bibitem{Amiri20}
M.~M. {Amiri}, D.~{Gündüz}, S.~R. {Kulkarni}, and H.~{Vincent Poor}, ``Update
  aware device scheduling for federated learning at the wireless edge,'' in
  \emph{2020 IEEE International Symposium on Information Theory (ISIT)}, 2020,
  pp. 2598--2603.

\bibitem{Nilsson18}
A.~Nilsson, S.~Smith, G.~Ulm, E.~Gustavsson, and M.~Jirstrand, ``A performance
  evaluation of federated learning algorithms,'' in \emph{Proceedings of the
  Second Workshop on Distributed Infrastructures for Deep Learning}, ser. DIDL
  '18.\hskip 1em plus 0.5em minus 0.4em\relax New York, NY, USA: Association
  for Computing Machinery, 2018, p. 1–8.

\bibitem{Niknam20}
S.~Niknam, H.~S. Dhillon, and J.~H. Reed, ``Federated learning for wireless
  communications: Motivation, opportunities, and challenges,'' \emph{IEEE
  Communications Magazine}, vol.~58, no.~6, pp. 46--51, 2020.

\bibitem{Wang20}
H.~Wang, Z.~Kaplan, D.~Niu, and B.~Li, ``Optimizing federated learning on
  non-iid data with reinforcement learning,'' in \emph{IEEE INFOCOM 2020 - IEEE
  Conference on Computer Communications}, 2020, pp. 1698--1707.

\bibitem{Zhu19}
G.~{Zhu}, Y.~{Wang}, and K.~{Huang}, ``Broadband analog aggregation for
  low-latency federated edge learning,'' \emph{IEEE Transactions on Wireless
  Communications}, vol.~19, no.~1, pp. 491--506, 2020.

\bibitem{Yang20}
Z.~Yang, M.~Chen, W.~Saad, C.~S. Hong, and M.~Shikh-Bahaei, ``Delay
  minimization for federated learning over wireless communication networks,''
  in \emph{Proc. Int. Conf. Machine Learning Workshop}, July 2020.

\bibitem{Luo20}
S.~{Luo}, X.~{Chen}, Q.~{Wu}, Z.~{Zhou}, and S.~{Yu}, ``Hfel: Joint edge
  association and resource allocation for cost-efficient hierarchical federated
  edge learning,'' \emph{IEEE Transactions on Wireless Communications},
  vol.~19, no.~10, pp. 6535--6548, 2020.

\bibitem{Dinh21}
C.~T. Dinh, N.~H. Tran, M.~N.~H. Nguyen, C.~S. Hong, W.~Bao, A.~Y. Zomaya, and
  V.~Gramoli, ``Federated learning over wireless networks: Convergence analysis
  and resource allocation,'' \emph{IEEE/ACM Transactions on Networking},
  vol.~29, no.~1, pp. 398--409, 2021.

\bibitem{Vu20}
T.~T. {Vu}, D.~T. {Ngo}, N.~H. {Tran}, H.~Q. {Ngo}, M.~N. {Dao}, and R.~H.
  {Middleton}, ``Cell-free massive mimo for wireless federated learning,''
  \emph{IEEE Transactions on Wireless Communications}, vol.~19, no.~10, pp.
  6377--6392, 2020.

\bibitem{Wang19}
S.~{Wang}, T.~{Tuor}, T.~{Salonidis}, K.~K. {Leung}, C.~{Makaya}, T.~{He}, and
  K.~{Chan}, ``Adaptive federated learning in resource constrained edge
  computing systems,'' \emph{IEEE Journal on Selected Areas in Communications},
  vol.~37, no.~6, pp. 1205--1221, 2019.

\bibitem{Yang21}
Z.~{Yang}, M.~{Chen}, W.~{Saad}, C.~S. {Hong}, and M.~{Shikh-Bahaei}, ``Energy
  efficient federated learning over wireless communication networks,''
  \emph{IEEE Transactions on Wireless Communications}, vol.~20, no.~3, pp.
  1935--1949, 2021.

\bibitem{McMahan17}
H.~McMahan, E.~Moore, D.~Ramage, S.~Hampson, and B.~A. y~Arcas,
  ``Communication-efficient learning of deep networks from decentralized
  data,'' in \emph{AISTATS}, 2017.

\bibitem{Abdulrahman21}
S.~Abdulrahman, H.~Tout, A.~Mourad, and C.~Talhi, ``Fedmccs: Multicriteria
  client selection model for optimal iot federated learning,'' \emph{IEEE
  Internet of Things Journal}, vol.~8, no.~6, pp. 4723--4735, 2021.

\bibitem{Nishio19}
T.~{Nishio} and R.~{Yonetani}, ``Client selection for federated learning with
  heterogeneous resources in mobile edge,'' in \emph{ICC 2019 - 2019 IEEE
  International Conference on Communications (ICC)}, 2019, pp. 1--7.

\bibitem{shi20}
W.~{Shi}, S.~{Zhou}, and Z.~{Niu}, ``Device scheduling with fast convergence
  for wireless federated learning,'' in \emph{ICC 2020 - 2020 IEEE
  International Conference on Communications (ICC)}, 2020, pp. 1--6.

\bibitem{Hado15}
H.~van Hasselt, A.~Guez, and D.~Silver, ``Deep reinforcement learning with
  double q-learning,'' 2015.

\bibitem{Chen21}
M.~Chen, H.~V. Poor, W.~Saad, and S.~Cui, ``Convergence time optimization for
  federated learning over wireless networks,'' \emph{IEEE Transactions on
  Wireless Communications}, vol.~20, no.~4, pp. 2457--2471, 2021.

\bibitem{Chen2021}
M.~Chen, Z.~Yang, W.~Saad, C.~Yin, H.~V. Poor, and S.~Cui, ``A joint learning
  and communications framework for federated learning over wireless networks,''
  \emph{IEEE Transactions on Wireless Communications}, vol.~20, no.~1, pp.
  269--283, 2021.

\bibitem{Shi21}
W.~Shi, S.~Zhou, Z.~Niu, M.~Jiang, and L.~Geng, ``Joint device scheduling and
  resource allocation for latency constrained wireless federated learning,''
  \emph{IEEE Transactions on Wireless Communications}, vol.~20, no.~1, pp.
  453--467, 2021.

\bibitem{Zeng20}
Q.~{Zeng}, Y.~{Du}, K.~{Huang}, and K.~K. {Leung}, ``Energy-efficient radio
  resource allocation for federated edge learning,'' in \emph{2020 IEEE
  International Conference on Communications Workshops (ICC Workshops)}, 2020,
  pp. 1--6.

\bibitem{Tengchan20}
T.~Zeng, O.~Semiari, M.~Mozaffari, M.~Chen, W.~Saad, and M.~Bennis, ``Federated
  learning in the sky: Joint power allocation and scheduling with uav swarms,''
  in \emph{ICC 2020 - 2020 IEEE International Conference on Communications
  (ICC)}, 2020, pp. 1--6.

\bibitem{Xu21}
J.~{Xu} and H.~{Wang}, ``Client selection and bandwidth allocation in wireless
  federated learning networks: A long-term perspective,'' \emph{IEEE
  Transactions on Wireless Communications}, vol.~20, no.~2, pp. 1188--1200,
  2021.

\bibitem{Ren20}
J.~Ren, Y.~He, D.~Wen, G.~Yu, K.~Huang, and D.~Guo, ``Scheduling for cellular
  federated edge learning with importance and channel awareness,'' \emph{IEEE
  Transactions on Wireless Communications}, vol.~PP, pp. 1--1, 08 2020.

\bibitem{Bottou12}
L.~Bottou, ``Stochastic gradient descent tricks,'' vol. 7700, pp. 430--445,
  January 2012.

\bibitem{Krizhevsky09}
A.~Krizhevsky, ``Learning multiple layers of features from tiny images,'' 2009.

\bibitem{Zhu21}
\BIBentryALTinterwordspacing
H.~Zhu, J.~Xu, S.~Liu, and Y.~Jin, ``Federated learning on non-iid data: A
  survey,'' \emph{Neurocomputing}, vol. 465, pp. 371--390, 2021. [Online].
  Available:
  \url{https://www.sciencedirect.com/science/article/pii/S0925231221013254}
\BIBentrySTDinterwordspacing

\bibitem{MacQueen67}
J.~B. MacQueen, ``Some methods for classification and analysis of multivariate
  observations,'' in \emph{Proc. of the fifth Berkeley Symposium on
  Mathematical Statistics and Probability}, L.~M.~L. Cam and J.~Neyman, Eds.,
  vol.~1.\hskip 1em plus 0.5em minus 0.4em\relax University of California
  Press, 1967, pp. 281--297.

\bibitem{Yang2021}
Z.~Yang, M.~Chen, W.~Saad, W.~Xu, M.~Shikh-Bahaei, H.~V. Poor, and S.~Cui,
  ``Energy-efficient wireless communications with distributed reconfigurable
  intelligent surfaces,'' \emph{IEEE Transactions on Wireless Communications},
  pp. 1--1, 2021.

\bibitem{Cun90}
{Cun, Yann Le and Denker, John S. and Solla, Sara A.}, \emph{Optimal Brain
  Damage}.\hskip 1em plus 0.5em minus 0.4em\relax San Francisco, CA, USA:
  Morgan Kaufmann Publishers Inc., 1990, p. 598–605.

\bibitem{xiao2017}
H.~Xiao, K.~Rasul, and R.~Vollgraf. (2017) Fashion-mnist: a novel image dataset
  for benchmarking machine learning algorithms.

\bibitem{Lawrence85}
\BIBentryALTinterwordspacing
L.~Hubert and P.~Arabie, ``{Comparing partitions},'' \emph{Journal of
  Classification}, vol.~2, no.~1, pp. 193--218, December 1985. [Online].
  Available: \url{https://ideas.repec.org/a/spr/jclass/v2y1985i1p193-218.html}
\BIBentrySTDinterwordspacing

\bibitem{Nguyen10}
\BIBentryALTinterwordspacing
N.~X. Vinh, J.~Epps, and J.~Bailey, ``Information theoretic measures for
  clusterings comparison: Variants, properties, normalization and correction
  for chance,'' \emph{Journal of Machine Learning Research}, vol.~11, no.~95,
  pp. 2837--2854, 2010. [Online]. Available:
  \url{http://jmlr.org/papers/v11/vinh10a.html}
\BIBentrySTDinterwordspacing

\end{thebibliography}

\begin{IEEEbiography}[{\includegraphics[width=1in,clip,keepaspectratio]{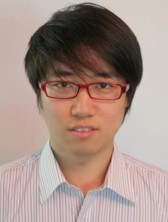}}]{Tinghao Zhang} is currently a Ph.D student in School of Computer Science and Engineering at Nanyang Technological University, Singapore. His main research interests include federated learning, edge computing, and wireless communication. He obtained his B.S. degree in the School of Electrical Engineering and Automation and M.S. degree in the School of Instrumentation Science and Engineering from Harbin Institute of Technology in 2014 and 2018, respectively.
\end{IEEEbiography}
\begin{IEEEbiography}[{\includegraphics[width=1in,clip,keepaspectratio]{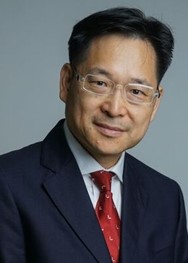}}]{Kwok-Yan Lam}(Senior Member, IEEE) received his B.Sc. degree ($1^\text{st}$ Class Hons.) from University of London, in 1987, and Ph.D. degree from University of Cambridge, in 1990. He is the Associate Vice President (Strategy and Partnerships) and Professor in the School of Computer Science and Engineering at the Nanyang Technological University, Singapore. He is currently also the Executive Director of the National Centre for Research in Digital Trust, and Director of the Strategic Centre for Research in Privacy-Preserving Technologies and Systems (SCRiPTS). From August 2020, he is on part-time secondment to the INTERPOL as a Consultant at Cyber and New Technology Innovation. Prior to joining NTU, he has been a Professor of the Tsinghua University, PR China (2002–2010) and a faculty member of the National University of Singapore and the University of London since 1990. He was a Visiting Scientist at the Isaac Newton Institute, Cambridge University, and a Visiting Professor at the European Institute for Systems Security. In 1998, he received the Singapore Foundation Award from the Japanese Chamber of Commerce and Industry in recognition of his research and development achievement in information security in Singapore. He is the recipient of the Singapore Cybersecurity Hall of Fame Award in 2022. His research interests include Distributed Systems, Intelligent Systems, IoT Security, Distributed Protocols for Blockchain, Homeland Security and Cybersecurity.
\end{IEEEbiography}
\begin{IEEEbiography}[{\includegraphics[width=1in,clip,keepaspectratio]{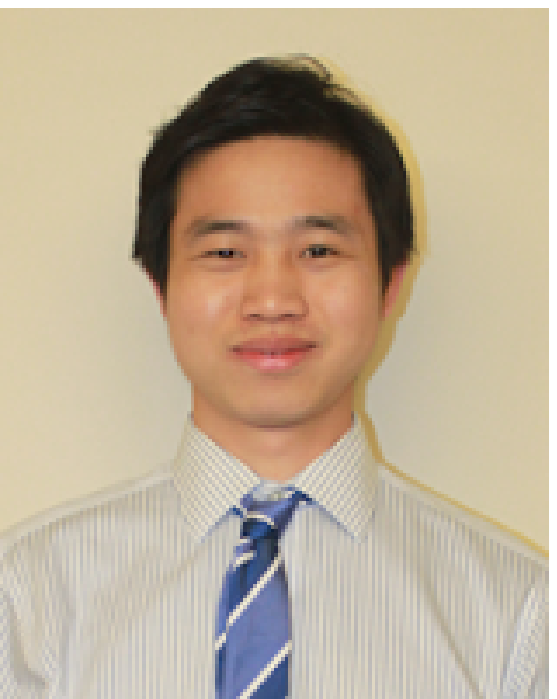}}]{Jun Zhao}(Senior Member, IEEE) (S'10-M'15) is currently an Assistant Professor in the School of Computer Science and Engineering (SCSE) at Nanyang Technological University (NTU) in Singapore. He received a PhD degree in May 2015 in Electrical and Computer Engineering from Carnegie Mellon University (CMU) in the USA (advisors: Virgil Gligor, Osman Yagan; collaborator: Adrian Perrig), affiliating with CMU's renowned CyLab Security \& Privacy Institute, and a bachelor's degree in July 2010 from Shanghai Jiao Tong University in China. Before joining NTU first as a postdoc with Xiaokui Xiao and then as a faculty member, he was a postdoc at Arizona State University as an Arizona Computing PostDoc Best Practices Fellow (advisors: Junshan Zhang, Vincent Poor).
\end{IEEEbiography}
\begin{IEEEbiography}[{\includegraphics[width=1in,clip,keepaspectratio]{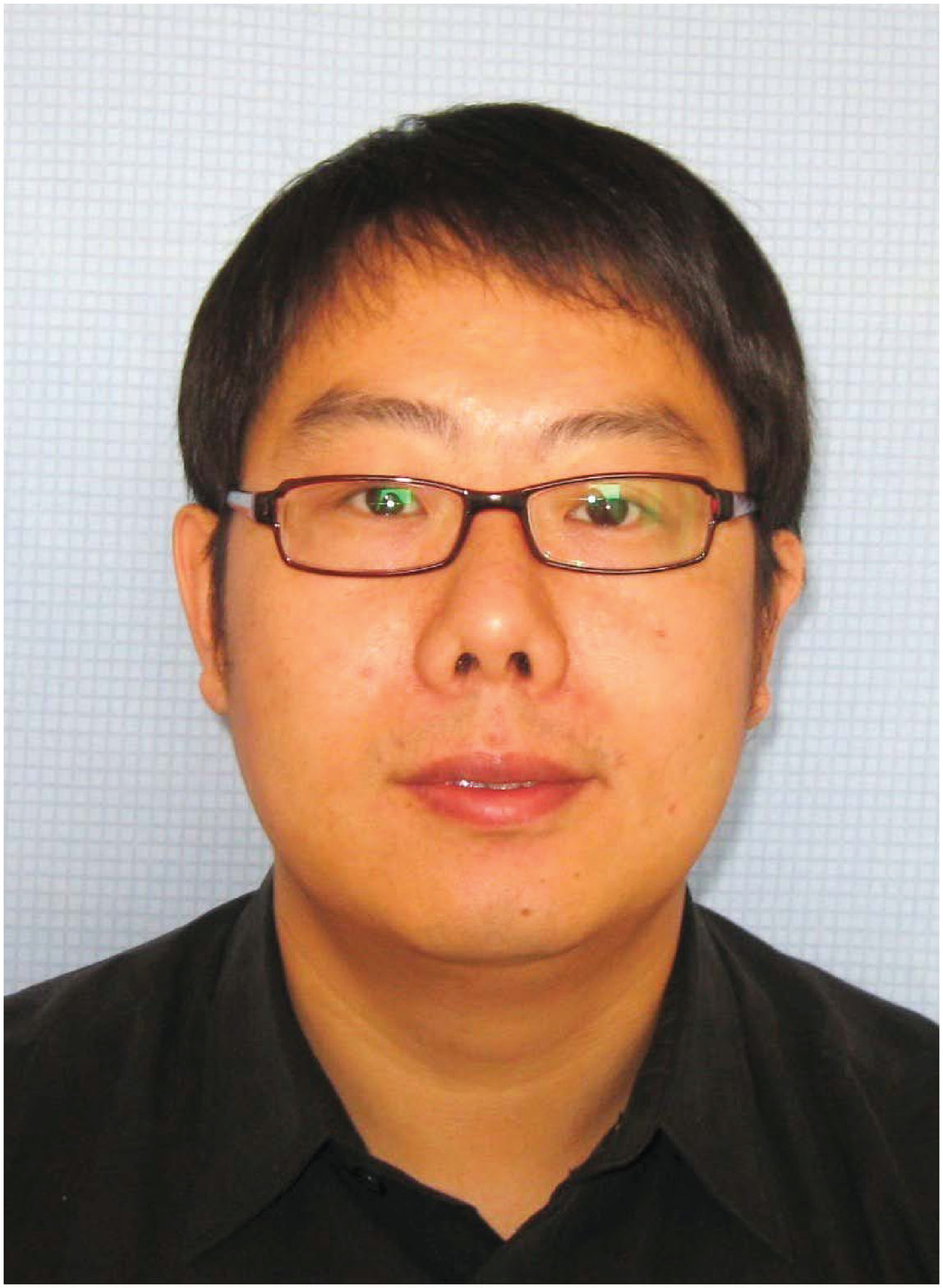}}]{Feng Li}
received his Ph.D degree from the Harbin Institute of Technology, Harbin, China in 2013. He is a full Professor at School of Information and Electronic Engineering, Zhejiang Gongshang University. F. Li is also at School of Computer Science and Engineering, Nanyang Technological University. His research interests include cognitive radio networks, sensor networks and satellite systems.
\end{IEEEbiography}
\begin{IEEEbiography}[{\includegraphics[width=1in,clip,keepaspectratio]{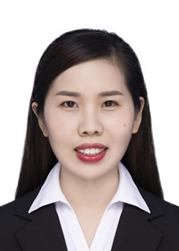}}]{HuiMei Han} obtained her Ph.D. degree in telecommunication engineering from Xidian University, Xi'an, P. R. China in 2019. From 2017 to 2018, she was with Florida Atlantic University, USA, as an exchange Ph.D. student.  She has been with College of Information Engineering, Zhejiang University of Technology since 2019, and with School of Computer Science and Engineering, Nanyang Technological University, Singapore, as a Research Fellow since 2021. Her current research interests include random access schemes for massive MIMO systems, machine-to-machine communications, machine learning, intelligent reflecting surface, and federated learning.
\end{IEEEbiography}
\begin{IEEEbiography}[{\includegraphics[width=1in,clip,keepaspectratio]{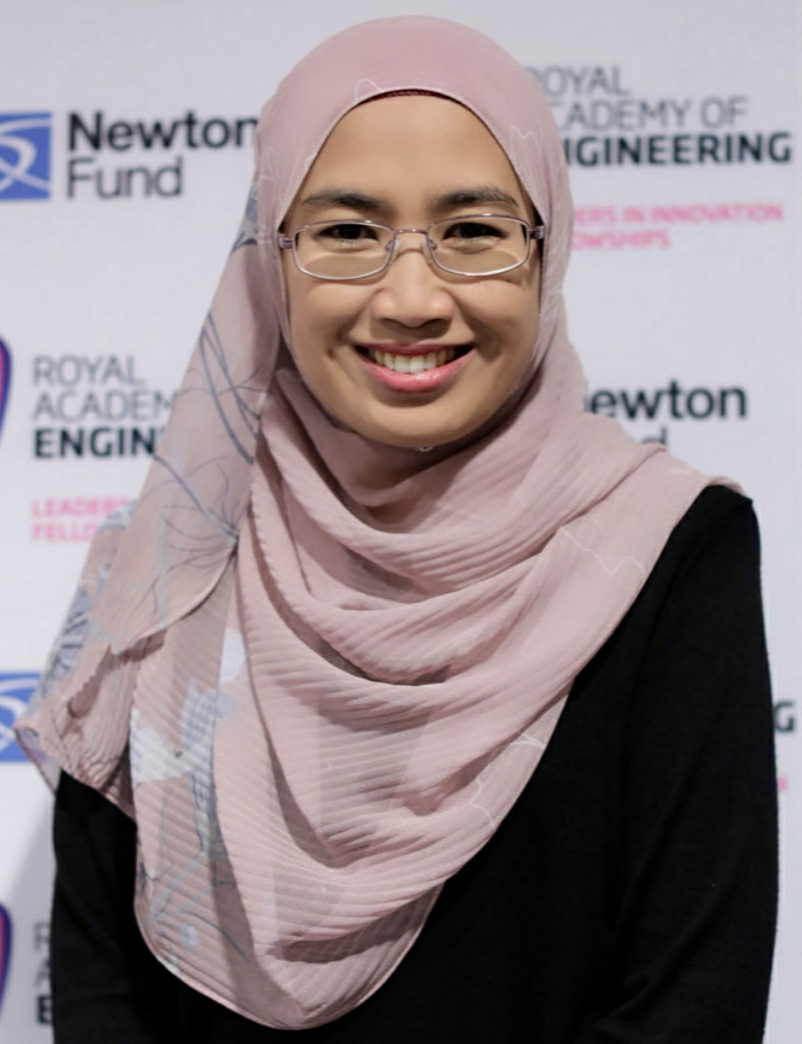}}]{Norziana Jamil}
received her PhD in Security in Computing in 2013. She is now an Associate Professor at the University Tenaga Nasional, Malaysia. Her area of research specialization and interest includes Cryptography, security for Cyber-Physical Systems, security analytics and intelligent system. She is an alumni of Leadership in Innovation Fellowship by UK Royal Academy of Engineering, a Project Leader and consultant of various cryptography and cyber security related research and consultancy projects, has been actively involving in advisory for cryptography and cyber security projects, and works with several international prominent researchers and professors.
\end{IEEEbiography}


% % insert where needed to balance the two columns on the last page with
% %\newpage


\end{document}